\documentclass[11pt]{llncs}

\usepackage{latexsym}
\usepackage[cm]{fullpage}

\usepackage{amssymb,amsmath,amsfonts,bm}
\usepackage{subfigure} 
\usepackage{epic} 
\usepackage{eepic} 
\usepackage{enumerate}
\usepackage{vinayak}
\usepackage{graphicx}
\usepackage{verbatim}
\usepackage[mathscr]{eucal}
\usepackage{xspace}
\usepackage{stmaryrd}
\usepackage[linesnumbered]{algorithm2e}

\title{Synthesis of Memory-Efficient  Real-Time  Controllers for \\Safety Objectives
\thanks{%
This work has been financially supported in part by the
 European Community's Seventh Framework Programme
 via project Control for coordination of distributed systems
 (C4C; Grant Agreement number INFSO-ICT-223844); and
 by Austrian FWF NFN ARiSE funding.}
}
\author{Krishnendu Chatterjee$^1$    
\and Vinayak S.~Prabhu$^2$ }
\institute{
  $^1$ Institute of Science and Technology (IST) Austria\\
  $^2$ University of Porto\\
  {\tt krish.chat@ist.ac.at, vinayak@eecs.berkeley.edu}
}
\date{}

\begin{document}
\maketitle

\thispagestyle{empty}

\begin{abstract}
We study synthesis of controllers for real-time systems, where the 
objective is to stay in a given safe set.
The problem is solved by obtaining winning strategies in concurrent two-player 
\emph{timed automaton games} with safety objectives.
To prevent a player from winning by blocking time, we restrict each player 
to strategies that ensure that the player cannot be responsible for causing 
a zeno run.
We construct winning strategies for the controller which require access 
only to (1)~the system clocks (thus,  controllers which require their own 
internal infinitely precise clocks are not necessary), and 
(2)~a linear (in the number of clocks) number of  memory bits.
Precisely, we show that a memory of size $\big(3\cdot|C|+1 + \lg(|C|+1)\big)$ 
bits suffices for winning controller
strategies for safety objectives, where $C$ is the set of clocks of the 
timed automaton game, significantly improving the previous known 
exponential bound.
We also settle the open question of whether \emph{region} strategies 
for controllers require memory for safety objectives by showing 
with an example that region strategies do require memory for safety objectives.
\end{abstract}

\section{Introduction}

Synthesizing controllers to ensure that a plant stays in a safe set is an
important problem in the area of systems control.
We study the synthesis of \emph{timed} controllers in the present paper.
Our formalism is based on  timed automata~\cite{AlurD94}, which are 
models of real-time systems in which
states consist of discrete locations and values for real-time clocks.
The transitions between locations are dependent on the clock values.
The real-time controller synthesis problem is modeled using 
\emph{timed automaton games}, which are played by two players on 
timed automata, where player~1 is the  ``controller'' and player~2 
the ``plant''.
Obtaining  winning strategies for player~1 in such games 
corresponds to the construction of  controllers for real-time
systems with desired objectives.

The issue of \emph{time divergence} is crucial in timed games, as a 
naive control strategy might simply block time, leading to 
``zeno'' runs.
The following approaches have been proposed to avoid such invalid zeno
solutions: 
(1)~discretize time so that 
players can only take 
transitions at integer multiples of some fixed time period, e.g. 
in~\cite{HenKop99};
(2)~put syntactic restrictions on the timed game structure so that
zeno runs are not possible (the syntactic restriction is usually presented 
as the 
\emph{strong non-zenoness} assumption where the obtained 
controller synthesis algorithms
are guaranteed to work correctly only on
timed automaton games where every cycle is such that in it 
some clock is reset to 0 and is also greater than an integer value at some
point, e.g. in~\cite{AsarinM99,BouyerBL04,maler98controller});
(3)~require player~1 to ensure time divergence (e.g. by only taking transitions
if player~2 can never take transitions in the future from the current location,
as in~\cite{DSouzaM02,BouyerDMP03});
(4)~give the controller access to  an extra (infinitely precise) clock which
measure global time
and require that 
player~1 wins if either its moves are  chosen only finitely often, or if the 
ticks of this extra clock are seen infinitely often while satisfying the desired
objective, e.g, in~\cite{AFHM+03,AluHen97}.

The above  approaches are not optimal in many cases and below we point out 
some drawbacks.
Discretizing the system blows up the state space; and might not be faithful
to the real-time semantics.
Putting syntactic restrictions is troublesome as it can lead to disallowing
certain system models.
For example, consider the timed automaton game $\A$ in 
Figure~\ref{figure:example-syntactic}.
\begin{figure}[t]
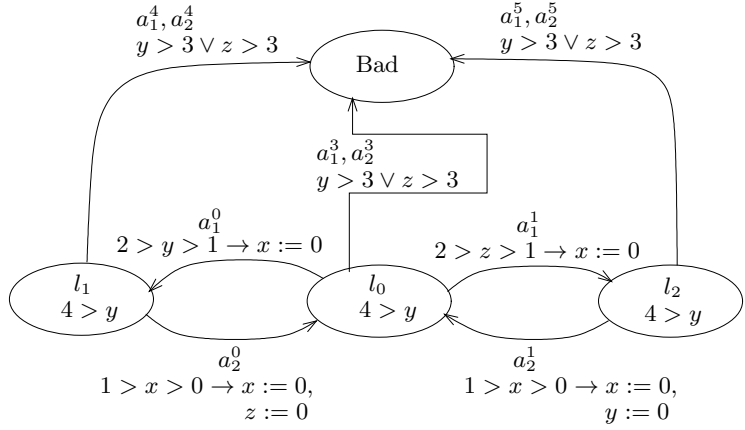

\strut\centerline{\input Figures/example-syntactic.eepic}
\caption{A timed automaton game.}
\label{figure:example-syntactic}
\end{figure}
The details of the game are not important and are omitted here for the sake
of brevity.
In the figure, the edges are labelled as $a_1^j$ for actions controlled
by player~1; and by $a_2^j$ for actions controlled
by player~2.
The safety objective is to avoid the location ``$\bad$'' (player~1 can satisfy
this objective without blocking time).
One can easily show that zeno runs are possible in this timed automaton game,
mainly, due to the edges $a_2^0$ and $a_2^1$.
The game $\A$ can be made to be non-zeno syntactically by changing the
guards of the edges $a_2^0$ and $a_2^1$ to $1>x>d$, where $d$ is some 
conservative constant (say  $0.001$ time units , where it is assumed that the 
plant takes
at least $0.001$ time units to transition out of $l_1$ and $l_2$).
This change unfortunately blows up the finite state \emph{region abstraction}
of the timed automaton game (the region abstraction
is used in every current solution to the real-time controller synthesis
problem for timed automaton games).
If the constant $d$ is $0.001$, then the number of states in the region
abstraction blows up from roughly $2.5*10^5$ for the original game to
$2.5*10^5*10^9$; a blow up by a \emph{factor} of $10^9$.
Admittedly however, on the fly algorithms for controller synthesis may 
help mitigate the situation in some cases (\cite{CassezDFLL05})
 by not explicitly
constructing the full graph of the region abstraction.

Requiring player~1 to guarantee time divergence by only taking
transitions if player~2 cannot take transitions from the current location is
too conservative.
If we consider the game in Figure~\ref{figure:example-syntactic}, this
approach would prevent player~1 from taking \emph{any} of the actions, making
the system uncontrollable.
Finally, adding an extra infinitely precise clock to measure time, and making it
observable to the controller amounts to giving unfair and unrealistic 
power to the
controller in many situations.

In the present paper, we avoid the shortcomings of the previous approaches
by using two techniques.
First, we use \emph{receptive} \cite{AluHen97,SegalaGSL98}, 
player-1 strategies,  which, 
while being required to 
not prevent time from diverging, are not required to ensure time 
divergence.
Receptiveness is incorporated by using the more general, semantic and 
fully symmetric formalism of
\cite{AFHM+03} for dealing with the issue of time divergence.
This setting places no syntactic restriction on the game structure,
and gives both players equally powerful options for advancing time, but
for a player to win, it must not be \emph{responsible} for causing
time to converge.
Formally, our timed games proceed in an infinite sequence of rounds.
In each round, both players simultaneously propose moves, with each move 
consisting of an action and a time delay after which the player wants 
the proposed action to take place.
Of the two proposed moves, the move with the shorter time delay ``wins'' 
the round and determines the next state of the game.
Let a set $\Phi$ of runs be the desired objective for player~1.
Then player~1 has a \emph{winning} strategy for $\Phi$ if 
it has a strategy to ensure that, no matter what player~2 does, one 
of the following two conditions hold:
(1)~time diverges and the resulting run belongs to $\Phi$, or
(2)~time does not diverge but player-1's moves are chosen only finitely 
often (and thus it is not to be blamed for the convergence of time).
Second, in the current work, 
the controller only uses the system clocks of the model  
(unlike ~\cite{AFHM+03} which makes available to the controller 
an extra infinitely precise clock to measure time), ensuring that the controller
bases its actions only on the variables corresponding to the  physical processes
of the system (the system clocks).
Time divergence is \emph{inferred} from the history of certain predicates of 
the system clocks, rather
than from an extra infinitely precise clock that the controller has to keep
in memory.

\smallskip\noindent\textbf{Contributions.}
Our current work significantly improves  the results of~\cite{KCHenPra08}.
In~\cite{KCHenPra08} we showed that finite-memory receptive strategies suffice for 
safety objective in timed automaton games; the problem of establishing a memory 
bound was left open.
In this paper, we first show that a basic analysis using 
\emph{Zielonka trees} of the characterization of
receptive strategies of~\cite{KCHenPra08} leads to an \emph{exponential} number
of bits for the memory
bound (in the number of clocks) for the winning strategies.
We then present an improved new  characterization of receptive strategies
for safety objectives which allows us to obtain a \emph{linear}  number of bits
for the memory bound for winning strategies.
Precisely, we show that a memory of  size $\big(3\cdot|C|+1 + \lg(|C|+1)\big)$ 
bits suffices for  winning
 receptive strategies for safety objectives, where $C$ is the
set of clocks of the timed automaton game, considerably improving the 
exponential bound obtained from the previous result.
Finally, we settle the open question of whether \emph{region} strategies 
for controllers require memory for safety objectives.
We show with an example that region strategies in general do require memory for 
safety objectives.

\section{Timed Games}
\label{section:Setting}
\subsection{Timed Game Structures}
In this Subsection we present the definitions of timed game structures,
runs, objectives, strategies and the notions of sure and almost-sure 
winning in timed game structures.

\smallskip\noindent{\bf Timed game structures.} 
A \emph{timed game structure} is a tuple 
$\TG = \tuple{S,\acts_1,\acts_2,\Gamma_1,\Gamma_2,\delta}$
with the following components.
\begin{itemize}
\item 
$S$ is a set of states.
\item 
$\acts_1$ and $\acts_2$ are two disjoint sets of actions for players~1 
and~2, respectively.
We assume that $\bot_i\not\in \acts_i $, and write 
$\acts_i^{\bot}$ for $\acts_i\cup\set{\bot_i}$.
The set of \emph{moves} for player $i$ is 
$M_i =\reals_{\geq 0} \times \acts_i^{\bot}$.
Intuitively, a move $\tuple{\Delta,a_i}$ by player $i$ indicates a
waiting period of $\Delta$ time units followed by a discrete
transition labeled with action~$a_i$.
The move $\tuple{\Delta, \bot_i}$ is used to represent the move of player~$i$
where player-$i$ just lets time elapse for $\Delta$ time units without taking 
any of the  discrete actions from $A_i$.

\item
$\Gamma_i : S\mapsto 2^{M_i} \setminus \emptyset$ are two move assignments.
At every state~$s$, the set $\Gamma_i(s)$ contains the moves that are 
available to player $i$.
We require that $\tuple{0,\bot}\in\Gamma_i(s)$ for all states $s\in S$ 
and $i\in\set{1,2}$.
Intuitively, $\tuple{0,\bot_i}$ is a time-blocking stutter move.
\item
$\delta: S\times (M_1 \cup M_2) \mapsto S$ is the transition function.
We require that for all time delays 
$\Delta,\Delta'\in\reals_{\ge 0}$ with $\Delta'\leq \Delta$,
and all actions $a_i\in \acts_i^{\bot}$,
we have 
(1)~$\tuple{\Delta,a_i}\in\Gamma_i(s)$ iff both
$\tuple{\Delta',\bot_i}\in \Gamma_i(s)$ and 
$\tuple{\Delta -\Delta',a_i}\in\Gamma_i(\delta(s,\tuple{\Delta',\bot_i}))$;
and 
(2)~if $\delta(s,\tuple{\Delta',\bot_i})=s'$ and 
$\delta(s',\tuple{\Delta-\Delta',a_i})=s''$, then 
$\delta(s,\tuple{\Delta,a_i}) = s''$.
\end{itemize}
The game proceeds as follows.
If the current state of the game is~$s$, then both players simultaneously 
propose moves $\tuple{\Delta_1,a_1}\in\Gamma_1(s)$ and 
$\tuple{\Delta_2,a_2}\in\Gamma_2(s)$.
The move with the shorter duration ``wins'' in determining the next state 
of the game. 
If both moves have the same duration, then  player~2 determines whether
the next state will be determined by its move, or by the move of player~1.
We use this setting as  our goal is to compute the winning
set for player~1 against all possible strategies of player~2.
Formally, we define the \emph{joint destination function} 
$\delta_{\jd} : S\times M_1\times M_2 \mapsto 2^S$ by
\[
\delta_{\jd}(s,\tuple{\Delta_1,a_1},\tuple{\Delta_2,a_2}) = \left\{
\begin{array}{ll}
\set{\delta(s,\tuple{\Delta_1,a_1})} & \text{ if } \Delta_1 < \Delta_2; \\
\set{\delta(s,\tuple{\Delta_2,a_2})} & \text{ if } \Delta_2 < \Delta_1;\\
\set{\delta(s,\tuple{\Delta_2,a_2}),\delta(s,\tuple{\Delta_1,a_1})}
 & \text{ if } \Delta_2 = \Delta_1.

\end{array}
\right.
\]
The time elapsed when the moves $m_1=\tuple{\Delta_1,a_1}$ and 
$m_2=\tuple{\Delta_2,a_2}$ are proposed is given by 
$\delay(m_1,m_2) = \min(\Delta_1,\Delta_2)$. 
The boolean predicate $\Blfunc_i(s,m_1,m_2,s')$ indicates whether player~$i$ 
is ``responsible'' for the state change from $s$ to $s'$ when the moves $m_1$ 
and $m_2$ are proposed.
Denoting the opponent of player~$i$ by $\negspaceopp{i} = 3-i$, for $i \in
\set{1,2}$, we define 
 $$\Blfunc_i(s,\tuple{\Delta_1,a_1},\tuple{\Delta_2,a_2},s')\ =\ 
   \big(\Delta_i \leq \Delta_{\opp{i}}\ \wedge\ 
   \delta(s,\tuple{\Delta_i,a_i}) = s'\big).$$

\smallskip\noindent{\bf Runs.} A \emph{run} of the timed game
structure $\TG$ is an infinite sequence $r=\run$ such that $s_k\in S$
and $m_i^k \in \Gamma_i(s_k)$ and $s_{k+1} \in
\delta_{\jd}(s_k,m_1^k,m_2^k)$ for all $k\geq 0$ and $i\in\set{1,2}$.
For $k\ge 0$, let $\runtime(r,k)$ denote the ``time'' at position $k$ of the 
run, namely, 
$\runtime(r,k)=\sum_{j=0}^{k-1}\delay(m_1^j,m_2^j)$ (we let $\runtime(r,0)=0$).
By $r[k]$ we denote the $(k+1)$-th state $s_k$ of~$r$.
The run prefix $r[0..k]$ is the finite prefix of the run $r$ that ends in 
the state~$s_k$.
Let $\iruns$ be the set of all runs of $\TG$, and let $\VRuns$ be the set 
of run prefixes.

\smallskip\noindent{\bf Objectives.}
An \emph{objective} for the timed game structure $\TG$ is a set 
$\Phi\subseteq\iruns$ of runs.
We will be interested in the classical safety objectives.
Given a set of states $Y$, the \emph{safety} objective consists
of the set of runs that stay within $Y$, formally, 
$\safe(Y)=\set{ r\mid \text{ for all } i \text{ we have } r[i]\in Y}$.
To solve timed games for safety objectives, we shall need to solve for
for certain $\omega$-regular objectives (see~\cite{Thomas97} for the definition of
$\omega$-regular sets).

\smallskip\noindent{\bf Strategies.}
A \emph{strategy} for a player is a recipe that specifies
how to extend a run.
Formally, 
a \emph{probabilistic strategy} $\pi_i$ for player $i\in \set{1,2}$ is a function 
$\pi_i$ that assigns to every run prefix 
$r[0..k]$ a probability measure  $P_{\pi_i}^{r[0..k]}$ over 
$\Gamma_i(r[k])$, 
the set of moves available to  player~$i$ at the state $r[k]$ (the event
class can be suitably chosen).
%
\emph{Pure strategies} are strategies for which the state space of
the probability distribution of $P_{\pi_i}^{r[0..k]}$
is a singleton set for every run $r$ and all $k$.
We let $\Pi_i^{\pure}$ denote the set of pure strategies for player~$i$, with
$i\in\set{1,2}$.
We call probability distributions with singleton support sets as \emph{pure 
distributions}.

%
%
For $i\in\set{1,2}$, let $\Pi_i$ be the set of strategies for player~$i$.
If both both players propose the same time delay, then the tie is
broken by a \emph{scheduler}.
Let $\tiebreak$ be the set of functions from 
$\reals_{\geq 0}\times A_1^{\bot}\times A_2^{\bot}$ to 
$\set{1,2}$.
A \emph{scheduler strategy} $\pi_{\sched}$ is a mapping 
from $\VRuns$ to $\tiebreak$.
If $\pi_{\sched}(r[0..k]) = h$, then 
 the resulting state given player~1 and player~2
moves $\tuple{\Delta,a_1}$ and $\tuple{\Delta,a_2}$ respectively, is
determined by the move of player~$h(\Delta, a_1, a_2)$. 
We denote the set of all scheduler strategies by $\Pi_{\sched}$.
Given two strategies $\pi_1\in \Pi_1$ and $\pi_2\in \Pi_2$, the set of 
possible \emph{outcomes} of the game starting from a state $s\in S$ is 
denoted $\outcomes(s,\pi_1,\pi_2)$.
We let $\outcomes_k(s,\pi_1,\pi_2)$ denote the set of finite runs
$r[0..k-1]$
which are possible according to the
two strategies given the initial state $s$.
If we fix the scheduler strategy $\pi_{\sched}$ then the set of possible
outcomes is denoted by $\outcomes(s,\pi_1,\pi_2, \pi_{\sched})$.
Given strategies $\pi_1$ and $\pi_2$, for player~1 and player~2, respectively,
a scheduler strategy $\pi_{\sched}$ and a starting state $s$ we denote 
by $\Pr_s^{\pi_1,\pi_2,\pi_{\sched}}(\cdot)$ the 
probability space over $\iruns$ given the strategies and the initial state $s$.

\smallskip\noindent{\bf Receptive strategies.}
We will be interested in strategies that are meaningful
(in the sense that they do not block time).
To define them formally we first present  the following two
sets of runs.
\begin{itemize}
\item
A run $r$ is \emph{time-divergent} if 
$\lim_{k\rightarrow\infty}\runtime(r,k) = \infty$.
We denote by $\td$  the set of all time-divergent runs.
\item
The set $\blameless_i\subseteq\iruns$ consists of
the set of runs in which player $i$ is 
responsible only for finitely many transitions.
A run $\run$ belongs to the set $\blameless_i$, for $i=\set{1,2}$, 
if there exists a $k\ge 0$ such that for all $j\ge k$, we have 
$\neg\Blfunc_i(s_j,m_1^{j},m_2^{j}, s_{j+1})$.
\end{itemize}
A strategy $\pi_i$ is \emph{receptive} if  for
all strategies $\pi_{\opp{i}}$, all states $s\in S$, and all
runs $r\in\outcomes(s,\pi_1,\pi_2)$, either $r\in\td$ or
$r\in\blameless_i$.
Thus, no what matter what the opponent does, a receptive  
strategy of player~$i$ cannot be responsible for blocking time.
Strategies that are not receptive  are not physically meaningful.
A timed game structure $\TG$ is \emph{well-formed} if both players have 
receptive strategies.
We restrict our attention to well-formed timed game structures.
We denote $\Pi_i^R$ to be the set of receptive strategies for player~$i$.
Note that for $\pi_1\in\Pi_1^R, \pi_2\in\Pi_2^R$, we have
$\outcomes(s,\pi_1,\pi_2)\subseteq \td$.


\smallskip\noindent{\bf Sure and almost-sure winning modes.}
Let $\sure_1^{\TG}(\Phi)$ (resp. $\almostsure_1^{\TG}(\Phi)$) 
be the set of states $s$ in $\TG$ such that 
player~1 has a receptive strategy
 $\pi_1\in \Pi_1^R$ such that for all scheduler strategies
$\pi_{\sched}\in \Pi_{\sched}$ and for all
player-2 receptive strategies 
 $\pi_2\in \Pi_2^R$, we have $\outcomes(s,\pi_1,\pi_2)\subseteq \Phi$
 (resp. $\Pr_{s}^{\pi_1,\pi_2, \pi_{\sched}}(\Phi)=1$).
Such a winning strategy is said to be a sure (resp. almost sure)
winning receptive strategy.
In computing the winning sets, we shall quantify over \emph{all} strategies,
but modify the objective to take care of time divergence.
Given an objective $\Phi$, let $\timedivbl_1(\Phi) = (\td\cap\ \Phi)\cup 
( \blameless_1 \setminus \td)$, i.e., $\timedivbl_1(\Phi)$ denotes the
set of paths such that either time diverges and $\Phi$ holds, or else 
time converges and player~1 is not responsible for time to converge.
A player-1  strategy  is hence receptive iff it 
ensures that against all player-2 strategies,
the resulting runs belong to $\timedivbl_1(\iruns)$.
Let $\sureu_1^{\TG}(\Phi)$ (resp. $\almostsureu_1^{\TG}(\Phi)$) be the 
set of states in $\TG$ such that for all $s\in \sureu_1^{\TG}(\Phi)$ (resp.
$\almostsureu_1^{\TG}(\Phi)$), player~1 has a strategy $\pi_1\in
\Pi_1$ such that for all strategies 
for all scheduler strategies
$\pi_{\sched}\in \Pi_{\sched}$ and for all
player-2 strategies $\pi_2\in \Pi_2$, we have
$\outcomes(s,\pi_1,\pi_2)\subseteq\, \Phi$ (resp.
$\Pr_s^{\pi_1,\pi_2, \pi_{\sched}}(\Phi)=1$). 
Such a winning strategy is said to be a sure (resp. almost sure)
winning for the non-receptive game.
The following result establishes the connection between $\sure$ and
$\sureu$ sets.

\begin{theorem} [\cite{HenPra06}]
\label{theorem:ReceptiveTimedivbl} 
For all well-formed timed game structures
$\TG$, and for all $\omega$-regular objectives $\Phi$, we have
$\sureu_1^{\TG}(\timedivbl_1(\Phi))= \sure_1^{\TG}(\Phi)$.
\end{theorem}

We observe here that $\timedivbl_1(\Phi)$ is \emph{not} equivalent to
$(\neg\blameless_1) \rightarrow \td\cap\,\Phi$. 
Player~1 loses even if it does not get moves infinitely often, provided time
diverges and the run does not belong to $\Phi$.
\subsection{Timed Automaton Games}
\label{subsection:TimedAutomatonGames}
In this Subsection 
we define a special class of timed game structures, namely,
timed automaton games, and the notion of region equivalence.

\smallskip\noindent{\bf Timed automaton games.}
Timed automata~\cite{AlurD94} suggest a finite syntax for specifying
infinite-state timed game structures.
A \emph{timed automaton game} is a tuple 
$\A=\tuple{L,C,\acts_1, \acts_2,E,\inv}$ 
with the following components:
\begin{itemize}	
\item 
$L$ is a finite set of locations.
\item 
$C$ is a finite set of clocks.
\item 
$\acts_1$ and $\acts_2$ are two disjoint sets of actions for players~1 
and~2, respectively.
\item 
$E \subseteq L\times (\acts_1\cup \acts_2)\times \clkcond(C)\times L
  \times 2^C$
is the edge relation, where the set $\clkcond(C)$ of 
\emph{clock constraints} is generated by the grammar 
  $$\theta ::= x\leq d \mid d\leq x\mid \neg\theta \mid 
    \theta_1\wedge\theta_2$$ 
for clock variables $x\in C$ and nonnegative integer constants~$d$.
For an edge $e=\tuple{l,a_i,\theta,l',\lambda}$, the clock constraint 
$\theta$ acts as a guard on the clock values which specifies when the 
edge $e$ can be taken, and by taking the edge~$e$, the clocks in the set 
$\lambda\subseteq C$ are reset to~0.
We require that for all edges 
$\tuple{l,a_i,\theta',l',\lambda'},\tuple{l,a_i,\theta'',l'',\lambda''}\in E$ 
with $l'\neq l''$, the conjunction $\theta'\wedge\theta''$ is unsatisfiable.
This requirement ensures that a state and a move together uniquely determine 
a successor state.
\item 
$\inv: L\mapsto\clkcond(C)$ is a function that assigns to 
every location an invariant for both players.  
All clocks increase uniformly at the same rate.
When at location~$l$, each player~$i$ must propose a move out of $l$ 
before the invariant $\inv(l)$ expires.
Thus, the game can stay at a location only as long as the invariant is 
satisfied by the clock values.
\end{itemize}
A \emph{clock valuation} is a function  $\kappa : C\mapsto\reals_{\geq 0}$
that maps every clock to a nonnegative real. 
The set of all clock valuations for $C$ is denoted by $K(C)$.
Given a clock valuation $\kappa\in K(C)$ and a time delay 
$\Delta\in\reals_{\geq 0}$, we write 
$\kappa +\Delta$ for the clock valuation in $K(C)$ defined by 
$(\kappa +\Delta)(x) =\kappa(x) +\Delta$ for all clocks $x\in C$.
For a subset $\lambda\subseteq C$ of the clocks, we write 
$\kappa[\lambda:=0]$ for the clock valuation in $K(C)$ defined by 
$(\kappa[\lambda:=0])(x) = 0$ if $x\in\lambda$, 
and $(\kappa[\lambda:=0])(x)=\kappa(x)$ if $x\not\in\lambda$.
A clock valuation $\kappa\in K(C)$ \emph{satisfies} the clock constraint 
$\theta\in\clkcond(C)$, written $\kappa\models \theta$, if the condition 
$\theta$ holds when all clocks in $C$ take on the values specified 
by~$\kappa$.
A \emph{state} $s=\tuple{l,\kappa}$ of the timed automaton game $\A$ is a 
location $l\in L$ together with a clock valuation $\kappa\in K(C)$ such 
that the invariant at the location is  satisfied, that is,
$\kappa\models\inv(l)$.
Let $S$ be the set of all states of~$\A$.
In a state, each player~$i$ proposes a time delay allowed by the 
invariant map~$\inv$, together either with the action~$\bot$, 
or with an action $a_i\in\acts_i$ such that an edge labeled $a_i$ 
is enabled after the proposed time delay.
We require that for $i\in\set{1,2}$ and for all states $s=\tuple{l,\kappa}$, 
if $\kappa\models\inv(l)$, either $\kappa+\Delta\models\inv(l)$ for all
$\Delta\in\reals_{\geq 0}$, or there exist a time delay 
$\Delta\in\reals_{\geq 0}$ and an edge 
$\tuple{l,a_i,\theta,l',\lambda}\in E$ such that 
(1)~$a_i\in\acts_i$ and 
(2)~$\kappa+\Delta\models\theta$ and 
for all $0\le\Delta'\le\Delta$, we have $\kappa+\Delta'\models\inv(l)$, and 
(3)~$(\kappa+\Delta)[\lambda:=0]\models\inv(l')$.
This requirement is necessary (but not sufficient) for well-formedness of 
the game.

The timed automaton game $\A$ defines the following timed game
structure $\symb{\A} =
\tuple{S,\acts_1,\acts_2,\Gamma_1,\Gamma_2,\delta}$:

\begin{itemize}
\item
$S = \set{\tuple{l,\kappa} \mid l\in L \text{ and } \kappa(l) \text{ satisfies
} \gamma(l)}$.
\item
For $i\in\set{1,2}$, the set 
$\Gamma_i(\tuple{l,\kappa})$ contains the following elements:
\begin{enumerate}
\item 
$\tuple{\Delta,\bot_i}$ if for all $0\le\Delta'\le\Delta$, we have $\kappa+\Delta'\models\inv(l)$.
\item
$\tuple{\Delta,a_i}$ if  for all $0\le\Delta'\le\Delta$, we have
$\kappa+\Delta'\models\inv(l)$,  $a_i\in\acts_i$, and 
there exists an edge $\tuple{l,a_i,\theta,l',\lambda}\in E$ 
such that $\kappa+\Delta\models\theta$.
\end{enumerate}
\item
The transition function $\delta$ is specified by:
\begin{enumerate}
\item           
$\delta(\tuple{l,\kappa},\tuple{\Delta,\bot_i}) = \tuple{l,\kappa +\Delta}$.
\item 
$\delta(\tuple{l,\kappa},\tuple{\Delta,a_i}) = 
\tuple{l',(\kappa +\Delta)[\lambda:=0]}$ 
for the unique edge $\tuple{l,a_i,\theta,l',\lambda} \in E$ with 
$\kappa+\Delta\models \theta$.
\end{enumerate}
\end{itemize}
The timed game structure $\symb{\A}$ is not necessarily well-formed, because 
it may contain cycles along which time cannot diverge.
Well-formedness of timed automaton games can be checked in
\EXPTIME~\cite{HenPra06}.
We restrict our focus to well-formed timed automaton games in this paper.
We shall also restrict our attention to randomization over time ---
 a random move of a player in a timed automaton game 
will consist of a distribution over time over
some interval $I$, denoted $\distro^I$, together with a discrete action $a_i$.

\smallskip\noindent{\bf Clock region equivalence.}
Timed automaton games can be solved using a region construction from 
the theory of timed automata~\cite{AlurD94}.
For a real $t\ge 0$, let  $\fractional(t)=t-\floor{t}$ denote the 
fractional part of~$t$.
Given a timed automaton game $\A$, for each clock $x\in C$, let $c_x$
denote the largest integer constant that appears in any clock
constraint involving $x$ in~$\A$ (let $c_x=1$ if there is no clock
constraint involving~$x$).
Two clock valuations $\kappa_1,\kappa_2$ are said to be
\emph{region equivalent}, denoted by $\kappa_1\cong \kappa_2$ 
 when all the following conditions hold.
\begin{enumerate}
\item
  For all  clocks $x$ with $\kappa_1(x) \leq c_x $ and 
  $\kappa_2(x) \leq c_x$, we have  
  $\lfloor \kappa_1(x) \rfloor = \lfloor \kappa_2(x) \rfloor$.
  
\item
  For all  clocks $x,y$ with $\kappa_i(x) \leq c_x $ and 
  $\kappa_i(y) \leq c_y$, we have  
  $\fractional(\kappa_1(x)) \leq  \fractional(\kappa_1(y))$
  iff
  $\fractional(\kappa_2(x)) \leq  \fractional(\kappa_2(y))$.

\item
  For all  clocks $x$ with $\kappa_1(x) \leq c_x $ and 
  $\kappa_2(x) \leq c_x$, we have  
  $\fractional(\kappa_1(x))=0$  iff $ \fractional(\kappa_2(x))=0$.

\item
  For any clock $x$, $\kappa_1(x) > c_x$ iff $\kappa_2(x) > c_x$.
  Two states $\tuple{\kappa_1,l_1}$ and $\tuple{\kappa_1,l_1}$ are
  region equivalent iff $l_1=l_2$ and $\kappa_1\cong \kappa_2$.
\end{enumerate}
A {\em region} $R$ of a timed automaton game $\A$ is an
equivalence class of states with respect to the region equivalence
relation.

\smallskip\noindent{\bf Representing regions.}
We find it useful to sometimes denote a region $R$ by  a tuple 
$\tuple{l,h,\parti(C)}$ where
\begin{itemize}
\item
  $l$ is a location of $\A$.

\item
  $h$ is a function which specifies the integer values of
  clocks $ h : C \rightarrow (\nat\cap [0,M])$ 
  ($M$ is the largest constant in $\A$). 
\item
  $\parti(C)$ is a disjoint partition of the clocks into the tuple
  $\tuple{C_{-1},C_0,\dots C_n}$ such that
 $\set{C_{-1},C_0,\dots C_n \mid \uplus C_i = C, C_i\neq\emptyset \text{ for } i>0 }$.
\end{itemize}
A state $s$ with clock valuation $\kappa$ is then in the region $R$ when all
the following conditions hold.
\begin{enumerate}
\item
  The location of $s$ corresponds to the location of $R$.
\item
  For all  clocks $x$ with $\kappa(x) \leq c_x $, 
  $\lfloor \kappa(x) \rfloor = h(x)$.
\item
  For $\kappa(x) >  c_x$, $h(x) =c_x$.
\item
  For all pair of clocks $(x,y)$, with $\kappa(x) \leq c_x$ and 
  $\kappa(y)\leq c_y$, we have 
  $\fractional(\kappa(x)) < \fractional(\kappa(y))$ iff 
  $ x\in C_i \text{ and } y\in C_j \text{ with } 0\leq i<j$ 
  (so, $x,y \in C_k$ with $k\geq 0$ implies
  $\fractional(\kappa(x)) = \fractional(\kappa(y))$).
\item
  For $\kappa(x) \leq c_x$, $\fractional(\kappa(x)) = 0$ iff $ x\in C_0$.
\item
  $x\in C_{-1}$ iff $\kappa(x) > c_x$.
\end{enumerate}
There are finitely many clock regions;
more precisely, the number of clock regions is bounded by 
$|L|\cdot\prod_{x\in C}(c_x+1)\cdot |C|!\cdot 2^{2|C|}$.

\smallskip\noindent{\bf Region  equivalent  runs.}
For a state $s\in S$, we write $\reg(s)\subseteq S$ for the clock region 
containing~$s$.
For a run $r$, we let the \emph{region flow sequence} $\reg(r)$ be the
sequence of regions
$R_0,R_1,\cdots$ which intuitively denotes the regions
encountered (including those during time passage specified by moves) in $r$.
Formally, $\reg(r)$ is the region sequence $R_0,R_1,\cdots$ is such that
there exist $i_0=0 < i_1 < i_2\dots$ with
(1)~$\reg(r[j]) = R_{i_j}$;
(2)~$R_{k_1}\neq R_{k_2}$ for $i_j\leq k_1<k_2<i_{j+1}$ for any $i_j$; and
(3)~if  $r=\run$, and $r[j+1]=\delta(r[j],m_p^j)$ (for $p\in\set{0,1}$), with
$m_p^j=\tuple{\Delta,a}$; then $R_{i_j}, R_{i_j + 1},R_{i_j +2},\dots R_{i_{j+1}-1}$
are the unique regions encountered when $\Delta$ time passes from $r[j]$.
The region flow sequence of a run is unique.
Two runs $r,r'$ are \emph{region equivalent} 
if (1)~their region flow sequences are the same, and 
(2)~$\reg(r[j])=\reg(r'[j])$ for all $j\geq 0$.
Region equivalence for finite runs can be defined similarly.
We similarly define location equivalence for runs (note that a location
flow sequence is just the sequence of locations of the states in a run).
An $\omega$-regular objective $\Phi$ is a location objective if for all 
location-equivalent runs 
$r,r'$, we have $r\in \Phi$ iff $r'\in \Phi$.
A parity index function $\Omega$ is a location  parity index function if
$\Omega(s_1)=\Omega(s_2)$ whenever $s_1$ and $s_2$ have the same location.
Henceforth, we shall restrict our attention to location objectives.

\smallskip\noindent{\bf Region equivalent strategies.}
Given a strategy $\pi$, a run prefix $r[0..k]$, a region $R$, and an
action $a_i\in A_i^{\bot}$, let $\mathcal{W}(\pi,r[0..k],R,a_i)$ denote
the set $\set{\tuple{\Delta, a_i} \mid 
  \tuple{\Delta, a_i} \in \support( \pi(r[0..k]))\text{ and }
  \reg(r[k]+\Delta) = R})$.
A strategy $\pi_1$ is a \emph{region strategy}, if for all run prefixes 
$r_1[0..k]$ and $r_2[0..k]$
such that $\reg(r_1[0..k])=\reg(r_2[0..k])$, 
and for all regions $R$ and player-1
actions $a_1\in A_1^{\bot}$, we have 
(1)~$\mathcal{W}(\pi_1,r_1[0..k],R,a_1) =\emptyset$ iff
$ \mathcal{W}(\pi_1,r_2[0..k],R,a_1)=\emptyset$;
and
(2)~$P_{\pi_1}^{r_1[0..k]}(\mathcal{W}(\pi_1,r_1[0..k],R,a_1)) =
P_{\pi_1}^{r_2[0..k]}(\mathcal{W}(\pi_1,r_2[0..k],R,a_1))$.
The definition for player~2 strategies is analogous.
Two region strategies $\pi_1$ and $\pi_1'$ are region-equivalent if for all 
run prefixes $r[0..k]$, and for all regions $R$ and player-1 actions
$a_1\in A_1^{\bot}$, we have
(1)~$\mathcal{W}(\pi_1,r[0..k],R,a_1) = \emptyset$ iff
$ \mathcal{W}(\pi_1',r[0..k],R,a_1) = \emptyset$;
and
(2)~$P_{\pi_1}^{r[0..k]}(\mathcal{W}(\pi_1,r[0..k],R,a_1)) =
P_{\pi_1'}^{r[0..k]}(\mathcal{W}(\pi_1',r[0..k],R,a_1))$.


\subsection{Winning Sets and Winning Strategies
for Timed Automaton Games}
\label{subsection:ResultsTimedAutomatonGames}

In this Subsection we present the computation of winning sets for 
timed automaton games
based on the framework of~\cite{AFHM+03}, and derive various basic properties of 
winning strategies.

\smallskip\noindent{\bf Encoding Time-Divergence by Enlarging the Game 
Structure.}
Given a timed automaton game $\A$, consider the enlarged game structure
$\w{\A}$ (based mostly on the construction in~\cite{AFHM+03}) 
with the state space $S^{\w{\A}} \subseteq S \times
\reals_{[0,1)}\times\set{\true,\false}^2$,
and an augmented transition relation $\delta^{\w{\A}}:
S^{\w{\A}}\times (M_1 \cup M_2) \mapsto S^{\w{\A}}$.  In an
augmented state $\tuple{s,\z,\tick,\bl_1} \in S^{\w{\A}}$, the
component $s\in S$ is a state of the original game structure
$\symb{\A}$, $\z$ is value of a fictitious clock $z$ which gets reset to 0
every time it crosses 1 (i.e., if $\kappa'$ is the clock valuation resulting
from letting time $\Delta$ elapse from an initial clock valuation $\kappa$, 
then, $\kappa'(z) = (\kappa(z)+\Delta)\mod 1$),
$\tick$ is true 
iff $z$ crossed  1 at last transition and $\bl_1$ is 
true if player~1 is to blame 
for the last transition (ie., $\Blfunc_1$ is true for the last transition).
Note that any strategy $\pi_i$ in $\symb{\A}$, can be considered a strategy in
$\w{\A}$.
The values of the clock $z$, $\tick$ and $\bl_1$ correspond to the values
each player keeps in memory in constructing his strategy.
Given any initial value of $\z=\z^*,\tick=\tick^*,\bl_1=\bl_1^*$;
any run $r$ in $\A$ has a corresponding unique run $\w{r}$ in 
$\w{\A}$ with $\w{r}[0]=\tuple{r[0],\z^*,\tick^*,\bl_1^*}$ such 
 that $r$ is a projection of $\w{r}$ onto $\A$. 
For an objective $\Phi$, we can now encode time-divergence as the 
objective:  
$\timedivbl_1(\Phi)=(\Box\Diamond \tick \rightarrow \Phi)\ \wedge\ 
(\neg\Box\Diamond\tick \rightarrow \Diamond\Box \neg\bl_1)$, where
$\Box$ and $\Diamond$ are the standard LTL modalities (``always'' and
``eventually'' respectively), the combinations $\Box\Diamond$ and
$\Diamond\Box$ denoting
``infinitely often'' and ``all but for a finite number of steps''  
respectively.
This is formalized in the following  proposition.

\begin{proposition}[$\timedivbl_1()$ in terms of $\tick, \bl_1$]
\label{proposition:ExpandedGame}
Let $\A$ be a timed automaton game  and
$\w{\A}$ be the corresponding enlarged game
structure.
Let $\Phi$ be an objective on $\A$.
Consider a run $r=s^0,\tuple{m_1^0,m_2^0}, s^1,\tuple{m_1^1,m_2^1}, \dots$ in 
$\A$.
Let $\w{r}$ denote the corresponding run in $\w{\A}$ such that
$\w{r}= \tuple{s^0,\z^0,\tick^0,\bl^0_1},\tuple{m_1^0,m_2^0}, 
\tuple{s^1,\z^1,\tick^1,\bl_1^1},\tuple{m_1^1,m_2^1}$ with
$\z^0=0, \tick^0=\false,\bl_1^0=\false$.
Then $r \in  \timedivbl_1(\Phi)$ iff $\w{r}\in
\left((\Box\Diamond \tick \rightarrow \Phi)\ \wedge\ 
(\neg\Box\Diamond\tick \rightarrow \Diamond\Box \neg\bl_1)\right)$ 
\end{proposition}
\begin{proof}
Time diverges in the run $r$ iff it diverges in the corresponding run $\w{r}$.
Also, the run $r$ belongs to $\blameless_1$ iff the run 
$\w{r}$ belongs to $\blameless_1$, which happens iff 
player~1 is blamed only
finitely often, ie., $\Diamond\Box \neg\bl_1$ holds.
Hence $r \in  \timedivbl_1(\Phi)$ iff $\w{r}\in\timedivbl_1(\Phi)$.
The result follows from noting that time diverges iff time crosses 
integer boundaries infinitely often, which happens iff $\Box\Diamond \tick$
holds.
%
\qed
\end{proof}

The following lemma states that because of the correspondence between
$\A$ and $\w{\A}$, we can obtain the winning sets of  $\A$ by obtaining
the winning sets in $\w{\A}$.

\begin{lemma}[Equivalence of winning sets of $\bm{\A}$ and $\bm{\w{\A}}$]
\label{lemma:ExpandedGame}
Let $\A$ be a timed automaton game  and
$\w{\A}$ be the corresponding enlarged game
structure.  
Let $\Phi$ be an objective on $\A$. 
Given any state $s$ of $\A$, we have
$s\in\sureu_1^{\A}(\timedivbl_1(\Phi))$ iff 
$\tuple{s, 0, \false,\false}\in 
\sureu_1^{\w{\A}}\left((\Box\Diamond \tick \rightarrow \Phi)\ \wedge\ 
(\neg\Box\Diamond\tick \rightarrow \Diamond\Box \neg\bl_1)\right)$.
\end{lemma}
\begin{proof}
Consider a state $s$ of $\A$, and a corresponding state 
$\tuple{s, 0, \false,\false}$ of $\w{\A}$.
The variables $\z,\tick$ and $\bl_1$ only ``observe'' properties in $\w{\A}$, they
do not restrict transitions. 
Thus, given a run $r$ of $\A$ from $s$, there is a unique run $\w{r}$ of $\w{\A}$ 
from  $\tuple{s, 0, \false,\false}$ and vice versa.
Similarly, any player-$i$ strategy $\pi_i$ in $\A$ corresponds to a strategy 
$\w{\pi}_i$ in $\w{\A}$; and any strategy $\w{\pi}_i$ in $\w{\A}$ corresponds
to a strategy $\pi_i$ in $\A$ such that both strategies propose the same moves for
corresponding runs.
The result then follows from Proposition~\ref{proposition:ExpandedGame}.
\qed
\end{proof}

Let $\w{\kappa}$ be a valuation for the clocks in $\w{C}=C\cup\set{z}$.
A state of $\w{\A}$ can then be considered as 
$\tuple{\tuple{l,\w{\kappa}},\tick,\bl_1}$.
We extend the clock equivalence relation to these expanded states: 
$\tuple{\tuple{l,\w{\kappa}}\tick,\bl_1}\cong 
\tuple{\tuple{l',\w{\kappa}'},\tick',\bl_1'}$ 
iff $l=l', \tick=\tick', \bl_1=\bl_1'$ and $\w{\kappa}\cong\w{\kappa}'$.
We let $\tuple{l,\tick,\bl_1}$ be the ``locations'' in $\w{\A}$.
For every $\omega$-regular location objective $\Phi$ of $\A$, we have
$\timedivbl(\Phi)$ to be an $\omega$-regular location objective of $\w{\A}$.

We start first recall 
the statement of a classical result of~\cite{AlurD94} that the region 
equivalence 
relation induces a time abstract bisimulation on the regions.

\begin{lemma}[\cite{AlurD94}]
\label{lemma:Bisimulation}
Let $Y,Y'$ be  regions in the timed game structure
$\A$.
Suppose player~$i$ has a move from $s_1\in Y$
to $s_1'\in Y'$, for $i\in\set{1,2}$.
Then, for any $s_2\in Y$,
player~$i$ has a move from $s_2$ to some $s_2'\in Y'$.
\end{lemma}
Let $Y,Y_1',Y_2'$ be  regions.
We prove in Lemma~\ref{lemma:RegionsBeatRegions} that 
one of the following two conditions hold:
(a) for all states in $Y$ there is a move for player~1 with destination in 
$Y_1'$, such 
that against all player~2 moves with destination in $Y_2'$, the next
state is guarenteed to be in $Y_1'$; or
(b) for all states in $Y$ for all moves for player~1 with destination in $Y_1'$
there is a move of player~2 to ensure that the next state is in $Y_2'$; or
(c) if $Y_1'=Y_2'$ (except for the $\bl_1$ component), then player~2 can pick 
the same time delay as player~1 and hence the winning move is decided by the
scheduler.
The proof of the lemma is in the appendix.

\begin{lemma}[Regions suffice for determining winning move]
  \label{lemma:RegionsBeatRegions}
Let $\A$ be a timed automaton game, and let $Y,Y_1',Y_2'$ be
regions in the corresponding enlarged timed game structure
$\w{\A}$.
Suppose player-$i$ has a move $\tuple{\Delta_i,\bot_i}$ 
from some $\w{s}\in Y$
to $\w{s}_{i}\in Y_i'$, for $i\in\set{1,2}$.
Then,  for all states $\widehat{s}\in Y$ and for all  player-1 moves
  $m_1^{\widehat{s}} = \tuple{\Delta_1,a_1}$ with
  $\widehat{s} + \Delta_1 \in Y_1'$, 
one of the following cases must hold.
\begin{enumerate}
\item 
  $Y_1'\neq Y_2'$ and
  for all moves $m_2^{\widehat{s}}=\tuple{\Delta_2,a_2}$  of player-2 with
  $\widehat{s} + \Delta_2 \in Y_2'$, we have $\Delta_1 < \Delta_2$ 
  (and hence $\Blfunc_1(\widehat{s},m_1^{\widehat{s}},m_2^{\widehat{s}},\widehat{\delta}
(\widehat{s},m_1^{\widehat{s}}))=\true$ and 
$\Blfunc_2(\widehat{s},m_1^{\widehat{s}},m_2^{\widehat{s}},
\widehat{\delta}(\widehat{s},m_2^{\widehat{s}}))=\false$).
\item 
  $Y_1'\neq Y_2'$ and
  for all 
  player-2  moves $m_2^{\widehat{s}}= \tuple{\Delta_2,a_2}$ with
  $\widehat{s} + \Delta_2\in Y_2'$, we have $\Delta_2 < \Delta_1$
  (and hence $\Blfunc_2(\widehat{s},m_1^{\widehat{s}},m_2^{\widehat{s}},
  \widehat{\delta}(\widehat{s},m_2^{\widehat{s}}))=\true$ and
$\Blfunc_1(\widehat{s},m_1^{\widehat{s}},m_2^{\widehat{s}},\widehat{\delta}
(\widehat{s},m_1^{\widehat{s}}))=\false$).
\item
  $Y_1'=Y_2'$ and
  there exists a  player~2 move $m_2^{\widehat{s}}= \tuple{\Delta_2,a_2}$
  with
  $\widehat{s} + \Delta_2\in Y_2'$ 
  such that $\Delta_1 = \Delta_2$ 
  (and hence 
$\Blfunc_1(\widehat{s},m_1^{\widehat{s}},m_2^{\widehat{s}},\widehat{\delta}
(\widehat{s},m_1^{\widehat{s}}))=\true$ and 
$\Blfunc_2(\widehat{s},m_1^{\widehat{s}},m_2^{\widehat{s}},
\widehat{\delta}(\widehat{s},m_2^{\widehat{s}}))=\true$).

\end{enumerate}
\end{lemma}

We now show that 
(1)~pure strategies of player~1 suffice for winning from 
$\sure_1$ states; and
(2)~pure strategies of player~2 suffice for spoiling from
states that are not $\sure_1$.

\begin{lemma}[Existence of pure strategies for sure winning sets]
\label{lemma:PureStrategies}
Let $\TG$ be a timed game structure, and 
let $\Phi$ be an  objective of $\TG$.
\begin{enumerate}
\item
Pure strategies of player~1 suffice for winning 
from $\sure_1^{\TG}(\Phi)$.
\item
Pure strategies of player~2 suffice for preventing sure winning of player~1
from states outside of  $\sure_1^{\TG}(\Phi)$.
\end{enumerate}
\end{lemma}
\begin{proof}
\begin{enumerate}
\item
  Let $\pi_1$ be a sure-winning player-1 receptive strategy.
  Consider any player-1 pure receptive strategy $\pi_1'$ such that for any run $r$ of
  $\TG$, we have $\pi_1'(r[0..k]) \in \support(\pi_1(r[0..k]))$.
  Since $\pi_1$ is sure-winning, $\pi_1'$ must be sure winning too.
\item
  Let $s\notin \sure_1^{\TG}(\Phi)$ and let 
  $\pi_1$ be any player-1 receptive strategy. 
  Let $\pi_2$ be a player-2 spoiling receptive strategy against $\pi_1$ for the state 
  $s$.
  We have $\outcomes(s,\pi_1,\pi_2) \not\subseteq \Phi$.
  This means there exists a run $r^* =
  s_0,\tuple{m_1^0,m_2^0}, s_1,\tuple{m_1^1,m_2^1}, \dots$ with
  $m_i^k \in \support(\pi_i(r^*[0..k]))$ for $i\in\set{1,2}$ such that
  $ r^* \notin \Phi$.
  Consider the pure player-2 receptive strategy $\pi_2'$ such that
  \[
  \pi_2'(r[0..k]) = \left\{
  \begin{array}{ll}
    m_2^k & \text{ if } r[0..k]=r^*[0..k]\\
    \tuple{\Delta_2, a_2} & \text{ otherwise, with } \tuple{\Delta_2, a_2} 
    \text{ being in the support of } \pi_2(r[0..k])
  \end{array}
  \right.
  \]
  The receptive strategy $\pi_2'$ spoils $\pi_1$ from winning surely from $s$ as
  $r^*$  belongs to $\outcomes(s,\pi_1,\pi_2')$, and is not in $\Phi$.
\qed
 
\end{enumerate}

\end{proof}

Lemma~\ref{lemma:PureStrategies} gives us the following corollary which states
that $\sureu_1$ sets are equal to the winning sets if only pure strategies
are allowed for both players.

\begin{corollary}[Equivalence of $\pureu_1$ and $\sureu_1$ sets]
\label{corollary:PureGame}
Let $\A$ be a timed automaton game  and
$\w{\A}$ be the corresponding enlarged game
structure. 
Let $\w{\Phi}$ be an $\omega$-regular location objective of $\w{\A}$, and
let $\pureu_1^{\w{\A}}(\w{\Phi})$ denote the winning set 
for player~1 when
both players are restricted to using only pure strategies.
Then, $\pureu_1^{\w{\A}}(\w{\Phi}) = \sureu_1^{\w{\A}}(\w{\Phi})$.
\end{corollary}

\smallskip\noindent\textbf{A $\mathbf{\mu}$-calculus formulation for
  describing the sure winning sets.}
Given an $\omega$-regular objective $\w{\Phi}$ 
of the expanded game structure $\w{\A}$,
 a $\mu$-calculus formula $\varphi$ to describe the winning set 
$\pureu_1^{\w{\A}}(\w{\Phi})$ (which is equal to $\sureu_1^{\w{\A}}(\w{\Phi})$
by Corollary~\ref{corollary:PureGame}) is given in~\cite{AFHM+03}.
The $\mu$-calculus formula uses the 
\emph{controllable predecessor} operator for player~1,
$\CPre_1 : 2^{\w{S}}\mapsto 2^{\w{S}}$ (where $\w{S}=S^{\w{\A}}$), 
defined formally
by $ \w{s}\in \CPre_1(Z)$ iff
$\exists m_1\in\Gamma^{\w{\A}}_1(\w{s})\;
\forall m_2\in\Gamma^{\w{\A}}_2(\w{s})\,.\, \delta^{\w{\A}}_{\jd}
(\w{s},m_1,m_2) \subseteq Z$.
Informally, $\CPre_1(Z)$ consists of the set of states from which player~1
can ensure that the next state will be in $Z$,  no matter what player~2 does.
The operator  $\CPre_1$ preserves regions of $\w{\A}$ 
(this  follows from the results of Lemma~\ref{lemma:RegionsBeatRegions}).
It was also shown in~\cite{AFHM+03} that only unions of regions
arise in the $\mu$-calculus iteration for $\omega$-regular location
objectives.

We now present a lemma that pure finite-memory strategies suffice for 
winning $\omega$-regular objectives, and all strategies
region-equivalent to a region winning strategy are also 
winning.

\begin{lemma}[Properties of pure winning strategies]
\label{lemma:RegionStrategies} 
Let $\A$ be a timed automaton game  and
$\w{\A}$ be the corresponding enlarged game
structure. 
Let $\w{\Phi}$ be an $\omega$-regular location objective of $\w{\A}$.
Then the following assertions hold.
\begin{itemize}
\item
  If $\pi$ is a player-1 pure  strategy that wins against all player-2
  pure strategies from state $\w{s}$, then $\pi_1$ wins against all
  player-2 strategies from state $\w{s}$.
\item 
  There is a  pure finite-memory region strategy 
 $\pi_1$ that is sure winning for $\w{\Phi}$
  from the states in $\sureu_1^{\w{\A}}(\w{\Phi})$.
\item 
  If $\pi_1$ is a pure region strategy that is sure winning for 
  $\w{\Phi}$ from $\sureu_1^{\w{\A}}(\w{\Phi})$ and $\pi_1'$ is 
  a pure strategy that is region-equivalent to $\pi_1$,
  then $\pi_1'$ is a sure  winning strategy for 
  $\w{\Phi}$ from 
  $\sureu_1^{\w{\A}}(\w{\Phi})$.
\end{itemize}

\end{lemma}
\begin{proof}
\begin{enumerate}
\item
  Since $\pi_1$ wins against all player~2 pure strategies,
  it must also
  win against all player~2 strategies (possibly randomized)  from $\w{s}$ 
  (a randomized player-2 strategy may be viewed as a random choice  over pure
  player-2 strategies).

\item
  It follows from the $\mu$-calculus formulation of~\cite{AFHM+03} that
  there exists a pure finite-memory region strategy $\pi_1$ that wins against
  any pure player~2 strategy from the states in $\pureu_1^{\w{\A}}(\w{\Phi})$.
  From the previous result, $\pi_1$ wins against all  player~2 strategies 
  (possibly randomized)  from
  $\pureu_1^{\w{\A}}(\w{\Phi})$.
  The claim is proved noting that 
  $\pureu_1^{\w{\A}}(\w{\Phi}) = \sureu_1^{\w{\A}}(\w{\Phi})$ from
  Corollary~\ref{corollary:PureGame}.

\item
  Let $\pi_1$ be a pure region strategy that is sure winning for 
  $\w{\Phi}$ from
  a state $\w{s}$.
  Let $\pi_1^*$ be a player-1  pure 
  strategy that is region equivalent to $\pi_1$.
  The strategy  $\pi_1^*$ is a region strategy as $\pi_1$ is a region
  strategy.
  We show that $\pi_1^*$ wins against all player-2 pure strategies.
  The result then follows from the first part of the lemma.

  Consider any player-2 pure strategy $\pi_2$.
  Suppose $\pi_2$ spoils the player-1 strategy $\pi_1^*$ 
  from winning for $\w{\Phi}$ .
  Then, there from the state $\w{s}$ there exists a run 
  $\w{r}^{\,*} = 
  \w{s}_0,\tuple{m_1^0,m_2^0}, \w{s}_1,\tuple{m_1^1,m_2^1}, \dots$ 
  with
  $m_1^k = \pi_1^*(\w{r}^{\,*}[0..k])$ and $m_2^k = \pi_2(\w{r}^{\,*}[0..k])$ 
  such that
  $\w{r}^{\,*} \notin \w{\Phi}$.
  We show that there exists a player-2 pure strategy $\pi_2^{\dagger}$ and a
  run $\w{r}^{\,\dagger} \in \outcomes(\w{s},\pi_1,\pi_2^{\dagger})$ with 
  $\reg(\w{r}^{\,*}) =
  \reg(\w{r}^{\,\dagger})$ (contradicting the assumption that $\pi_1$ was a
  player-1 winning strategy).
  Intuitively, the strategy $\pi_2^{\dagger}$ prescribes moves to the same
  regions as $\pi_2$ if the region sequence observed is the same as
  that of $\reg(\w{r}^{\,*})$.
  Formally, the strategy $\pi_2^{\dagger}$ is defined as follows.
  Given a run $\w{r}$,
  \[
  \pi_2^{\dagger}(\w{r}[0..k]) = \left\{
  \begin{array}{ll}
    \tuple{\Delta_2, a_2} & \text{ if } 
    \reg(\w{r}[0..k]) = \reg(\w{r}^{\,*}[0..k]), \text{ and }
    \pi_1^*(\w{r}^{\,*}[0..k]) = \tuple{\Delta_1^*,a_1}, \text{ and }\\
    &  \pi_2(\w{r}^{\,*}[0..k]) = \tuple{\Delta_2^*,a_2}, \text{ and }
    \Delta_1^* \bowtie  \Delta_2^* \text{ for } \bowtie\,\in\set{<,>,=}, 
    \text{ and }\\
    & \pi_1(\w{r}[0..k]) = \tuple{\Delta_1,a_1} 
    \text{ with } \reg(\w{r}[k]+\Delta_1) =
    \reg(\w{r}^{\,*}[k] +\Delta_1^*) \\
    &(\text{observe that }
    \pi_1 \text{ is a region strategy and } \pi_1^* \text{ is region
      equivalent to } \pi_1),\\
    &\text{ and }
     \Delta_2 \text{ is such that }
    \reg(\w{r}[k]+\Delta_2) =
    \reg(\w{r}^{\,*}[k] +\Delta_2^*) \text{ and }
    \Delta_1 \bowtie \Delta_2 \\
    & (\text{such a } \Delta_2 \text{ must exist by
      Lemma~\ref{lemma:RegionsBeatRegions}.})\\
    \tuple{0,\bot_2} & \text{ otherwise.}
  \end{array}
  \right.
  \]
  It can be checked that there exists a run $\w{r}^{\,\dagger}\in 
  \outcomes(\w{s},\pi_1, \pi_2^{\dagger})$ such that 
  $\reg(\w{r}^{\,\dagger}) = \reg(\w{r}^{\,*})$.
  This contradicts the fact that $\pi_1$ was a winning strategy.
  Thus, there cannot exist a player-2 pure strategy $\pi_2$ which
  prevents the player-1 strategy $\pi_1^*$ from winning.
  Hence, from the first part of the Lemma, $\pi_1^*$ is a player-1 
  winning strategy.
  \qed
\end{enumerate}

\end{proof}

Note that there is an infinitely precise global clock $z$ in the enlarged
game structure $\w{\A}$.
If $\A$ does not have such a global clock, then strategies in $\w{\A}$ 
correspond to strategies in $\A$ where player~1 (and player~2) maintain
the value of the infinitely precise global clock in memory (requiring
infinite memory).


\section{Pure Finite-memory Receptive Strategies 
for Safety Objectives}
\label{section:Safety}
In this section we show the existence of pure finite-memory sure winning 
strategies for safety objectives in timed automaton games, and their memory
requirements.
The encoding of time-divergence in 
Subsection~{subsection:ResultsTimedAutomatonGames} required an infinitely 
precise which had to be kept in memory of player~1, requiring infinite
memory.
In this section, we derive an alternative characterization of receptive
strategies which does not requires this extra clock.
The characterization of receptive strategies is then used to derive 
receptive strategies for safety objectives.
We also show that our derived winning strategies for safety objectives
require only $(|C|+1)$ memory (where $C$ is the set of clocks of the
timed automaton game).

\subsection{Analyzing Spoiling Strategies of Player~2}
\label{subsection:SafetyFiniteStateGame}

In this subsection we analyze the spoiling strategies of player~2.
This analysis will be used in characterizing the receptive strategies of
player~1.

\smallskip\noindent\textbf{Adding predicates to the game structure}.
We add some  predicates to  timed automaton games; the predicates will be 
used later to analyze receptive safety  strategies.
 Given a timed automaton game $\A$ and a state $s$ of $\A$,
 we define two functions
 $V_{>0} : C \mapsto \set{\true,\false}$
 and $V_{\geq 1} : C \mapsto \set{\true,\false}$.
 We obtain $2\cdot |C|$ predicates based on the two functions.
For a clock $x$, the values of the predicates $V_{>0}(x)$ and $V_{\geq 1} (x)$
 indicate if the 
value of clock $x$ was greater than 0, or 
greater than or equal to 1 respectively,
at the transition point, just  before the reset map. 
For example, for a state $s^p= \tuple{l^p,\kappa^p}$ and 
$\delta(s^p, \tuple{\Delta,a_1}) = s$, 
the predicate  $V_{>0}(x)$ is $\true$ at state $s$  iff 
$\kappa'(x) > 0$ for $\kappa' = \kappa^p+\Delta$.
Consider the enlarged game structure
$\widetilde{\A}$ with the state space 
$\widetilde{S} = S \times\set{\true,\false} \times \set{\true,\false}^C \times 
\set{\true,\false}^C$
and an augmented transition relation $\widetilde{\delta}$.
A  state of $\widetilde{\A}$ is a tuple 
$\tuple{s,\bl_1, V_{>0}, V_{\geq 1}}$, where $s$ is a state of $\A$, 
the component $\bl_1$ is  $\true$ iff player~1 is to be blamed for 
the last transition, and
$V_{>0}, V_{\geq 1}$ are as defined earlier.
The clock equivalence relation can be lifted to states of $\widetilde{\A}: 
\tuple{s,\bl_1, V_{>0}, V_{\geq 1}}  \cong_{\widetilde{A}} 
\tuple{s',\bl_1',  V_{>0}',V_{\geq 1}' }$ iff 
$s\cong_{\A} s'$, $\bl_1=\bl_1'$, $V_{>0}= V_{ >0}'$ and 
$V_{\geq 1}= V_{\geq 1}'$.
We next present a finite state concurrent game 
$\widetilde{\A}^{\f}$ based on the regions of 
$\widetilde{\A}$ which will be used to analyze spoiling strategies of
player~2.

\medskip\noindent\textbf{Finite state concurrent game 
$\mathbf{\widetilde{\A}^{\f}}$ based on the regions of 
$\mathbf{\widetilde{\A}}$}.
We first show that there exists an finite state concurrent game 
$\widetilde{\A}^{\f}$
 which can be used to obtain winning sets and
winning strategies of $\widetilde{\A}$.
The two ideas behind $\widetilde{\A}^{\f}$ are that 
(1)~only region sequences are
important for games with $\omega$-regular location objectives, and
(2)~only the destination regions of the players are important (due to
Lemma~\ref{lemma:RegionsBeatRegions}).
Formally, the game $\widetilde{\A}^{\f}$ is defined as the tuple
$\tuple{S^{\f}, M_1^{\f}, M_2^{\f},\Gamma_1^{\f}, \Gamma_2^{\f}, \delta^{\f}}$ 
where
\begin{itemize}
\item 
  $S^{\f}$ is the set of states of 
  $\widetilde{\A}^{\mathcal{F}}$, and is equal to the set of  regions of 
  $\widetilde{\A}$.
\item $M_i^{\f}$ for $i\in\set{1,2}$ is the set of moves of player-$i$.
  \begin{itemize}
  \item
     $M_1^{\f} = \set{\tuple{\widetilde{R},a_1} \mid
      \widetilde{R} \text{ is a region of } \widetilde{\A}, \text{ and }
      a_1\in A_1^{\bot}}$.
  \item 
    $M_2^{\f} = \set{\tuple{\widetilde{R},a_2,i} \mid
      \widetilde{R} \text{ is a region of } \widetilde{\A}, i\in\set{1,2},
      \text{ and }
      a_2\in A_2^{\bot}}$.
  \end{itemize}
  Intuitively, the moves of player-$i$ denote which region it wants to let
  time pass to, and then take the discrete action $a_i^{\bot}$.
  In addition, for player~2, the ``$i$'' denotes which player's move will
  be chosen should the two players propose moves to the same region.
  Recall from Lemma~\ref{lemma:RegionsBeatRegions} that in such a case, it is
  up to the scheduler to decide which player's move to ``win'' in a run.
  Here, the scheduler is collaborating with player~2.
\item 
  $\Gamma_i^{\f}$ for $i\in\set{1,2}$ is the move assignment function.
  Given a state $\widetilde{R}\in S^{\f}$, we have 
  $\Gamma_i^{\f}(\widetilde{R})$ to be the set of moves available to player~$i$
  at state $\widetilde{R}$.
  \begin{itemize}
  \item 
    $\Gamma_1^{\f}(\widetilde{R}) = 
    \set{\tuple{\widetilde{R}',a_1} \mid 
      \exists\, \widetilde{s}\in \widetilde{R}
      \text{ such that player~1} \text{ has a move } \tuple{\Delta,a_1}
      \text{ in } \widetilde{\A} \text{ from } \widetilde{s} \text{ with }
      \reg(\widetilde{s} +\Delta) = \widetilde{R}'}$.
  \item  
    $\Gamma_2^{\f}(\widetilde{R}) = 
    \set{\tuple{\widetilde{R}',a_2,i} \mid 
      \exists\, \widetilde{s}\in \widetilde{R}
      \text{ such that player~2} \text{ has a move } \tuple{\Delta,a_2}
      \text{ in } \widetilde{\A} \text{ from } \widetilde{s} \text{ with }
      \reg(\widetilde{s} +\Delta) = \widetilde{R}' \text{ and } i\in\set{1,2}}$.
  \end{itemize}
\item 
  The transition function $\delta^{\f}$ is specified as
   $\delta^{\f}(\widetilde{R}, \tuple{\widetilde{R}_1,a_1}, 
   \tuple{\widetilde{R}_2,a_2,i}) = $
  \[
  \begin{cases}
    \widetilde{R}' & \text{if }\widetilde{R}_1 \neq \widetilde{R}_2, 
    \ \widetilde{R}_2 \text{ is a time successor of }\widetilde{R}_1,
    \text{ and } \exists\, \widetilde{s}_1\in  \widetilde{R}_1 
    \text{ such that }
    \delta^{\widetilde{\A}}(\widetilde{s}_1, \tuple{0,a_1})\in \widetilde{R}'\\
    \widetilde{R}' & \text{if }\widetilde{R}_1 \neq \widetilde{R}_2, 
    \ \widetilde{R}_1 \text{ is a time successor of }\widetilde{R}_2,
    \text{ and } \exists\, \widetilde{s}_2\in  \widetilde{R}_2 
    \text{ such that }
    \delta^{\widetilde{\A}}(\widetilde{s}_2, \tuple{0,a_2})\in \widetilde{R}'\\
    \widetilde{R}' & \text{if }\widetilde{R}_1 = \widetilde{R}_2, 
    \ i=1 \text{ and } \exists\, \widetilde{s}_1\in  \widetilde{R}_1 
    \text{ such that }
    \delta^{\widetilde{\A}}(\widetilde{s}_1, \tuple{0,a_1})\in \widetilde{R}'\\
    \widetilde{R}' & \text{if }\widetilde{R}_1 = \widetilde{R}_2, 
    \ i=2 \text{ and } \exists\, \widetilde{s}_2\in  \widetilde{R}_2 
    \text{ such that }
    \delta^{\widetilde{\A}}(\widetilde{s}_2, \tuple{0,a_2})\in \widetilde{R}'\\

  \end{cases}
  \]
    
\end{itemize}

Note that given player-1 and player-2 pure strategies $\pi_1^{\widetilde{\A}^{\f}}$
and  $\pi_2^{\widetilde{\A}^{\f}}$, and any state $\widetilde{R}$, we have only
one run in 
$\outcomes(\widetilde{R}, \pi_1^{\widetilde{\A}^{\f}}, \pi_2^{\widetilde{\A}^{\f}})$.

\smallskip\noindent\textbf{Mapping runs and states in $\widetilde{\A}$
to those in $\widetilde{\A}^{\f}$ using $\regmap()$ and  $\regstates()$.}
Given a run $\widetilde{r}= \widetilde{s}_0,\tuple{m_1^0,m_2^0}, 
\widetilde{s}_1,\tuple{m_1^1,m_2^1}, \dots$ of $\widetilde{\A}$,
we let $\regmap(\widetilde{r})$ be the corresponding run in 
$\widetilde{\A}^{\f}$ such that the states in $\widetilde{r}$ are mapped to
their regions, and the moves of  $\widetilde{\A}$ are mapped to 
corresponding moves in
$\widetilde{\A}^{\f}$.
Formally,  $\regmap(\widetilde{r})$ is the run
$\reg(\widetilde{s}_0),\tuple{m_1^{0,\f},m_2^{0,\f}}, 
\reg(\widetilde{s}_1),\tuple{m_1^{1,\f},m_2^{1,\f}}, \dots$ in 
$\widetilde{\A}^{\f}$ such that 
for $m_1^j=\tuple{\Delta_1^j,a_1^j}$ and $m_2^j=\tuple{\Delta_2^j,a_1^j}$ 
we have 
(1)~$m_1^{j,\f} = \tuple{\reg\left(\widetilde{s}_j+\Delta_1^j\right),a_1^j}$,
 and 
(2)~$m_2^{j,\f} = \tuple{\reg\left(\widetilde{s}_j+\Delta_2^j\right),a_2^j,i} $
with $i=1$ if $\Delta_1^j <\Delta_2^j$, or $\Delta_1^j =\Delta_2^j$ and
$\widetilde{s}_{j+1}=\delta(\widetilde{s}_j,m_1^j)$ (i.e., the scheduler picks
player~1 in round $j$); otherwise $i=2$.
%
Given a set of regions $X$ of $\widetilde{\A}$ (i.e., $X$ is a set of states of
$\widetilde{\A}^{\f}$),
let $\regstates(X) = \set{\widetilde{s} \mid \widetilde{s}\in \bigcup X}$.

We have the following lemma which states the equivalence of the games
$\widetilde{\A}^{\f}$ and $\widetilde{\A}$ with respect to the
$\CPre_1$ operator of the $\mu$-calculus formulation mentioned in
Section~\ref{section:Setting}.

\begin{lemma}
\label{lemma:FiniteConcurrentCPre}
Let $\A$ be a timed automaton game,  $\widetilde{\A}$ the expanded game 
structure as mentioned above, and $\widetilde{\A}^{\f}$ the corresponding
finite state concurrent game structure.
If $X$ is a set  of
regions of $\widetilde{\A}$, then 
$\CPre_1^{\widetilde{\A}}(\bigcup X) = 
\regstates\left(\CPre_1^{\widetilde{\A}^{\f}}(X)\right)$
\end{lemma}
\begin{proof}
The proof follows from Lemma~\ref{lemma:RegionsBeatRegions}.
\qed
\end{proof}

\begin{lemma}[Relating sure winning sets in 
$\bm{\widetilde{\A}^{\f}}$
 and   $\bm{\widetilde{\A}}$]
\label{lemma:FiniteConcurrentEquivalence}
Let $\A$ be a timed automaton game,  $\widetilde{\A}$ the expanded game 
structure as described  above, and $\widetilde{\A}^{\f}$ the corresponding
finite state concurrent game structure.
Let $\widetilde{\Phi}$ be an $\omega$-regular location objective of
$\widetilde{\A}$ (and naturally also of $\widetilde{\A}^{\f}$).
We have $\sureu_1^{\widetilde{\A}}(\widetilde{\Phi}) =
\regstates\left(\sureu_1^{\widetilde{\A}^{\f}}(\widetilde{\Phi})\right)$.

\end{lemma}
\begin{proof}
Only unions of regions arise in the $\mu$-calculus iteration for
computing winning sets in $\widetilde{\A}$ for $\omega$-regular
objectives.
The proof follows from the fact
that equivalent sets of states 
arise in the $\mu$-calculus iteration for computing the 
winning sets in both game structures due to  
Lemma~\ref{lemma:FiniteConcurrentCPre}.
Corollary~\ref{corollary:PureGame} gives us the equivalence between $\pureu_1$
and $\sureu_1$ sets.
\qed
\end{proof}

\subsubsection{Obtaining a Class of Spoiling Player-2 Spoiling Strategies in 
 $\widetilde{\A}$ Using  the Game Structure $\widetilde{\A}^{\f}$.}
We use the finite state game  $\widetilde{\A}^{\f}$ to analyze
the spoiling strategies of player~2 for any given player-1 strategy $\pi_1$
in $\widetilde{\A}$.
To do this analysis, we 
(1)~map any player-1 strategy $\pi_1$  in 
$\widetilde{\A}$ to a corresponding player-1 strategy $\pi_1^{\f}$ 
in $\widetilde{\A}^{\f}$; and 
(2)~map any player-2 spoiling strategies in $\widetilde{\A}^{\f}$ against
$\pi_1^{\f}$ to a class of player-2 spoiling strategies in $\widetilde{\A}$, all of which will
be spoiling against  $\pi_1$.

We first present the  next Lemma which states that for every run of 
$\widetilde{\A}^{\f}$, there
exists a run of $\widetilde{\A}$ that has an equivalent  region sequence.
\begin{lemma}
\label{lemma:SameRunsSafety}
Let $\A$ be a timed automaton game,  $\widetilde{\A}$ the expanded game 
structure as described previously, and $\widetilde{\A}^{\f}$ the corresponding
finite state concurrent game structure.
For every finite run $\widetilde{r}^{\,\f}$ of $\widetilde{\A}^{\f}$, there
exists a finite run of $\widetilde{r}$ of $\widetilde{\A}$ such that
$\regmap(\widetilde{r}) = \widetilde{r}^{\,\f}$.
\end{lemma}
\begin{proof}
Let $\widetilde{r}^{\,\f}$ be any given finite run of $\widetilde{\A}^{\f}$.
We show by induction on the number of steps in 
 $\widetilde{r}^{\,\f}$ that there exists
a finite run of $\widetilde{r}$ of $\widetilde{\A}$ such that
$\regmap(\widetilde{r}) = \widetilde{r}^{\,\f}$.
Let the inductive hypothesis be true for all runs with at most $j$ steps.
Let $\widetilde{r}^{\,\f}$ contain $j+1$ steps.
By inductive hypothesis, there exists a finite run $\widetilde{r}^{\, *}$
with $j$ steps such that $\regmap(\widetilde{r}^{\, *}) = 
\widetilde{r}^{\,\f}[0..j]$.


Let $ \widetilde{r}^{\,\f}[0..j+1]\ =\  \widetilde{r}^{\,\f}[0..j],
\tuple{\tuple{\widetilde{R}^j_1,a_1^j}, \tuple{\widetilde{R}^j_2,a_2^j,i}},
\widetilde{R}^j$.
Since $\reg\left(\widetilde{r}^{\, *}[j]\right)=  \widetilde{r}^{\,\f}[j]$, 
we have
by Lemma~\ref{lemma:RegionsBeatRegions}, and by the construction of
$\widetilde{\A}^{\f}$ that
(1)~there exists a player-$i$ move $\tuple{\Delta_i, a_i^j}$ from 
$\widetilde{r}^{\,*}[j]$ such that
$\reg(\widetilde{r}^{\,*}[j]+\Delta_i^j) = 
\widetilde{R}^j_i$ for $i\in\set{1,2}$, and
(2) for some 
$\widetilde{s}^{\,*}\in \delta_{\jd}(\widetilde{r}^{\,*}[j], 
\tuple{\Delta_1, a_1^j} 
\tuple{\Delta_2, a_2^j})$,
we have $\reg(\widetilde{s}^{\,*}) =  \widetilde{r}^{\,\f}[j+1]$.
Thus, the run $\widetilde{r}^{\,*}$ can be extended to $\widetilde{r}$ by
one more step such that $\widetilde{r}$ has the desired
properties.
\qed
\end{proof}



\smallskip\noindent\textbf{Mapping player-1 strategies in 
$\bm{\widetilde{\A}}$ to player-1 strategies in $\bm{\widetilde{\A}^{\f}}$.}
Let $\VRuns^{\widetilde{\A}^{\f}}$ be the set of finite 
runs of $\widetilde{\A}^{\f}$.
A set of finite runs $\runcover$ of $\widetilde{\A}$ is said to \emph{cover}
$\VRuns^{\widetilde{\A}^{\f}}$ if for every (finite) run $\widetilde{r}^{\,\f}\in
\VRuns^{\widetilde{\A}^{\f}}$, there exists a \emph{unique} finite  run
$\widetilde{r}\in\runcover $ such that 
$\regmap(\widetilde{r})= \widetilde{r}^{\,\f}$.
There exists at least one such run-cover $\runcover$
 by Lemma~\ref{lemma:SameRunsSafety}.
Abusing notation, we let $\runcover(\widetilde{r}^{\,\f})$ denote the
unique run $\widetilde{r}\in\runcover $ such that 
$\regmap(\widetilde{r})= \widetilde{r}^{\,\f}$.
Given a player-1 pure strategy $\pi_1$ in $\widetilde{\A}$, and a run-cover 
$\runcover$ of $\VRuns^{\widetilde{\A}^{\f}}$, we obtain the mapped player-1 pure
strategy in $\widetilde{\A}^{\f}$, denoted,  
 $\mbf^{\runcover}(\pi_1)$,  as follows.
\[
\left(\mbf^{\runcover}(\pi_1)\right)(\widetilde{r}^{\,\f}) =
\begin{cases}
  \tuple{\widetilde{R},a_1} &
  \text{such that } \pi_1\left(\runcover\left(\widetilde{r}^{\,\f}\right)\right)= 
  \tuple{\Delta_1,a_1}, \text{ and  }
  \reg\left(\,\runcover\left(\widetilde{r}^{\,\f}\right)[k]\, +\,\Delta_1\,
  \right) = 
  \widetilde{R}\\
  &\left(\text{where } \runcover\left(\widetilde{r}^{\,\f}\right)[k] 
    \text{ is the last state in }
    \runcover\left(\widetilde{r}^{\,\f}\right)\right)
\end{cases}
\]
Intuitively, the strategy $\mbf^{\runcover}(\pi_1)$, on the  finite run
$\widetilde{r}^{\,\f}$,
acts like $\pi_1$ on the finite run  $\runcover\left(\widetilde{r}\right)$ 
(i.e., the move is
to the same region, with the same discrete action). 

\smallskip\noindent\textbf{Mapping player-2  pure strategies in 
$\bm{\widetilde{\A}^{\f}}$ to player-2 pure strategies in 
$\bm{\widetilde{\A}}$.}
We now  map any given player-2 pure strategy $\pi_2^{\widetilde{\A}^{\f}}$ in 
$\widetilde{\A}^{\f}$ to  player-2  pure strategies in $\widetilde{\A}$.
This mapping will 
depend on a given player-1 pure strategy $\pi_1$ in  $\widetilde{\A}$ (the
strategy $\pi_1$ will be given as a parameter).
Given a player-2 pure strategy 
$\pi_2^{\widetilde{\A}^{\f}}$ in 
$\widetilde{\A}^{\f}$, and a player-1 pure strategy $\pi_1$ in  
$\widetilde{\A}$, we
define a \emph{set} of player-2  pure strategies in  $\widetilde{\A}$.
The set, denoted  as $\tset_{\pi_1}(\pi_2^{\widetilde{\A}^{\f}})$, is
defined as containing all player-2 pure strategies $\pi_2$ in $\widetilde{\A}$
satisfying the following condition: 
given any run prefix $\widetilde{r}[0..k]$ in $\widetilde{\A}$, 
with $\pi_1(\widetilde{r}[0..k]) = \tuple{\Delta_1,a_1}$, the
strategy $\pi_2$ satisfies Equation~\ref{equation:FiniteSpoilingStrategy}.
%
\begin{equation}
\label{equation:FiniteSpoilingStrategy}
\pi_2(\widetilde{r}[0..k]) =
\begin{cases}
  \tuple{\Delta_2,a_2} &\text{ such that }
  \Delta_2 < \Delta_1 \text{ and } \reg(\widetilde{r}[k]+\Delta_2)=
  \widetilde{R}_2; 
  \text{ if } \\
  &(1)~\pi_2^{\widetilde{\A}^{\f}}(\regmap(\widetilde{r}[0..k])) = 
  \tuple{\widetilde{R}_2,a_2,i},
  \text{ and }\\
  \vspace*{1mm} 
  &(2)~\reg(\widetilde{r}[k]+\Delta_1) \text{ is a time successor of }\widetilde{R}_2\\
\tuple{\Delta_2,a_2} &\text{ such that }
  \Delta_2 > \Delta_1 \text{ and } \reg(\widetilde{r}[k]+\Delta_2)=
  \widetilde{R}_2; 
  \text{ if } \\
  &(1)~\pi_2^{\widetilde{\A}^{\f}}(\regmap(\widetilde{r}[0..k])) = 
  \tuple{\widetilde{R}_2,a_2,i},
  \text{ and }\\
\vspace*{1mm} 
  &(2)~\widetilde{R}_2 \text{ is a time successor of }
  \reg(\widetilde{r}[k]+\Delta_1)\\
\tuple{\Delta_2,a_2} &\text{ such that }
  \Delta_2 \geq \Delta_1 \text{ and } \reg(\widetilde{r}[k]+\Delta_2)=
  \widetilde{R}_2; 
  \text{ if } \\
  &(1)~\pi_2^{\widetilde{\A}^{\f}}(\regmap(\widetilde{r}[0..k])) = 
  \tuple{\widetilde{R}_2,a_2,1},
  \text{ and }\\
  \vspace*{1mm} 
  &(2)~\reg(\widetilde{r}[k]+\Delta_1) =\widetilde{R}_2\\
\tuple{\Delta_2,a_2} &\text{ such that }
  \Delta_2 \leq \Delta_1 \text{ and } \reg(\widetilde{r}[k]+\Delta_2)=
  \widetilde{R}_2; 
  \text{ if } \\
  &(1)~\pi_2^{\widetilde{\A}^{\f}}(\regmap(\widetilde{r}[0..k])) = 
  \tuple{\widetilde{R}_2,a_2,2},
  \text{ and }\\
  &(2)~\reg(\widetilde{r}[k]+\Delta_1) =\widetilde{R}_2
\end{cases}
\end{equation}
  
Intuitively, a strategy $\pi_2$ in $\tset_{\pi_1}(\pi_2^{\widetilde{\A}^{\f}})$
picks a move of 
time duration bigger than that of 
$\pi_1$ if the strategy $\pi_2^{\widetilde{\A}^{\f}}$  in $\widetilde{\A}^{\f}$
 allows a corresponding player-1 move 
$\tuple{\reg(\widetilde{r}[k] +\Delta_1), a_1}$.
Otherwise, the  strategies $\pi_2$ pick a move of shorter duration.

\smallskip\noindent\textbf{Player-2 spoiling strategies set 
${\spoil^{\runcover}(\pi_1,\pi_2^{\widetilde{\A}^{\f}, \runcover, \pi_1} )}$
in 
$\bm{\widetilde{\A}}$.}
Given a player-1 pure strategy $\pi_1$ in $\widetilde{\A}$ such that
$\pi_1$ is not a winning player-1 strategy from a state $\widetilde{s}$
(for some
$\omega$-regular location objective $\widetilde{\Phi}$ of $\widetilde{\A}$),
we now obtain a 
specific \emph{set} of
player-2 spoiling pure strategies in $\widetilde{\A}$ against $\pi_1$ from $\widetilde{s}$ .
The set is  
denoted as $\spoil^{\runcover}(\pi_1,\pi_2^{\widetilde{\A}^{\f}, \runcover,\pi_1})$, where
$\runcover$ is a runcover of $\VRuns^{\widetilde{\A}^{\f}}$, and
$\pi_2^{\widetilde{\A}^{\f}, \runcover, \pi_1}$ is a  
given player-2 spoiling pure strategy
against $\mbf^{\runcover}\left(\pi_1\right)$ in $\widetilde{\A}^{\f}$ for the 
same objective $\widetilde{\Phi}$, for the starting state
$\reg(\widetilde{s})$.
We observe that some  player-2 spoiling pure strategy 
$\pi_2^{\widetilde{\A}^{\f}, \runcover, \pi_1}$
must exist by Lemma~\ref{lemma:FiniteConcurrentEquivalence} and 
Corollary~\ref{corollary:PureGame}.
The set  $\spoil^{\runcover}(\pi_1,\pi_2^{\widetilde{\A}^{\f}, \runcover, \pi_1} )$ 
of player-2 spoiling pure strategies for $\pi_1$ is  defined to be equal to
$\tset_{\pi_1}(\pi_2^{\widetilde{\A}^{\f}, \runcover, \pi_1})$.

The next Lemma relates spoiling player-2 strategies in 
$\widetilde{\A}^{\f}$
 and   $\widetilde{\A}$ (the proof is by an involved induction argument).
The intuition behind the Lemma is that given a state 
$\widetilde{s}\notin 
\win_1^{\widetilde{\A}}(\widetilde{\Phi})$,  we have that 
(a)~$\reg(\widetilde{s}) \notin 
\win_1^{\widetilde{\A}^{\f}}(\widetilde{\Phi})$; and 
(b)~player-2 can obtain
spoiling strategies for any player-1 strategy $\pi_1$ in
 $\widetilde{\A}$ by prescribing moves to the same regions
as the player-2 spoiling strategy in $\widetilde{\A}^{\f}$, which spoils
$\mbf^{\runcover}\left(\pi_1\right)$ (for some suitably chosen $\runcover$).
This result will be used in the next subsection to show that receptive player-1
strategies \emph{must} satisfy certain requirements.

\begin{lemma}[Relating spoiling player-2 pure strategies in 
$\bm{\widetilde{\A}^{\f}}$
 and   $\bm{\widetilde{\A}}$]
\label{lemma:FiniteSafetyGame}
Let $\A$ be a timed automaton game, $\widetilde{\A}$ the 
expanded game structure, and $\widetilde{\A}^{\f}$ the corresponding
finite state concurrent game structure.
Given an $\omega$-regular location objective $\widetilde{\Phi}$ 
of player~1 (in $\widetilde{\A}$ and $\widetilde{\A}^{\f}$),  
the following assertions  hold.
\begin{enumerate}
\item 
  $\widetilde{s}\in 
  \pureu_1^{\widetilde{\A}}(\widetilde{\Phi})$ iff
  $\reg(\widetilde{s}) \in 
  \pureu_1^{\widetilde{\A}^{\f}}(\widetilde{\Phi})$.
\item 
  Let $\widetilde{s}\notin 
  \pureu_1^{\widetilde{\A}}(\widetilde{\Phi})$.
  Given any player-1 strategy $\pi_1$ in $\widetilde{\A}$ there exists a 
  runcover
  $\runcover$ of $\VRuns^{\widetilde{\A}^{\f}}$ such that for 
  any player-2 pure 
  spoiling strategy 
  $\pi_2^{\widetilde{\A}^{\f},  \runcover, \pi_1}$  
  against 
  $\mbf^{\runcover}\left(\pi_1\right)$ in $\widetilde{\A}^{\f}$ from the state
  $\reg(\widetilde{s})$ for the
  objective $\widetilde{\Phi}$ (such spoiling strategies exist by the previous
  part of the lemma);
  we have that
  every player-2 strategy in 
  $\spoil^{\runcover}(\pi_1, \pi_2^{\widetilde{\A}^{\f}, \runcover, \pi_1})$ is
  a spoiling strategy against $\pi_1$ in the structure  $\widetilde{\A}$
  for the objective $\widetilde{\Phi}$ from the state $\widetilde{s}$.
\end{enumerate}

\end{lemma}

\begin{proof}
\begin{enumerate}
\item 
  Only unions of regions arise in the $\mu$-calculus iteration to obtain
  winning sets of player~1 for the objective $\widetilde{\Phi}$ in
  the game structure $\widetilde{\A}$.
  Using Lemma~\ref{lemma:FiniteConcurrentCPre} in the $\mu$-calculus
  iteration for  obataining the player-1 winning set for $\widetilde{\Phi}$, 
  we deduce that
  $\widetilde{s} \in \pureu_1^{\widetilde{\A}}(\widetilde{\Phi})$ iff
  $\reg(\widetilde{s}) \in \pureu_1^{\widetilde{\A}^{\f}}(\widetilde{\Phi})$.
\item
  By the first part of the lemma, we have that 
  $\reg(\widetilde{s}) \notin \pureu_1^{\widetilde{\A}^{\f}}(\widetilde{\Phi})$.
  Thus, given any runcover $\runcover$, there exists a pure player-2 spoiling strategy
  $\pi_2^{\widetilde{\A}^{\f},  \runcover, \pi_1}$  
  against 
  $\mbf^{\runcover}\left(\pi_1\right)$ in $\widetilde{\A}^{\f}$ from the state
  $\reg(\widetilde{s})$ for the
  objective $\widetilde{\Phi}$.
%

  We show that
  there exists  
  a runcover $\runcover$ of $\VRuns^{\widetilde{\A}^{\f}}$ such that given
  any pure player-2 strategy $\pi_2^{\widetilde{\A}^{\f}, \runcover, \pi_1}$ which
  spoils $\mbf^{\runcover}\left(\pi_1\right)$ in $\widetilde{\A}^{\f}$ from
   winning the objective $\widetilde{\Phi}$
   starting from   $\reg(\widetilde{s})$, and given any
   player-2 strategy $\pi_2$ from 
  $\spoil^{\runcover}(\pi_1,\pi_2^{\widetilde{\A}^{\f}, \runcover, \pi_1} )$ in
  $\widetilde{\A}$, there exists a 
  run $\widetilde{r}^{\,*}\in \outcomes(\widetilde{s},\pi_1,\pi_2)$ in
  $\widetilde{\A}$,
  such that the region sequence of $\widetilde{r}^{\,*} $ is the 
  same as the sequence
  of regions in the (only) run from
  $\outcomes\big(\reg(\widetilde{s}), \mbf^{\runcover}\left(\pi_1\right),
  \pi_2^{\widetilde{\A}^{\f}, \runcover, \pi_1}\big)$.
  This proves the Lemma due to the following:
  since $ \pi_2^{\widetilde{\A}^{\f}, \runcover, \pi_1}$ is a player-2 spoiling strategy against
  $\mbf^{\runcover}\left(\pi_1\right)$,
  we must have that $\reg(\widetilde{r}^{\,*})$ satisfies $\neg\widetilde{\Phi}$,
  and hence $\widetilde{r}^{\,*}$ satisfies $\neg\widetilde{\Phi}$ implying 
  $\pi_2$ to be a spoiling strategy of player~2 in $\widetilde{\A}$ against 
  $\pi_1$. 
%
%
  The proof of the statement is by an involved induction.
  \qed

\end{enumerate}

\end{proof}

\subsection{Characterizing Receptive Strategies Without Using Extra 
Clocks}
We now present characterizations of receptive strategies in timed
automaton games, and show that receptiveness can be expressed as an \LTL
condition on the states of $\widetilde{\A}$, from which it follows that
receptive strategies require finite memory in timed automaton games.
First, we consider the case where all clocks are bounded in the game (i.e.,
location invariants of the form $\bigwedge_{x\in C} x \leq d_x$ can be put 
on all locations).

\begin{lemma}[Receptive strategies when all clocks bounded in $\bm{\A}$]
\label{lemma:ReceptiveBounded}
Let $\A$ be a timed automaton game in which all clocks are bounded (i.e., 
for all clocks $x$ we have $x \leq d_x$, for  constants $d_x$ in all reachable
states).
Let $\widetilde{\A}$ be the enlarged game structure obtained from $\A$.
Then player~1 has a receptive  strategy from a state $s$ of $\A$ iff 
$\tuple{s,\cdot}\in \sureu_1^{\widetilde{\A}}(\Phi)$, where
\[ \Phi = 
\Box\Diamond(bl_1=\true) \rightarrow
\left (
\left(\bigwedge_{x\in C}\Box\Diamond (x=0)\right) \wedge 
\left ( 
\begin{array}{c}
\Box\Diamond\, 
  (\bl_1=\true)\, \wedge\,\bigwedge_{x\in C} (V_{>0}(x)=\true)  \\
\vee \\
\Box\Diamond 
  (\bl_1=\false)\, \wedge\,  \bigvee_{x\in C} (V_{\geq 1}(x)=\true)  
\end{array}
\right )
\right ).
\]
\end{lemma}
\begin{proof}
We prove inclusion in both directions.
\begin{enumerate}
\item ($\Leftarrow$). 
For a state $\widetilde{s}\in \sureu_1^{\widetilde{\A}}(\Phi)$, we
show that player~1 has a receptive strategy from $\widetilde{s}$.
Let $\pi_1$ be a pure sure winning region strategy: 
since $\Phi$ is an
$\omega$-regular region objective such a strategy exists by 
Lemma~\ref{lemma:RegionStrategies}.
Consider a strategy $\pi_1'$ for player~1 that is region-equivalent to $\pi_1$ 
such that whenever  the strategy 
$\pi_1$ proposes a move $\tuple{\Delta, a_1}$ for any  run prefix 
$\widetilde{r}[0..k]$ with
$\widetilde{r}[k]+\Delta$  satisfying  $\bigwedge_{x\in C}(x >0) $, then 
$\pi_1'$ proposes the move $\tuple{\Delta', a_1}$ for $\widetilde{r}[0..k]$
such that
$\reg(\widetilde{r}[k]+\Delta)=\reg(\widetilde{r}[k]+\Delta')$ and
$\widetilde{r}[k]+\Delta'$ satisfies
$(\vee_{y\in C}\ y>1/2)\,\wedge\, \bigwedge_{x\in C}(x >0)$.
Such a move always exists; in particular, for any state $\widetilde{s}$,
 if there exists $\Delta$ such that
 $\widetilde{s} + \Delta \in R \subseteq \bigwedge_{x\in C}(x >0) $, 
then there exists $\Delta'$ such that $\widetilde{s} + \Delta' \in R\, \cap
\left ((\vee_{y\in C}\ y>1/2)\, \wedge\,\bigwedge_{x\in C}(x >0) \right)$.
Intuitively, player~1 jumps near the boundary of $R$.
By Lemma~\ref{lemma:RegionStrategies}, $\pi_1'$ is also sure-winning for $\Phi$.
The strategy $\pi_1'$ ensures that in all resulting runs, if player~1 
is not blameless, then all clocks are 0 infinitely often (since for all
clocks $\Box \Diamond (x=0)$), and that some clock has value more than $1/2$
 infinitely
often (either due to player~1 ensuring some clock being greater than 1/2 
infinitely often; or player~2 playing moves which result in some clock being
greater than 1 infinitely often).. 
This implies time divergence.
Hence player~1 has a receptive winning strategy from $\widetilde{s}$.

\item ($\Rightarrow$).
  For  a state $\widetilde{s}\notin \sureu_1^{\widetilde{\A}}(\Phi)$, 
  we show that player~1 does not have any receptive strategy starting from state
  $\widetilde{s}$.
  We have $\neg\Phi \equiv
  \left(\Box\Diamond (bl_1=\true)\right) \wedge 
  \left(\neg\Psi_1 \vee (\neg\Psi_2\wedge \neg\Psi_2)\right)$, where
  \begin{eqnarray*}
    \neg\Psi_1^{\dagger} & = & \bigvee_{x\in C}\Diamond\Box (x>0)\\
    \neg\Psi_2^{\dagger} & = &  \Diamond\Box\,\left( (\bl_1=\true) \rightarrow 
      \left(\bigvee_{x\in C} \left(V_{>0}(x)=\false\right)\right)\right) \\
    \neg \Psi_3^{\dagger} & = & \Diamond\Box\,\left( (\bl_1=\false) \rightarrow 
      \left(\bigwedge_{x\in C} (V_{\geq 1}(x)=\false) \right)\right)
  \end{eqnarray*}


Recall the finite state game $\widetilde{\A}^{\f}$ based on the regions of
$\widetilde{\A}$.
Suppose $\widetilde{s}\notin\sureu_1^{\widetilde{\A}}(\Phi)$.
Then $\widetilde{s}\notin\pureu_1^{\widetilde{\A}}(\Phi)$ by 
Corollary~\ref{corollary:PureGame}.
Consider any pure player-1 strategy $\pi_1$ in $\widetilde{\A}$.
By Lemma~\ref{lemma:FiniteSafetyGame}, 
$\reg(\widetilde{s})\notin\pureu_1^{\widetilde{\A}^{\f}}(\Phi)$, and
there exists a runcover
$\runcover$ for $\VRuns^{\widetilde{\A}^{\f}}$
such that for   
any player-2 pure 
spoiling strategy $\pi_2^{\widetilde{\A}^{\f}}$ against 
$\mbf^{\runcover}\left(\pi_1\right)$ in $\widetilde{\A}^{\f}$ from
$\reg(\widetilde{s})$,
we have that
every player-2 strategy in 
$\spoil^{\runcover}(\pi_1, \pi_2^{\widetilde{\A}^{\f}})$ is
a spoiling strategy against $\pi_1$ in the structure  $\widetilde{\A}$.

Let $\runcover$ be such a runcover, and let 
$\pi_2^{\widetilde{\A}^{\f}}$ be any such player-2 strategy against 
$\mbf^{\runcover}\left(\pi_1\right)$ in $\widetilde{\A}^{\f}$ from
$\reg(\widetilde{s})$.
We show that with an appropriately chosen $\pi_2$ in 
$\spoil^{\runcover}(\pi_1, \pi_2^{\widetilde{\A}^{\f}})$,
player~2 can ensure that
in one of the resulting runs, player~1 is not blameless, and time converges,
 and hence player~1 does not have a receptive pure strategy in 
$\widetilde{\A}$.
The result follows from observing that if player~1 does not have a pure 
receptive strategy, then it does not have a (possibly randomized) receptive
strategy (as a randomized strategy may be viewed as a random choice over pure
strategies).

Consider runs $\widetilde{r} \in \outcomes(\widetilde{s},\pi_1,\pi_2)$ for
$\pi_2\in \spoil^{\runcover}(\pi_1, \pi_2^{\widetilde{\A}^{\f}})$.
One of the runs must satisfy $\neg\Phi$, which can happen in
 one of the following ways.
 \begin{enumerate}
 \item $(\Box\Diamond(\bl_1=\true))\, \wedge\, \neg\Psi_1$.
   The condition $\neg\Psi_1^{\dagger}$ means that there is some clock $x$
   which eventually stays strictly greater than 0.
   Since all clocks are bounded, this condition means that the run is time
   convergent, and player~1 is not blameless.
 \item 
   $(\Box\Diamond(\bl_1=\true))\, \wedge\,  \neg\Psi_2 \wedge
  \neg\Psi_3 $.
  The clause $\neg\Psi_2$ means that eventually
  if an action of player~1 is chosen, then for  some  clock $x$, 
  the value of $x$ stays at 0 throughout the move
  (which means that the  move of player-1 is of duration $0$).
  This clause $\neg\Psi_3$ means that eventually
  if an action of player~2 is chosen, then 
  for every  clock $x$, the value of $x$ is strictly less than 
  $1$ during the move.

  Player~2 can have a strategy which takes moves smaller than $1/2^j$ during
  the $j$-th visit to a region $\widetilde{R}$ in which  every clock $x$  has
  value less than 1.
  We formalize the above statement.
  The strategy $\pi_2^{\widetilde{\A}^{\f}}$ spoils 
  $\mbf^{\runcover}\left(\pi_1\right)$ from winning in $\widetilde{\A}^{\f}$
  for the objective $\Phi$.
  Given a run prefix $\widetilde{r}[0..k]$ of $\widetilde{\A}$, let 
  $\pi_1(\widetilde{r}[0..k]) = \tuple{\Delta_1,a_1}$.
  Consider  a player-2 strategy $\pi_2$ in 
  $\spoil^{\runcover}(\pi_1, \pi_2^{\widetilde{\A}^{\f}})$, and
  let $\pi_2^{\widetilde{\A}^{\f}}(\regmap(\widetilde{r}[0..k])) = 
  \tuple{\widetilde{R}_2,a_2,i}$.
  Let $\pi_2$ be a strategy in 
  $\spoil^{\runcover}(\pi_1, \pi_2^{\widetilde{\A}^{\f}})$
  such  that for
  $\pi_2((\widetilde{r}[0..k]) = \tuple{\Delta_2,a_2}$ we have
  $\Delta_2 \leq \Delta_1$ and $\Delta_2 < 1/2^k$ 
  whenever the following conditions hold.
  \begin{enumerate}
  \item 
    For every  clock $x$, the value of $x$ is strictly less than 
    $1$ in $\widetilde{R}_2$.
  \item 
    Either
    \begin{enumerate}
    \item 
      $\widetilde{R}_2$ is a region predecessor of
      $\reg(\widetilde{r}[k] +\Delta_1)$; or
    \item 
      $i=2$ and  $\reg(\widetilde{r}[k] +\Delta_1) = \widetilde{R}_2$.
    \end{enumerate}
  \end{enumerate}
  It can be observed from Equation~\ref{equation:FiniteSpoilingStrategy}
  that such
  a $\Delta_2$ and such a strategy $\pi_2$ in
  $\spoil^{\runcover}(\pi_1, \pi_2^{\widetilde{\A}^{\f}})$ 
  always exist.
  The above condition ensures that if a move of player~2 is chosen to
  a region $\widetilde{R}$ in which  every clock $x$ has
  value less than 1, then the 
  moves are smaller than $1/2^j$ during
  the $j$-th stage of the game.
  The strategy $\pi_2$ is a spoiling strategy against $\pi_1$ by 
  Lemma~\ref{lemma:FiniteSafetyGame} as $\pi_2$ is in
  $\spoil^{\runcover}(\pi_1, \pi_2^{\widetilde{\A}^{\f}})$.
  Moreover, this strategy ensures that at  least one of the resulting
  runs  $\widetilde{r}$ satisfies 
  $\neg\Phi$.
  \begin{enumerate}
  \item 
    If $\widetilde{r}$ satisfies 
    $(\Box\Diamond(\bl_1=\true))\, \wedge\, \neg\Psi_1$, then
    the run is time
    convergent, and player~1 is not blameless.
  \item 
    If   $\widetilde{r}$ satisfies 
  $(\Box\Diamond(\bl_1=\true))\, \wedge\,  \neg\Psi_2 \wedge
  \neg\Psi_3 $, then we have that:
  \begin{enumerate}
  \item 
    Eventually, every chosen move of player~2 results
    in a region $\widetilde{R}$ in which  every clock $x$ has
    value less than 1, with the duration of the player-2 move being
    smaller than $1/2^j$ during
    the $j$-th stage of the game; and
  \item 
    Eventually every chosen move of player~1 is of time duration 0.
  \end{enumerate}
  Thus, time is convergent in the run $\widetilde{r}$ and player~1 is not
  blameless.
\end{enumerate}
\end{enumerate}
Hence, 
player~1 does not have a pure receptive strategy from $\widetilde{s}$ 
(from which
it follows that it does not have any receptive  strategy from $\widetilde{s}$).
 \qed

\end{enumerate}
\end{proof}

We next present a couple of examples to demonstrate the role of the
various clauses in the the formula $\Phi$ of Lemma~\ref{lemma:ReceptiveBounded}.
\begin{example}
Consider the timed automaton game in Figure~\ref{figure:ExampleReceptiveOne}.
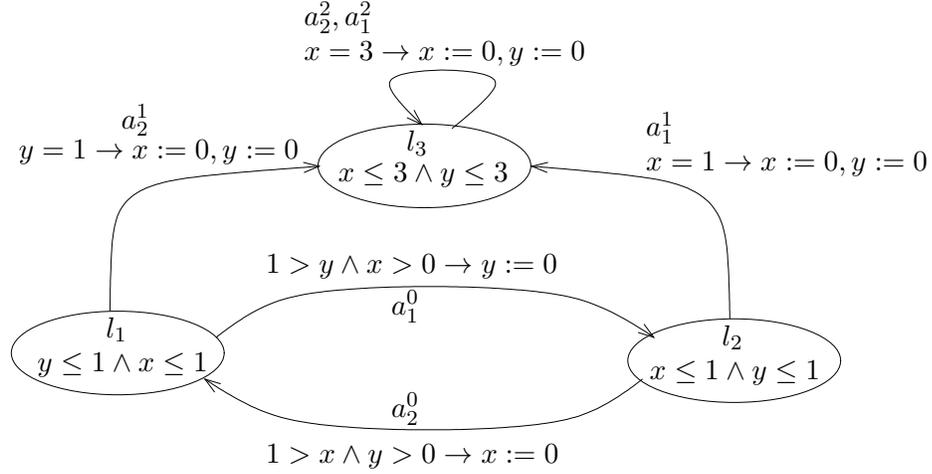
\begin{figure}[t]
      \strut\centerline{\setlength{\unitlength}{0.00043745in}
\begingroup\makeatletter\ifx\SetFigFontNFSS\undefined%
\gdef\SetFigFontNFSS#1#2#3#4#5{%
  \reset@font\fontsize{#1}{#2pt}%
  \fontfamily{#3}\fontseries{#4}\fontshape{#5}%
  \selectfont}%
\fi\endgroup%
{\renewcommand{\dashlinestretch}{30}
\begin{picture}(9834,5559)(0,-10)
\put(1268,1414){\ellipse{2520}{990}}
\put(8566,1321){\ellipse{2520}{990}}
\put(4900,3629){\ellipse{2520}{990}}
\path(5227,4074)(5230,4076)(5235,4080)
	(5245,4088)(5261,4101)(5282,4118)
	(5308,4139)(5339,4165)(5375,4195)
	(5413,4227)(5452,4262)(5492,4297)
	(5531,4332)(5568,4367)(5602,4401)
	(5633,4433)(5661,4463)(5684,4491)
	(5704,4517)(5719,4541)(5730,4563)
	(5737,4583)(5739,4602)(5738,4619)
	(5732,4634)(5723,4649)(5709,4662)
	(5692,4674)(5674,4684)(5653,4693)
	(5630,4702)(5604,4711)(5575,4718)
	(5544,4725)(5510,4732)(5475,4738)
	(5437,4743)(5398,4748)(5356,4752)
	(5313,4755)(5269,4758)(5225,4760)
	(5179,4761)(5133,4762)(5087,4762)
	(5041,4761)(4995,4759)(4951,4757)
	(4907,4754)(4865,4750)(4824,4746)
	(4785,4741)(4748,4736)(4713,4730)
	(4680,4723)(4650,4716)(4621,4708)
	(4596,4700)(4573,4691)(4552,4682)
	(4530,4669)(4513,4656)(4499,4641)
	(4490,4625)(4485,4607)(4485,4588)
	(4489,4567)(4497,4543)(4510,4518)
	(4526,4490)(4547,4460)(4571,4428)
	(4599,4394)(4630,4359)(4662,4323)
	(4696,4287)(4729,4253)(4761,4221)
	(4789,4193)(4814,4169)(4834,4150)(4867,4119)
\path(4694.727,4198.511)(4867.000,4119.000)(4776.888,4285.972)
\path(2437,1599)(2438,1600)(2441,1602)
	(2446,1606)(2454,1613)(2465,1622)
	(2480,1633)(2499,1648)(2522,1665)
	(2548,1685)(2578,1706)(2610,1730)
	(2645,1754)(2682,1779)(2721,1805)
	(2761,1830)(2802,1855)(2843,1880)
	(2885,1903)(2927,1926)(2969,1947)
	(3011,1967)(3054,1985)(3097,2003)
	(3140,2019)(3184,2035)(3230,2049)
	(3276,2062)(3324,2074)(3373,2085)
	(3424,2096)(3478,2106)(3534,2115)
	(3592,2124)(3631,2129)(3672,2135)
	(3713,2140)(3756,2145)(3800,2149)
	(3845,2154)(3891,2158)(3939,2162)
	(3988,2166)(4037,2170)(4088,2173)
	(4141,2177)(4194,2180)(4248,2183)
	(4303,2185)(4360,2188)(4417,2190)
	(4475,2192)(4534,2194)(4593,2196)
	(4654,2197)(4714,2198)(4775,2199)
	(4837,2199)(4898,2200)(4960,2200)
	(5022,2200)(5084,2199)(5145,2199)
	(5206,2198)(5267,2197)(5328,2196)
	(5388,2194)(5447,2192)(5505,2190)
	(5563,2188)(5620,2185)(5676,2183)
	(5731,2180)(5785,2177)(5838,2173)
	(5889,2170)(5940,2166)(5989,2162)
	(6038,2158)(6085,2154)(6131,2149)
	(6176,2145)(6220,2140)(6262,2135)
	(6304,2129)(6345,2124)(6401,2116)
	(6456,2107)(6509,2098)(6560,2088)
	(6609,2077)(6657,2066)(6704,2054)
	(6750,2041)(6795,2027)(6840,2012)
	(6884,1996)(6928,1979)(6972,1961)
	(7016,1942)(7060,1922)(7104,1901)
	(7148,1878)(7192,1855)(7236,1832)
	(7279,1808)(7321,1784)(7361,1760)
	(7400,1736)(7436,1714)(7470,1693)
	(7500,1674)(7527,1656)(7550,1641)
	(7569,1629)(7584,1618)(7595,1611)(7612,1599)
\path(7430.345,1653.785)(7612.000,1599.000)(7499.547,1751.821)
\path(7477,1104)(7476,1103)(7473,1101)
	(7468,1097)(7460,1090)(7449,1081)
	(7434,1070)(7415,1055)(7392,1038)
	(7366,1018)(7336,997)(7304,973)
	(7269,949)(7232,924)(7193,898)
	(7153,873)(7112,848)(7071,823)
	(7029,800)(6987,777)(6945,756)
	(6903,736)(6860,718)(6817,700)
	(6774,684)(6730,668)(6684,654)
	(6638,641)(6590,629)(6541,618)
	(6490,607)(6436,597)(6380,588)
	(6322,579)(6283,574)(6242,568)
	(6201,563)(6158,558)(6114,554)
	(6069,549)(6023,545)(5975,541)
	(5926,537)(5877,533)(5826,530)
	(5773,526)(5720,523)(5666,520)
	(5611,518)(5554,515)(5497,513)
	(5439,511)(5380,509)(5321,507)
	(5260,506)(5200,505)(5139,504)
	(5077,504)(5016,503)(4954,503)
	(4892,503)(4830,504)(4769,504)
	(4708,505)(4647,506)(4586,507)
	(4526,509)(4467,511)(4409,513)
	(4351,515)(4294,518)(4238,520)
	(4183,523)(4129,526)(4076,530)
	(4025,533)(3974,537)(3925,541)
	(3876,545)(3829,549)(3783,554)
	(3738,558)(3694,563)(3652,568)
	(3610,574)(3570,579)(3513,587)
	(3458,596)(3405,605)(3354,615)
	(3305,626)(3257,637)(3210,649)
	(3164,662)(3119,676)(3074,691)
	(3030,707)(2986,724)(2942,742)
	(2898,761)(2854,781)(2810,802)
	(2766,825)(2722,848)(2678,871)
	(2635,895)(2593,919)(2553,943)
	(2514,967)(2478,989)(2444,1010)
	(2414,1029)(2387,1047)(2364,1062)
	(2345,1074)(2330,1085)(2319,1092)(2302,1104)
\path(2483.655,1049.215)(2302.000,1104.000)(2414.453,951.179)
\path(1177,1914)(1177,1916)(1177,1922)
	(1177,1932)(1177,1948)(1178,1970)
	(1178,1998)(1179,2032)(1179,2072)
	(1180,2116)(1181,2164)(1182,2214)
	(1183,2265)(1185,2317)(1186,2368)
	(1188,2418)(1189,2465)(1191,2511)
	(1193,2554)(1196,2595)(1198,2634)
	(1200,2670)(1203,2703)(1206,2735)
	(1209,2765)(1213,2794)(1216,2821)
	(1221,2847)(1225,2872)(1230,2896)
	(1234,2920)(1240,2944)(1246,2967)
	(1252,2989)(1259,3012)(1267,3034)
	(1276,3056)(1286,3077)(1296,3099)
	(1308,3120)(1321,3141)(1335,3161)
	(1350,3181)(1367,3201)(1384,3220)
	(1403,3238)(1424,3256)(1445,3273)
	(1468,3290)(1492,3306)(1517,3321)
	(1543,3336)(1571,3349)(1600,3363)
	(1630,3375)(1661,3387)(1693,3398)
	(1727,3409)(1763,3419)(1800,3429)
	(1827,3436)(1856,3442)(1885,3449)
	(1916,3455)(1949,3461)(1983,3467)
	(2019,3473)(2056,3479)(2096,3485)
	(2138,3490)(2182,3496)(2229,3502)
	(2278,3508)(2329,3513)(2384,3519)
	(2441,3525)(2501,3531)(2563,3537)
	(2628,3543)(2695,3549)(2764,3555)
	(2834,3561)(2905,3567)(2977,3573)
	(3049,3578)(3119,3584)(3188,3589)
	(3255,3595)(3317,3599)(3376,3604)
	(3430,3608)(3478,3611)(3520,3614)
	(3556,3617)(3585,3619)(3608,3621)
	(3626,3622)(3652,3624)
\path(3477.132,3550.371)(3652.000,3624.000)(3467.928,3670.018)
\path(8512,1824)(8512,1826)(8512,1831)
	(8512,1840)(8512,1855)(8511,1875)
	(8511,1901)(8510,1932)(8509,1970)
	(8508,2012)(8507,2058)(8506,2107)
	(8505,2158)(8503,2211)(8501,2263)
	(8499,2315)(8497,2365)(8495,2413)
	(8492,2460)(8489,2504)(8487,2546)
	(8483,2586)(8480,2623)(8477,2658)
	(8473,2691)(8469,2722)(8464,2752)
	(8460,2780)(8455,2807)(8449,2832)
	(8443,2857)(8437,2881)(8430,2907)
	(8422,2931)(8414,2955)(8405,2979)
	(8395,3002)(8385,3026)(8373,3048)
	(8361,3071)(8347,3093)(8333,3115)
	(8317,3136)(8300,3157)(8283,3178)
	(8264,3198)(8243,3218)(8222,3236)
	(8200,3255)(8176,3272)(8152,3289)
	(8127,3305)(8100,3321)(8073,3335)
	(8045,3349)(8016,3363)(7986,3375)
	(7955,3387)(7924,3398)(7891,3409)
	(7857,3419)(7822,3429)(7795,3436)
	(7767,3443)(7737,3450)(7707,3456)
	(7675,3463)(7642,3469)(7607,3475)
	(7571,3481)(7532,3487)(7492,3493)
	(7449,3499)(7404,3505)(7356,3511)
	(7306,3517)(7253,3524)(7198,3530)
	(7140,3536)(7080,3542)(7018,3549)
	(6954,3555)(6889,3561)(6823,3567)
	(6756,3574)(6690,3580)(6625,3585)
	(6562,3591)(6502,3596)(6445,3601)
	(6393,3606)(6346,3610)(6304,3613)
	(6269,3616)(6239,3619)(6216,3620)
	(6198,3622)(6172,3624)
\path(6356.072,3670.018)(6172.000,3624.000)(6346.868,3550.371)
\put(1132,1599){\makebox(0,0)[lb]{\smash{{\SetFigFontNFSS{11}{13.2}{\familydefault}{\mddefault}{\updefault}$l_1$}}}}
\put(8422,1509){\makebox(0,0)[lb]{\smash{{\SetFigFontNFSS{11}{13.2}{\familydefault}{\mddefault}{\updefault}$l_2$}}}}
\put(3022,114){\makebox(0,0)[lb]{\smash{{\SetFigFontNFSS{11}{13.2}{\familydefault}{\mddefault}{\updefault}$1>x\wedge y>0 \rightarrow x:=0$}}}}
\put(3022,2364){\makebox(0,0)[lb]{\smash{{\SetFigFontNFSS{11}{13.2}{\familydefault}{\mddefault}{\updefault}$1>y\wedge x>0 \rightarrow y:=0$}}}}
\put(4507,654){\makebox(0,0)[lb]{\smash{{\SetFigFontNFSS{11}{13.2}{\familydefault}{\mddefault}{\updefault}$a_2^0$}}}}
\put(4507,1869){\makebox(0,0)[lb]{\smash{{\SetFigFontNFSS{11}{13.2}{\familydefault}{\mddefault}{\updefault}$a_1^0$}}}}
\put(4687,3804){\makebox(0,0)[lb]{\smash{{\SetFigFontNFSS{11}{13.2}{\familydefault}{\mddefault}{\updefault}$l_3$}}}}
\put(3877,3444){\makebox(0,0)[lb]{\smash{{\SetFigFontNFSS{11}{13.2}{\familydefault}{\mddefault}{\updefault}$x\leq 3\wedge y\leq 3$}}}}
\put(322,1194){\makebox(0,0)[lb]{\smash{{\SetFigFontNFSS{11}{13.2}{\familydefault}{\mddefault}{\updefault}$y\leq1 \wedge x\leq 1$}}}}
\put(7567,1104){\makebox(0,0)[lb]{\smash{{\SetFigFontNFSS{11}{13.2}{\familydefault}{\mddefault}{\updefault}$x\leq1\wedge y\leq 1$}}}}
\put(97,3714){\makebox(0,0)[lb]{\smash{{\SetFigFontNFSS{11}{13.2}{\familydefault}{\mddefault}{\updefault}$y=1\rightarrow x:=0, y:=0$}}}}
\put(7522,3579){\makebox(0,0)[lb]{\smash{{\SetFigFontNFSS{11}{13.2}{\familydefault}{\mddefault}{\updefault}$x=1\rightarrow x:=0, y:=0$}}}}
\put(3472,4884){\makebox(0,0)[lb]{\smash{{\SetFigFontNFSS{11}{13.2}{\familydefault}{\mddefault}{\updefault}$x=3 \rightarrow x:=0, y:=0$}}}}
\put(3472,5289){\makebox(0,0)[lb]{\smash{{\SetFigFontNFSS{11}{13.2}{\familydefault}{\mddefault}{\updefault}$a_2^2, a_1^2$}}}}
\put(1312,4074){\makebox(0,0)[lb]{\smash{{\SetFigFontNFSS{11}{13.2}{\familydefault}{\mddefault}{\updefault}$a_2^1$}}}}
\put(7522,3984){\makebox(0,0)[lb]{\smash{{\SetFigFontNFSS{11}{13.2}{\familydefault}{\mddefault}{\updefault}$a_1^1$}}}}
\end{picture}
}}
      \caption{A time automaton game $\A_1$ with player-1 receptive strategies.}
      \label{figure:ExampleReceptiveOne}
    \end{figure}
The edges $a_1^j$ are player-1 edges and $a_2^j$ player-2 edges.
The edges $a_2^2$ and $a_1^2$ have the same guards and reset maps.
It is clear that player~1 has a receptive strategy when at location $l_3$; it
repeatedly takes (or tries to take) the edge $a_1^2$.
Let us hence focus our attention on plays which consist of $(l_1, l_2)$ 
cycles (i.e., player~2 picks the edge $a_2^0$ from location $l_2$, and allows
player~1 to take the edge $a_1^0$ from location $l_1$).
Let the starting state satisfy $(x <1) \wedge (y <1)$.
In a run which consists of $(l_1, l_2)$ cycles, we have that
(1)~both clocks are reset infinitely often, and
(2)~both clocks are greater than 0 infinitely often when the
edge $a_1^0$ is taken (this is because the condition on the edge
$a_2^0$ ensures that clock $y$ is greater than 0 when at location
$l_1$, and the edge condition on $a_1^0$ further ensure $x > 0$ when
edge $a_1^0$ is taken). Thus, a run of $(l_1, l_2)$ cycles satisfies
the formula $\Phi$ of Lemma~\ref{lemma:ReceptiveBounded}. We next illustrate
why such a run would be time-divergent (with appropriate chosen player-1
moves for the edge $a_1^0$).

Observe that after one $(l_1, l_2)$ cycle, the states always
satisfy $1 > x>y>0$ when at $l_1$, and  $1 > y>x>0$ when at $l_2$.
Figure~\ref{figure:ExampleReceptiveOneRegions} illustrates two 
paths through these
two regions after at least one $(l_1, l_2)$ cycle.
\begin{figure}[t]
  \begin{minipage}[t]{0.5\linewidth}
    \strut\centerline{\setlength{\unitlength}{0.00043745in}
\begingroup\makeatletter\ifx\SetFigFontNFSS\undefined%
\gdef\SetFigFontNFSS#1#2#3#4#5{%
  \reset@font\fontsize{#1}{#2pt}%
  \fontfamily{#3}\fontseries{#4}\fontshape{#5}%
  \selectfont}%
\fi\endgroup%
{\renewcommand{\dashlinestretch}{30}
\begin{picture}(4818,4153)(0,-10)
\thicklines
\path(510,4105)(510,505)(4785,505)
\path(510,505)(3660,3655)
\path(510,3655)(4740,3655)
\path(3660,4105)(3660,505)
\thinlines
\dashline{45.000}(1140,505)(1590,865)
\path(1515.037,766.611)(1590.000,865.000)(1477.555,813.463)
\dashline{45.000}(1590,865)(510,865)
\path(630.000,895.000)(510.000,865.000)(630.000,835.000)
\dashline{45.000}(555,865)(2490,2710)
\path(2423.854,2605.479)(2490.000,2710.000)(2382.449,2648.903)
\dashline{45.000}(2490,2710)(2490,505)
\path(2460.000,625.000)(2490.000,505.000)(2520.000,625.000)
\dashline{45.000}(2490,550)(3615,1720)
\path(3553.452,1612.707)(3615.000,1720.000)(3510.202,1654.293)
\dashline{45.000}(3615,1720)(555,1720)
\path(675.000,1750.000)(555.000,1720.000)(675.000,1690.000)
\dashline{45.000}(555,1720)(2310,3340)
\path(2242.172,3236.562)(2310.000,3340.000)(2201.475,3280.650)
\put(465,100){\makebox(0,0)[lb]{\smash{{\SetFigFontNFSS{11}{13.2}{\familydefault}{\mddefault}{\updefault}$x$}}}}
\put(3570,100){\makebox(0,0)[lb]{\smash{{\SetFigFontNFSS{11}{13.2}{\familydefault}{\mddefault}{\updefault}$1$}}}}
\put(15,595){\makebox(0,0)[lb]{\smash{{\SetFigFontNFSS{11}{13.2}{\familydefault}{\mddefault}{\updefault}$y$}}}}
\put(15,3610){\makebox(0,0)[lb]{\smash{{\SetFigFontNFSS{11}{13.2}{\familydefault}{\mddefault}{\updefault}$1$}}}}
\end{picture}
}}
    \end{minipage}
  \begin{minipage}[t]{0.5\linewidth}
    \strut\centerline{\setlength{\unitlength}{0.00043745in}
\begingroup\makeatletter\ifx\SetFigFontNFSS\undefined%
\gdef\SetFigFontNFSS#1#2#3#4#5{%
  \reset@font\fontsize{#1}{#2pt}%
  \fontfamily{#3}\fontseries{#4}\fontshape{#5}%
  \selectfont}%
\fi\endgroup%
{\renewcommand{\dashlinestretch}{30}
\begin{picture}(4818,4153)(0,-10)
\thicklines
\path(510,4105)(510,505)(4785,505)
\path(510,505)(3660,3655)
\path(510,3655)(4740,3655)
\path(3660,4105)(3660,505)
\thinlines
\dashline{45.000}(1140,505)(3075,2440)
\path(3011.360,2333.934)(3075.000,2440.000)(2968.934,2376.360)
\dashline{45.000}(3075,2440)(510,2440)
\path(630.000,2470.000)(510.000,2440.000)(630.000,2410.000)
\dashline{45.000}(510,2395)(1770,3565)
\path(1702.478,3461.362)(1770.000,3565.000)(1661.651,3505.330)
\dashline{45.000}(1770,3520)(1770,505)
\path(1740.000,625.000)(1770.000,505.000)(1800.000,625.000)
\dashline{45.000}(1770,505)(3570,2305)
\path(3506.360,2198.934)(3570.000,2305.000)(3463.934,2241.360)
\dashline{45.000}(3570,2305)(510,2305)
\path(630.000,2335.000)(510.000,2305.000)(630.000,2275.000)
\dashline{45.000}(510,2305)(1950,3610)
\path(1881.227,3507.188)(1950.000,3610.000)(1840.936,3551.647)
\put(465,100){\makebox(0,0)[lb]{\smash{{\SetFigFontNFSS{11}{13.2}{\familydefault}{\mddefault}{\updefault}$x$}}}}
\put(3570,100){\makebox(0,0)[lb]{\smash{{\SetFigFontNFSS{11}{13.2}{\familydefault}{\mddefault}{\updefault}$1$}}}}
\put(15,595){\makebox(0,0)[lb]{\smash{{\SetFigFontNFSS{11}{13.2}{\familydefault}{\mddefault}{\updefault}$y$}}}}
\put(15,3610){\makebox(0,0)[lb]{\smash{{\SetFigFontNFSS{11}{13.2}{\familydefault}{\mddefault}{\updefault}$1$}}}}
\end{picture}
}}
  
  \end{minipage}
  \label{figure:ExampleReceptiveOneRegions} 
  \caption{Two trajectories of 
    the cycle $(a_1^0, a_2^0)$ traversing through two regions of $\A_1$.
  }
\end{figure}
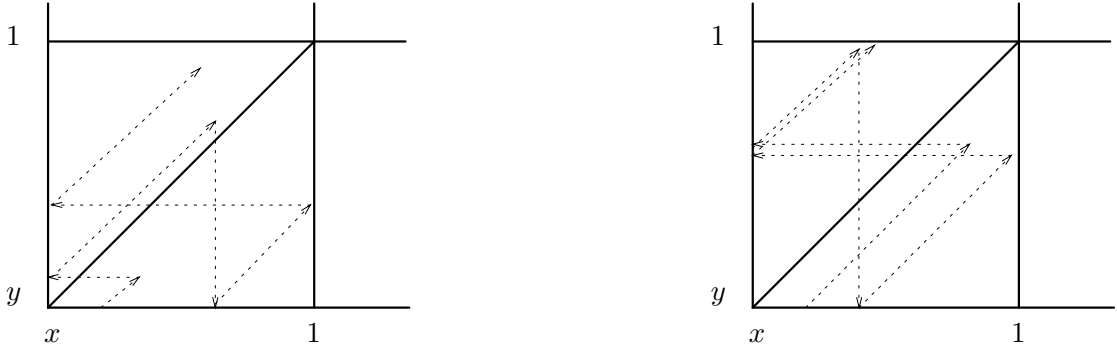
Note that the transitions into the region $1 > x>y>0$ are controlled by
player~2, and those into $1 > y>x>0$ controlled by player~1.
In the second trajectory, player~1 is \emph{not} able to take transitions
which make the clock $x$ more than $1/2$; but it is able to ensure that
the clock $y$ is more than  $1/2$ infinitely often.
Since the clock $y$ is more than $1/2$ infinitely often and is also reset 
infinitely often, time diverges (we will present a more formal 
 proof of time divergence of the run shortly).
It is easy to construct another timed automaton $\A_*$ in which player~1
can only ensure that clock $x$ is more than $1/2$ infinitely often.
It can then be seen that the automatons $\A_1$ and $\A_*$ can be ``combined''
by a player-2 action so that player~1 can only ensure that \emph{some} clock
is more than $1/2$ infinitely often; it cannot ensure that any one 
\emph{particular} clock will satisfy this property.
To ensure time divergence, player~1 hence also needs to ensure that 
\emph{all} clocks are reset infinitely often (as it does not know
which clock will be more than $1/2$ infinitely often).

We now formally show time divergence of the runs shown in
Figure~\ref{figure:ExampleReceptiveOneRegions}.
Let the duration of the $j$-th player~2 move be $\Delta_2^j$.
The value of the clock $y$ is then $\Delta_2^j$ when  location $l_1$
is entered for the $j$-th time, after the $j$-th $a_2^0$ move.
Player~1 picks its $j$-th $a_1^0$ move to be of duration $1-\Delta_2^j +
\varepsilon$.
Thus, in one cycle time passes by $1-\varepsilon$ time units.
With $\varepsilon <1$, it can be seen that time diverges. 
\qed
\end{example}

\begin{example}
In this example we illustrate why we require in the formula $\Phi$ of
Lemma~\ref{lemma:ReceptiveBounded} that if 
$\Box\Diamond\, 
  (\bl_1=\true)\, \wedge\,\bigwedge_{x\in C} (V_{>0}(x)=\true)$ does not hold,
then
$\Box\Diamond 
  (\bl_1=\false)\, \wedge\,  \bigvee_{x\in C} (V_{\geq 1}(x)=\true) $ must hold.
Consider the timed automaton game $\A_2$ in 
Figure~\ref{figure:ExampleReceptiveTwo}
\begin{figure}[t]
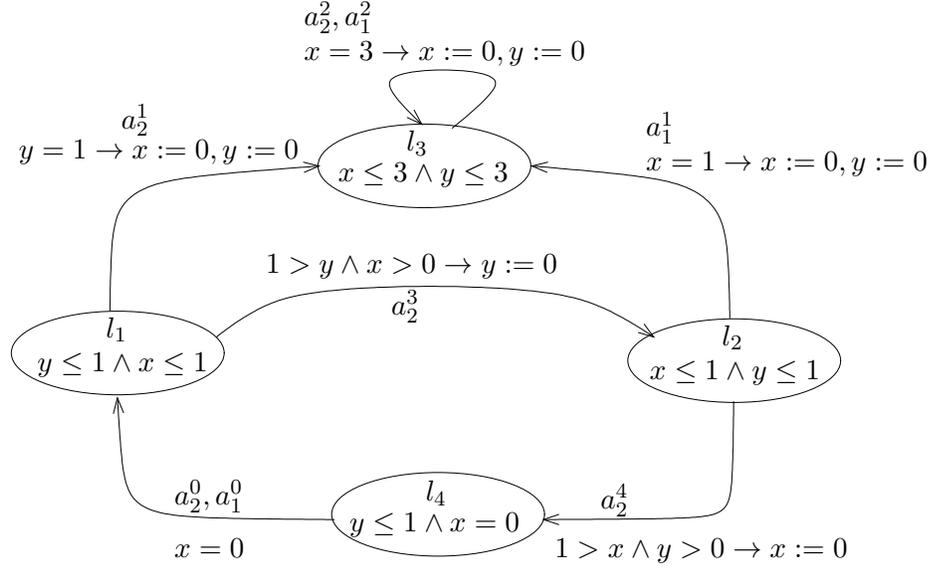

\strut\centerline{\input Figures/example-receptive-two.eepic}
\caption{A timed automaton game $\A_2$ without player-1 receptive strategies.}
\label{figure:ExampleReceptiveTwo}
\end{figure}
The edges $a_1^j$ are player-1 edges and $a_2^j$ player-2 edges.
The edges $a_2^2$ and $a_1^2$ have the same guards and reset maps.
It is clear that player~1 has a receptive strategy when at location $l_3$; it
repeatedly takes (or tries to take) the edge $a_1^2$.
Hence, player~2 keeps the game in the $(l_1, l_2, l_4)$.
For the $j$-th $a_2^3$ and the $j$-th  $a_2^4$ move, player~2  chooses a time
duration of $1/2^j$.
Player~1 is forced to take the move $a_1^0$ (of time duration 0) 
when at location $l_4$.
In this cycle with such a strategy by player~2, we have that
(1)~all clocks are reset infinitely often,
(2)~the moves of player~1 are picked infinitely often, and
(3)~all clock values are greater than 0 infinitely often
(i.e., $\Box\Diamond \bigwedge_{x\in C} (V_{>0}(x)=\true)$ holds).
But, time converges in such a run (and thus player~1 does not have  a
receptive strategy).
The states in $l_1,l_2,l_4$ (with $x<1\wedge y<1$) do \emph{not} satisfy 
$\Phi$ of
Lemma~\ref{lemma:ReceptiveBounded} because even though 
$\Box\Diamond \bigwedge_{x\in C} (V_{>0}(x)=\true)$ holds,
$\Box\Diamond\, 
  (\bl_1=\true)\, \wedge\,\bigwedge_{x\in C} (V_{>0}(x)=\true)$ does not hold.
As this example shows, if player~2 picks moves to
satisfy $\bigwedge_{x\in C} (V_{>0}(x)=\true)$, then it can choose arbitrarily 
small moves.
That is why require that if we are considering player~2 moves, then
$\bigvee_{x\in C} (V_{\geq 1}(x)=\true) $ must hold infinitely often.
\qed
\end{example}

\smallskip\noindent\textbf{Characterization of receptive strategies
for general timed automaton games(\cite{KCHenPra08})}.
Lemma~\ref{lemma:ReceptiveBounded} was generalized to all timed automaton games
in the following lemma presented in~\cite{KCHenPra08}.
The idea of the generalization is to identify the subset of clocks which
``escape'' to infinity; and then to take a disjunction over all such possible
subsets.
Note that 
once a clock $x$ becomes more than $c_x$, then its actual value can be
considered irrelevant in determining regions.
If only the clocks in $X\subseteq C$ have escaped beyond their maximum tracked
values, the rest of the clocks still need to be tracked.

\begin{lemma}[\cite{KCHenPra08}]
\label{lemma:ReceptiveGeneral}
Let $\A$ be a timed automaton game, and $\widetilde{\A}$ be the corresponding
enlarged game.
Then player~1 has a receptive strategy from a state $s$ iff 
$\tuple{s,\cdot}\in \sureu_1^{\widetilde{\A}}(\Phi^*)$, where
$\Phi^*= \Box\Diamond(bl_1=\true) \rightarrow 
\bigvee_{X\subseteq C}\ \phi_X$, and $\phi_X =$
\[
  \left (\bigwedge_{x\in  X}\Diamond\Box ( x > c_x) \right )
  \wedge\\
  \left (
    \left(\bigwedge_{x\in C\setminus X}\Box\Diamond (x=0)\right) \wedge 
    \left ( 
      \begin{array}{c}
        \Box\Diamond \left ((\bl_1=\true)\,\wedge\,
          \bigwedge_{x\in C\setminus X}(V_{>0}(x)=\true)  \right )\\
        \vee \\
        \Box\Diamond \left( (\bl_1=\true)\,\wedge\,
          \bigvee_{x\in C\setminus   X} (V_{\geq 1}(x)=\true)  \right )
      \end{array}
    \right )
  \right )
\]
\end{lemma}

\smallskip\noindent\textbf{New characterization of receptive strategies
for general timed automaton games}.
We shall see later that player-1 strategies which win for the objective
$\Phi^*$ of Lemma~\ref{lemma:ReceptiveGeneral} have a bound of
$(|C+1|)^{2^{|C|}}$ for the number of memory states required.
We present a new characterization of receptive strategies for which we can
prove a memory bound of only $(|C|+1)$.
First, we need to add $|C|$ predicates to the game structure $\widetilde{\A}$.
For a state $s$ of $\A$, we define another  function
 $V_{>\max}^* : C \mapsto \set{\true,\false}$.
The value of the predicate 
$V_{> \max}^*(x)$ for a clock $x\in C$ is $\true$ at a state 
$s$ iff the value of clock $x$ is more than $c_x$, and was more than $c_x$ in 
the previous state.
That is, if a state $s^p= \tuple{l^p,\kappa^p}$ and 
$\delta(s^p, \tuple{\Delta,a_1}) = s$, then at the state
$s$, the predicate  $V_{>\max}^*(x)$ is $\true$ iff 
$\kappa'(x) > c_x$ for $\kappa' \in \set{\kappa^p+\Delta' 
\mid 0\leq \Delta'\leq \Delta}$.
Let $\ddot{\A}$ be the enlarged game structure similar to $\widetilde{\A}$
with the state space being enlarged to also 
 have $V_{>\max}^*$ values (in addition to $V_{>0}$ and 
$V_{\geq 1}$ values): $\ddot{S} = 
S \times\set{\true,\false} \times \set{\true,\false}^C \times 
\set{\true,\false}^C \times \set{\true,\false}^C$.
A  state of $\widetilde{\A}$ is a tuple 
$\tuple{s,\bl_1, V_{>0}, V_{\geq 1}, V_{ > \max}^*}$, where $s$ is a state 
of $\A$, 
the component $\bl_1$ is  $\true$ iff player~1 is to be blamed for 
the last transition, and
$V_{>0}, V_{\geq 1}, V_{>\max}^*$ are as defined earlier.
A finite state concurrent game $\ddot{\A}^{\f}$ analogous to
$\widetilde{\A}^{\f}$ can be constructed, and results
analogous to Lemmas~\ref{lemma:FiniteConcurrentCPre}, 
\ref{lemma:FiniteConcurrentEquivalence} and~\ref{lemma:FiniteSafetyGame}
hold for the structures $\ddot{\A}$ and $\ddot{\A}^{\f}$.

First we present the following technical Lemma which will be used later.

\begin{lemma}
\label{lemma:StreettElim}
Let $\A$ be a timed automaton game, and $\ddot{\A}$ be the corresponding
enlarged game.
A run $\ddot{r}$ in $\ddot{\A}$ satisfies 
\[\bigwedge_{x\in C}\left(\Box\Diamond (x=0)\,\vee\,
          \Diamond\Box(V^*_{> \max}(x) = \true)\right)\] iff
it satisfies
\[\bigwedge_{x\in C}\Box\Diamond\, \left((x=0)\,\vee\, 
  (V^*_{> \max}(x) = \true)\right).\]
\end{lemma}
\begin{proof}
We prove inclusion in both directions.
\begin{enumerate}
\item 
  ($\Rightarrow$).  Suppose a run $\ddot{r}$ in $\ddot{\A}$
  satisfies $\bigwedge_{x\in C}\left(\Box\Diamond (x=0)\,\vee\,
    \Diamond\Box(V^*_{> \max}(x) = \true)\right)$.  
  Consider a clock $x\in C$.  
  If either $\Box\Diamond (x=0)$ or $\Diamond\Box(V^*_{>
    \max}(x) = \true)$ holds on $\ddot{r}$, it can be seen that
  $\Box\Diamond\, \left((x=0)\,\vee\, (V^*_{> \max}(x) =
    \true)\right)$ holds on $\ddot{r}$.
\item 
  ($\Leftarrow$).
  Suppose a run $\ddot{r}$ in $\ddot{\A}$ satisfies 
  $\bigwedge_{x\in C}\Box\Diamond\, \left((x=0)\,\vee\, 
    (V^*_{> \max}(x) = \true)\right)$.
  Consider a clock $x\in C$.  
  We must have either $\Box\Diamond(x=0)$ or
  $\Box\Diamond(V^*_{> \max}(x) = \true)$.
  If  $\Box\Diamond(x=0)$ on the run $\ddot{r}$, then it satisfies our 
  requirement.
  We show that if  run $\ddot{r}$ satisfies
  $\Box\Diamond(V^*_{> \max}(x) = \true)$; then it must satisfy either
  $\Box\Diamond(x=0)$ or $\Diamond\Box(V^*_{> \max}(x) = \true)$.
  This is because the only way for the value of a clock to decrease is to be
  reset to 0.
  In particular, once the clock $x$ becomes more than $c_x$, the only way for
  it to become less than or equal to $c_x$ is to be reset to 0.
  If the clock $x$ becomes more than $c_x$ and is never reset, it will stay
  more than $c_x$ forever.
  \qed
\end{enumerate}

\end{proof}

\begin{lemma}[Receptive strategies when clocks may be unbounded in $\A$]
\label{lemma:NewReceptiveGeneral}
Let $\A$ be a timed automaton game, and $\ddot{\A}$ be the corresponding
enlarged game.
Then player~1 has a receptive strategy from a state $s$ of $\A$ iff 
$\tuple{s,\cdot}\in \sureu_1^{\ddot{\A}}(\Phi^{\dagger})$, where
$\Phi^{\dagger} = \Box\Diamond(bl_1=\true) \rightarrow \Psi^{\dagger}$,
and $\Psi^{\dagger} =$
\[ 
\begin{array}{c}
  \left (
    \begin{array}{c}
      \vspace*{1mm}
      \left(\bigwedge_{x\in C}\Box\Diamond\,\left( (x=0)\,\vee\,
          (V^*_{> \max}(x) = \true)\right)\right) \\
      \vspace*{1mm}
      \bigwedge\\
      \left ( 
        \begin{array}{c}
          \vspace*{1mm} 
          \Box\Diamond\, \left(\, (\bl_1=\true) \,\wedge\, 
            \left(\bigwedge_{x\in C}(V_{>0}(x)=\true) 
              \right)\,\wedge\, \left(\bigvee_{x\in C}(V^*_{> \max}(x) = \false)
              \right)\,\right) \\
          \vspace*{1mm}
          \bigvee \\
          \Box\Diamond\, \left(\,(\bl_1=\false)\, \wedge\, \bigvee_{x\in C}
              \,\left(\left(V_{\geq 1}(x)=\true\right) 
              \wedge (V^*_{> \max}(x) = \false)\right)\,
              \right) 
        \end{array}
      \right )
    \end{array}
  \right ) \\
  \\
  \bigvee\\
  \\
  \left(\bigwedge_{x\in C}  \Diamond\Box(V^*_{> \max}(x) = \true)\right)
\end{array}
\]
\end{lemma}
\begin{proof}
We prove inclusion in both directions.
\begin{enumerate}
\item ($\Leftarrow$).
For a state $\ddot{s}\in \sureu_1^{\ddot{\A}}(\Phi^{\dagger})$, we
show that player~1 has a receptive strategy from $\ddot{s}$.
Let $\pi_1$ be a pure sure winning region strategy: 
since $\Phi^{\dagger}$ is an
$\omega$-regular region objective such a strategy exists by 
Lemma~\ref{lemma:RegionStrategies}.
Let $\ddot{R}_{\max}$ denote the region where for every clock $x$, the value 
of $x$ is more than $c_x$.
Consider a region strategy $\pi_1'$ for player~1 that 
is region-equivalent to $\pi_1$ 
such that given a run prefix $\ddot{r}[0..k]$, the strategy $\pi_1'$ acts like
$\pi_1$ except when:
\begin{itemize}
\item 
  If  $\reg(\ddot{r}[k]) = \ddot{R}_{\max}$ and $\pi_1(\ddot{r}[0..k]) =
  \tuple{\Delta,a_1}$,
  then $\pi_1'(\ddot{r}[0..k]) =
  \tuple{\Delta',a_1}$ such that $\Delta' >1$ (observe that
  $\reg(\ddot{r}[k]+\Delta') = \ddot{R}_{\max}$ for any $\Delta'$).

\item 
  If  $\reg(\ddot{r}[k]) \neq \ddot{R}_{\max}$ and $\pi_1(\ddot{r}[0..k]) =
  \tuple{\Delta,a_1}$ with
  the state $\ddot{r}[k]+\Delta$ being  such that
  the value of some clock $x$ is less than or equal to $c_x$ but more than $0$,
  then $\pi_1'(\ddot{r}[0..k]) =
  \tuple{\Delta',a_1}$ such that (1)~$\reg(\ddot{r}[k]+\Delta') =
  \reg(\ddot{r}[k]+\Delta)$, and
  (2)~the value of some clock $y$ (possibly different from $x$)
  is less than $c_y$ at $\ddot{r}[k]$, and is more
  than $1/2$ at $\ddot{r}[k]+\Delta'$
  (intuitively, $\pi_1'$  jumps near the region boundary of
  $\reg(\ddot{r}[k]+\Delta)$).
\end{itemize}

We have $\Phi^{\dagger} \equiv \left(\neg\Box\Diamond (\bl_1=\true)\right)
\, \vee\, 
 \left(\left(\Psi_1^{\dagger} \wedge 
  \left(\Psi_2^{\dagger} \vee \Psi_3^{\dagger}\right)\right) 
\vee \Psi_4^{\dagger}\right) $, where
\begin{eqnarray*}
  \Psi_1^{\dagger} & = & \bigwedge_{x\in C}\Box\Diamond\, \left( (x=0)\,\vee\,
    (V^*_{> \max}(x) = \true)\right)\\
  \Psi_2^{\dagger} & = &
  \Box\Diamond\, \left(  (\bl_1=\true)\, \wedge\, \left(\bigwedge_{x\in C}(V_{>0}(x)=\true) 
    \right)\,\wedge\, \left(\bigvee_{x\in C}(V^*_{> \max}(x) = \false)
    \right)\right) \\
  \Psi_3^{\dagger} & = &
  \Box\Diamond\, \left(\,(\bl_1=\false)\, \wedge\, \bigvee_{x\in C}
    \,\left(\left(V_{\geq 1}(x)=\true\right) 
      \wedge (V^*_{> \max}(x) = \false)\right)\,
  \right)\\
  \Psi_4^{\dagger} & = &
  \bigwedge_{x\in C}  \Diamond\Box(V^*_{> \max}(x) = \true)
\end{eqnarray*}

Given any player-2 strategy $\pi_2$, 
consider any run $\ddot{r}\in \outcomes(\ddot{s},\pi_1',\pi_2)$.
The run $\ddot{r}$ must satisfy $\Phi^{\dagger}$.
One of the following conditions  must be satisfied on the run $\ddot{r}$.
\begin{enumerate}

\item 
  $\Diamond \Box (\bl_1= \false)$.
  This satisfies the receptiveness condition.
\item 
  $\left(\Box\Diamond (\bl_1=\true)\right)\,\wedge\, 
  \Psi^{\dagger}_4$
  This means that in the run $\ddot{r}$,  every clock $x$ eventually becomes
  greater than $c_x$; and  moves of player~1 are  chosen  infinitely often.
  Since the strategy $\pi_1'$ chooses moves of duration greater than 1 when 
  staying in $\ddot{R}_{\max}$, time diverges in the run $\ddot{r}$.
\item 
  $(\left(\Box\Diamond (\bl_1=\true)\right)\,\wedge\, \neg\Psi^{\dagger}_4)\,
  \bigwedge\, \left(\Psi_1^{\dagger} \wedge 
  \left(\Psi_2^{\dagger} \vee \Psi_3^{\dagger}\right)\right)$.
   The constraint $\neg\Psi^{\dagger}_4$  means that, 
   there is some clock $x$ which 
  is less than $c_x$ infinitely often.
  Satisfaction of the constraint $\Psi_1^{\dagger}$
  and 
  Lemma~\ref{lemma:StreettElim} imply that
  \[\bigwedge_{x\in C}\left(\Box\Diamond (x=0)\,\vee\,
    \Diamond\Box(V^*_{> \max}(x) = \true)\right)\]
  must be satisfied on the run $\ddot{r}$.
  That is, 
  each clock $x$ which is not eventually always greater than $c_x$ must be $0$
  infinitely often.
  Also, the run must satisfy either $\Psi^{\dagger}_2$ or $\Psi^{\dagger}_3$.

  Suppose we have the first case (i.e.,  $\Psi^{\dagger}_2$ holds).
  Then, for infinitely many $k$, player-1 moves are chosen 
  from $\ddot{r}[0..k]$ such that
  for some clock $x$, we have (1)~the value of the clock $x$ is less than $c_x$
  at $\ddot{r}[k]$ (note that if the value of $x$ is less than $c_x$ at some
  point during the move, then it must be less than $c_x$ at the origin), and 
  (2)~for $\pi_1'(\ddot{r}[0..k]) = \tuple{\Delta',a_1}$, 
  the value of the clock $x$ at $\ddot{r}[k] +\Delta'$ is more than $0$.
  Because of the design of $\pi_1'$, this means that for infinitely
  many $k$, there is some clock $y$ such that
  if $\pi'_1(\ddot{r}[0..k]) = \tuple{\Delta', a_1}$ then,
  (1)~the value of clock $y$ at $\ddot{r}[k]$ is not more than $c_y$, and
  (2)~the value of clock $y$ is more than $1/2$ at $\ddot{r}[k]+\Delta'$.
  Since the clock $y$ must also be equal to $0$ infinitely often (as it is
  not more than $c_y$ eventually from above, and due to $\Psi_1^{\dagger}$),
  this implies that time diverges.

  Suppose we have the second case (i.e., $\Psi_3^{\dagger}$ holds). 
  Then, for infinitely many $k$, player-2 moves are chosen 
  from $\ddot{r}[0..k]$ such that
  for some clock $x$, we have (1)~the value of the clock $x$ is less than $c_x$
  at $\ddot{r}[k]$,
  and,
  (2)~for $\pi_2(\ddot{r}[0..k]) = \tuple{\Delta_2,a_2}$, 
  the value of the clock $x$ at $\ddot{r}[k] +\Delta_2$ is more than or equal
  to $1$.
  Since the clock $x$ must also be equal to $0$ infinitely often (as it is
  not more than $c_x$ eventually from above, and due to $\Psi_1^{\dagger}$),
  this implies that time diverges.
\end{enumerate}
Thus, in all cases, the strategy $\pi_1'$ ensures that either player~1 is not
to blame, or time diverges.
Hence, $\pi_1'$ is a receptive strategy from $\ddot{s}$.
\item 
$(\Rightarrow)$.
For  a state $\ddot{s}\notin \sureu_1^{\ddot{\A}}(\Phi^{\dagger})$, 
we show that player~1 does not have any receptive strategy starting from state
$\ddot{s}$.
We have  $\neg\Phi^{\dagger} \equiv \left(\Box\Diamond (\bl_1=\true)\right) \wedge
\neg \left(\left(\Psi_1^{\dagger} \wedge 
  \left(\Psi_2^{\dagger} \vee \Psi_3^{\dagger}\right)\right) 
\vee \Psi_4^{\dagger}\right) $, where
$\Psi_1^{\dagger}, \Psi_2^{\dagger}, \Psi_3^{\dagger}$ and $\Psi_4^{\dagger}$ 
are as defined previously.
Simplifying, we get
$\neg\Phi^{\dagger} \equiv \left(\Box\Diamond (\bl_1=\true)\right) \wedge
 \left(\neg\Psi_1^{\dagger} \vee
  \left(\neg\Psi_2^{\dagger} \wedge \neg\Psi_3^{\dagger}\right)\right) 
\wedge  \neg\Psi_4^{\dagger} $, where
\begin{eqnarray*}
  \neg\Psi_1^{\dagger} & = & \bigvee_{x\in C}\Diamond\Box\, 
  \left( (x>0)\,\wedge\,
    (V^*_{> \max}(x) = \false)\right)\\
  \neg\Psi_2^{\dagger} & = &  \Diamond\Box\,\left( (\bl_1=\true) \rightarrow 
    \left(\left(\bigvee_{x\in C} \left(\left(V_{>0}(x)=\false\right)\right) 
      \right)
      \,\vee\, \left(\bigwedge_{x\in C}\left(V^*_{> \max}(x) = \true\right)\right)
    \right ) \right)\\
  \neg \Psi_3^{\dagger} & = & \Diamond\Box\,\left( (\bl_1=\false) \rightarrow 
     \bigwedge_{x\in C} \left(\,(V_{\geq 1}(x)=\false) 
      \vee (V^*_{> \max}(x) = \true)\,
       \right )\right )\\
%
   \neg \Psi_4^{\dagger} & = &
   \Box\Diamond\bigvee_{x\in C} (V^*_{> \max}(x) = \false)\\
   &&(\text{Using the identity } \bigvee_{x\in C} \Box\Diamond P(x)
   \equiv \Box\Diamond \bigvee_{x\in C} P(x))
\end{eqnarray*}
Recall the finite state game $\widetilde{\A}^{\f}$ based on the regions of
$\widetilde{\A}$.
There exists a similar finite state game $\ddot{\A}^{\f}$ based on the 
regions of $\ddot{\A}$, with results relating $\ddot{\A}^{\f}$ and $\ddot{\A}$
as the results relating $\widetilde{\A}^{\f}$ and $\widetilde{\A}$.
Suppose $\ddot{s}\notin\sureu_1^{\ddot{\A}}(\Phi^{\dagger})$.
Then $\ddot{s}\notin\pureu_1^{\ddot{\A}}(\Phi^{\dagger})$ by 
Corollary~\ref{corollary:PureGame}.
Consider any pure player-1 strategy $\pi_1$ in $\ddot{\A}$.
By Lemma~\ref{lemma:FiniteSafetyGame}, 
$\reg(\ddot{s})\notin\pureu_1^{\ddot{\A}^{\f}}(\Phi)^{\dagger}$, and
there exists a runcover
$\runcover$ for $\VRuns^{\widetilde{\A}^{\f}}$
such that for   
any player-2 pure 
spoiling strategy $\pi_2^{\ddot{\A}^{\f}}$ against 
$\mbf^{\runcover}\left(\pi_1\right)$ in $\ddot{\A}^{\f}$ from
$\reg(\ddot{s})$,
we have that
every player-2 strategy in 
$\spoil^{\runcover}(\pi_1, \pi_2^{\ddot{\A}^{\f}})$ is
a spoiling strategy against $\pi_1$ in the structure  $\ddot{\A}$.

Let $\runcover$ be such a runcover, and let 
$\pi_2^{\ddot{\A}^{\f}}$ be any such player-2 strategy against 
$\mbf^{\runcover}\left(\pi_1\right)$ in $\ddot{\A}^{\f}$ from
$\reg(\ddot{s})$.
We show that with an appropriately chosen $\pi_2$ in 
$\spoil^{\runcover}(\pi_1, \pi_2^{\ddot{\A}^{\f}})$,
player~2 can ensure that
in one of the resulting runs, player~1 is not blameless, and time converges,
 and hence player~1 does not have a receptive pure strategy in 
$\ddot{\A}$.
The result follows from observing that if player~1 does not have a pure 
receptive strategy, then it does not have a (possibly randomized) receptive
strategy (as a randomized strategy may be viewed as a random choice over pure
strategies).

Consider runs $\ddot{r} \in \outcomes(\ddot{s},\pi_1,\pi_2)$ for
$\pi_2\in \spoil^{\runcover}(\pi_1, \pi_2^{\ddot{\A}^{\f}})$..
One of the runs must satisfy $\neg\Phi^{\dagger}$, which can happen in
 one of the following ways.
\begin{enumerate}
\item $(\Box\Diamond(\bl_1=\true))\, \wedge\, \neg\Psi_1^{\dagger} \wedge
  \neg\Psi_4^{\dagger}$.
  The condition $\neg\Psi_1^{\dagger}$ means that there is some clock $x$
  which eventually stays strictly greater than 0, and also stays less than
  or equal to $c_x$.
  This is impossible in a time-divergent run as clocks can only be reset to
  0.
  Thus, in this run time does not diverge, and player~1 is not blameless.
\item 
  $(\Box\Diamond(\bl_1=\true))\, \wedge\,  \neg\Psi_2^{\dagger} \wedge
  \neg\Psi_3^{\dagger} \wedge  \neg\Psi_4^{\dagger}$.
  The clause $ \neg\Psi_4^{\dagger}$ implies that there is some clock
  $x$ such that it is not greater than $c_x$ infinitely often during transitions
  (including the originating state).
  The clause $\neg\Psi_2^{\dagger}$ means that eventually
  if an action of player~1 is chosen, then  either
  (1)~every clock $x$ has value greater than $c_x$ during  the move (this is 
  not possible if the run satisfies $\neq\Psi_4^{\dagger}$) , or
  (2)~for  some  clock $x$, the value of $x$ stays at 0 throughout the move
   (which means that the  move of player-1 is of duration $0$).
  This clause $\neg\Psi_3^{\dagger}$ means that eventually
  if an action of player~2 is chosen, then 
  for every  clock $x$,
  either the clock $x$ has value  greater
  than  $c_x$ during the move, 
  or the value of $x$ is strictly less than $1$ during the move.

  Player~2 can have a strategy which takes moves smaller than $1/2^j$ during
  the $j$-th visit to a region $\ddot{R}$ in which  every clock $x$ either has
  value less than 1, or greater than $c_x$.
  We formalize the above statement.
  The strategy $\pi_2^{\ddot{\A}^{\f}}$ spoils 
  $\mbf^{\runcover}\left(\pi_1\right)$ from winning in $\ddot{\A}^{\f}$
  for the objective $\Phi^{\dagger}$.
  Given a run prefix $\ddot{r}[0..k]$ of $\ddot{\A}$, let 
  $\pi_1(\ddot{r}[0..k]) = \tuple{\Delta_1,a_1}$.
  Consider  a player-2 strategy $\pi_2$ in 
  $\spoil^{\runcover}(\pi_1, \pi_2^{\ddot{\A}^{\f}})$, and
  let $\pi_2^{\ddot{\A}^{\f}}(\regmap(\ddot{r}[0..k])) = 
  \tuple{\ddot{R}_2,a_2,i}$.
  Let $\pi_2$ be a strategy in 
  $\spoil^{\runcover}(\pi_1, \pi_2^{\ddot{\A}^{\f}})$
  such  that for
  $\pi_2((\ddot{r}[0..k]) = \tuple{\Delta_2,a_2}$ we have
  $\Delta_2 \leq \Delta_1$ and $\Delta_2 < 1/2^k$ 
  whenever the following conditions hold.
  \begin{enumerate}
  \item 
    Each clock $x$ in $\ddot{R}_2$ is either less than 1, or
    more than $c_x$; and 
  \item
    Either
    \begin{enumerate}
    \item  
      $\ddot{R}_2$ is a region predecessor of
      $\reg(\ddot{r}[k] +\Delta_1)$; or
    \item 
      $i=2$ and  $\reg(\ddot{r}[k] +\Delta_1) = \ddot{R}_2$
    \end{enumerate}
  \end{enumerate}
  It can be observed from Equation~\ref{equation:FiniteSpoilingStrategy}
  that such
  a $\Delta_2$ and such a strategy $\pi_2$ in
  $\spoil^{\runcover}(\pi_1, \pi_2^{\ddot{\A}^{\f}})$ 
  always exist.
    The above condition ensures that if a move of player~2 is chosen to
    a region $\ddot{R}$ in which  every clock $x$ either has
    value less than 1, or greater than $c_x$, then the 
    moves smaller than $1/2^j$ during
    the $j$-th stage of the game.
    The strategy $\pi_2$ is a spoiling strategy against $\pi_1$ by
    Lemma~\ref{lemma:FiniteSafetyGame} as $\pi_2$ is in
  $\spoil^{\runcover}(\pi_1, \pi_2^{\ddot{\A}^{\f}})$.
    Moreover,  this strategy ensures that at least one of the resulting
    runs $\ddot{r}$  satisfies $\neg\Phi^{\dagger}$.
    \begin{enumerate}
    \item 
      If $\ddot{r}$ satisfies 
      $(\Box\Diamond(\bl_1=\true))\, \wedge\, \neg\Psi_1^{\dagger} \wedge
      \neg\Psi_4^{\dagger}$, then the run is time
      convergent, and player~1 is not blameless.
    \item 
      If $\ddot{r}$ satisfies 
      $(\Box\Diamond(\bl_1=\true))\, \wedge\,  \neg\Psi_2^{\dagger} \wedge
      \neg\Psi_3^{\dagger} \wedge  \neg\Psi_4^{\dagger}$, then we have that:
      \begin{enumerate}
      \item 
        Eventually every chosen move of player~2 results
        in a region $\ddot{R}$ in which  every clock $x$ either has
        value less than 1, or greater than $c_x$, with the duration 
        of the player-2 move being
        smaller than $1/2^j$ during
        the $j$-th stage of the game; and
      \item 
        Eventually every chosen move of player~1 is of time duration 0.
      \end{enumerate}
      Thus, time is convergent in the run $\ddot{r}$ and player~1 is not
      blameless.
    \end{enumerate}
  \end{enumerate}
  Hence, in both cases,
  player~1 does not have a pure receptive strategy from $\ddot{s}$ (from which
  it follows that it does not have any receptive  strategy from $\ddot{s}$).
    \qed

\end{enumerate}
 
\end{proof}

\subsection{Memory Requirement of Receptive Strategies}

In this subsection we deduce memory bounds on player-1 
receptive strategies
using Zielonka tree analysis (see~\cite{DziembowskiJW97} for details).
We first  deduce a bound that allows player~1 to win in the finite state 
concurrent game  $\ddot{\A}^{\f}$.
A player-1 winning strategy in $\ddot{\A}^{\f}$ can be mapped to a player-1
winning strategy in $\ddot{\A}$ by letting
$\pi_1^{\ddot{\A}}(\ddot{r}[0..k]) = \tuple{\Delta,a_1}$ such that
(a)~$\pi_1^{\ddot{\A}^{\f}}\left(\reg(\ddot{r}[0..k])\right) = 
\tuple{\ddot{R},a_1}$, and
(b)~$\reg(\ddot{r}[k]+\Delta)=\ddot{R}$.
Thus, the memory requirement for a player-1
winning strategy in $\ddot{\A}$  is not more than as for in the finite game
 $\ddot{\A}^{\f}$.
We note that Zielonka tree analysis holds only for turn based games,
 but since concurrent games with 
 sure winning conditions reduce to concurrent games in which both
players may use only pure strategies, which in turn reduce to
turn based games, the Zielonka tree analysis is valid for
game $\ddot{\A}^{\f}$ with sure winning conditions.

\smallskip\noindent\textbf{Zielonka tree analysis}.
Let $\ap$ be a set of atomic propositions, and let
$\ap_N$ be $AP$ together with the negations of the propositions, i.e.,
$\ap\,\cup\,\set{\neg P\mid P\in \ap}$.
We say a set  $\mcalb\subseteq \ap_N$ is \emph{consistent} with respect to
$\ap$ iff for all propositions  $P\in\ap$,
either $P\in \mcalb$, or $\neg P\in \mcalb$ (or both belong to $ \mcalb$).
A \emph{Muller} winning condition $\zf$ is a consistent subset of $2^{\ap_N}$.
An infinite play satisfies the Muller condition iff the set of
propositions (or the negation of propositions) 
occurring infinitely often in the play belongs to $\zf$.
Given $\mcalb\subseteq \ap_N$, let $\zf \upharpoonright \mcalb$ denote
the set $\set{D\in \zf \mid D\subseteq \mcalb}$.
The \emph{Zielonka tree} $\mcalz_{\zf,\mcalb}$ of a Muller condition
$\zf$ over $\ap$ with   $\mcalb = \ap_N$
is defined inductively as follows:
\begin{enumerate}
\item
  If $\mcalb\in\zf$, then the root of $\mcalz_{\zf,\mcalb}$ is labelled with 
  $\mcalb$.
  Let $\mcalb_1,\dots,\mcalb_k$ be all the maximal sets in:
  $
  \left\{
    \begin{array}{l}      
      \mcalb^*\notin\zf \mid \mcalb^*\subseteq \mcalb, \text{ and }
      \mcalb^* \text{ consistent with respect to } \ap
    \end{array}
  \right\}
  .
  $
  The root of $\mcalz_{\zf,\mcalb}$ then has as children the Zielonka trees 
  $\mcalz_{\zf\upharpoonright \mcalb_i, \mcalb_i}$ of $\zf\upharpoonright \mcalb_i$
  for $1\leq i\leq k$ .
\item
  If $\mcalb\notin\zf$, then $\mcalz_{\zf,\mcalb}= \mcalz_{\overline{\zf},\mcalb}$,
  where $\overline{\zf} = \set{D \in 2^{\mcalb}\mid D\notin \zf \text{ and }
    D \text{ is consistent with 
      respect to } \ap}$.
\end{enumerate}

A node of the Zielonka tree $\mcalz_{\zf,\mcalb}$ is a \emph{Good} node if it
is labelled with a set from $\zf$, otherwise it is a \emph{Bad} node.


\smallskip\noindent\textbf{Equivalent definition of Zielonka trees.}
We now present an equivalent definition  (which suffices for our purposes)
of the Zielonka tree
$\mcalz_{\zf,\ap_N}$ of a Muller condition
$\zf$ over $\ap$.
Every node of the Zielonka tree $\mcalz_{\zf,\ap_N}$  with is labelled with a 
consistent subset $\mcalb\subseteq \ap_N$.
A node of the Zielonka tree $\mcalz_{\zf,\ap_N}$ is a \emph{Good} node if it
is labelled with a set from $\zf$, otherwise it is a \emph{Bad} node.
The root is labelled with $\ap_N$.
The children of a node $v$  are defined inductively as follows:
\begin{enumerate}
\item
  Suppose $v$ is a Good node labelled with $\mcalb_v$.
  Let $\mcalb_1,\dots,\mcalb_k$ be all the maximal sets in:
  $\left\{
    \begin{array}{l}      
      \mcalb^*\notin\zf \mid \mcalb^*\subseteq \mcalb, \text{ and }
      \mcalb^* \text{ consistent with respect to } \ap
    \end{array}
  \right\}
  .
  $
  The node $v$ then has $k$  children (that are all Bad)  labelled with
  $\mcalb_1,\dots,\mcalb_k$.
\item
   Suppose $v$ is a Bad node labelled with $\mcalb_v$.
  Let $\mcalb_1,\dots,\mcalb_k$ be all the maximal sets in:
  $\left\{
    \begin{array}{l}      
      \mcalb^*\in\zf \mid \mcalb^*\subseteq \mcalb, \text{ and }
      \mcalb^* \text{ consistent with respect to } \ap
    \end{array}
  \right\}
  .
  $
  The node $v$ then has $k$ children (that are all Good) labelled with
  $\mcalb_1,\dots,\mcalb_k$.
\end{enumerate}

\smallskip\noindent\textbf{The number $\mathbf{m_{\zf}}$ of a Muller condition}.
Let $\zf$ be a a Muller condition that is a consistent subset of $2^{\ap_N}$.
Consider the Zielonka tree $\mcalz_{\zf,\ap_N}$ of  $\zf$.
We define a number  $m_{\zf}^v$ for each node $v$ of $\mcalz_{\zf,\ap_N}$ 
inductively.
\[ 
m_{\zf}^v = 
\begin{cases}
  1 & \text{if } v \text{ is a leaf},\\
  \sum_{i=1}^k m_{\zf}^{v_i} & \text{if } v \text{ is a Good node and has 
    children } v_1,\dots,v_k,\\
  \max\set{ m_{\zf}^{v_1},\dots,  m_{\zf}^{v_k}} 
  &\text{if } v \text{ is a Bad  node and has 
    children } v_1,\dots,v_k.\\
\end{cases}
\]
The number $m_{\zf}$ of the Muller condition $\zf$ is defined to be
$m_{\f}^{v_r}$ where $v_r$ is the root of the Zielonka tree $\mcalz_{\f,\ap_N}$.

\begin{lemma}[\cite{DziembowskiJW97}]
\label{lemma:ZielonkaMemory}
Let $\TG^f$ be a finite state turn based game.
If player~1 has a sure winning strategy for a Muller objective $\f$ from a
state $s$ in $\TG^f$, then it has a pure sure winning strategy from $s$
with at most $m_{\f}$ memory states.
\end{lemma}

Now we use Zielonka tree analysis to deduce memory
requirements of receptive strategies.
\begin{lemma}
\label{lemma:ZielonkaMemoryReceptiveFinite}
\begin{enumerate}
\item 
  Let $\phi_1 = (\Diamond\Box F_1 )\, \vee\, (\Diamond\Box F_2)\, \vee\,
  \bigwedge_{i\leq n}(\Box\Diamond I_j)$, where $F_1,F_2,I_j$ are boolean
  predicates on states of a finite state game $\TG^f$.
  Player~1 has a pure sure winning strategy from $\sureu_1(\phi_1)$ 
  that requires 
  at most $n$  memory states 
  for the objective $\phi_1$.

\item 
  Let $\phi_2 =$
  $(\Diamond\Box F )\, \vee\, 
  \bigvee_{\alpha\leq m} \left(
    \Diamond\Box F_{\alpha} \,\wedge\,
    \left( \bigwedge_{i\leq n} \Box\Diamond I_{\alpha,i}\right)\,\wedge\,
    \Box\Diamond I_{\alpha}\right)$, where $F,F_{\alpha}, I_{\alpha,i},I_{\alpha}$ 
  are boolean
  predicates on states of a finite state game $\TG^f$.
  Player~1  has a pure sure winning strategy from
  $\sureu_1(\phi_2)$
  that requires 
  at most $(n+1)^m$  memory states 
  for the objective $\phi_2$.

\end{enumerate}

\end{lemma}
\begin{proof}
  We present Zielonka tree analysis for each case (in the figures $U=\ap_N$),
  and use Lemma~\ref{lemma:ZielonkaMemory} to deduce the memory bounds.
  The leaves are depicted with double boundaries in the Figures.
  Bad nodes are pictured  as boxes, and Good nodes as ovals.
  \begin{enumerate}
  \item 
    Consider the Zielonka tree in Figure~\ref{figure:Zielonka-gen}.
    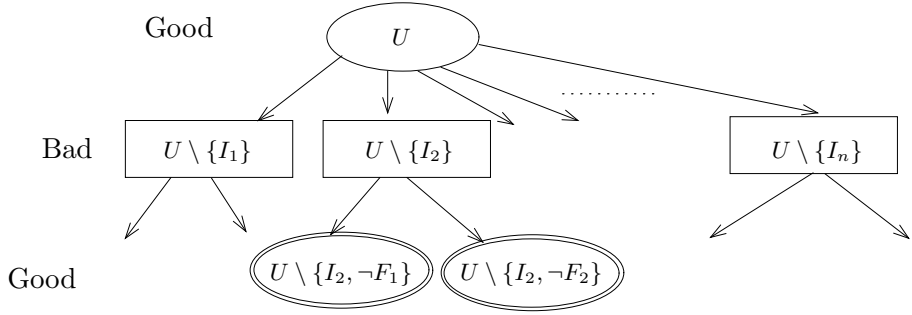
\begin{figure}[t]
      \strut\centerline{\setlength{\unitlength}{0.00043745in}
\begingroup\makeatletter\ifx\SetFigFontNFSS\undefined%
\gdef\SetFigFontNFSS#1#2#3#4#5{%
  \reset@font\fontsize{#1}{#2pt}%
  \fontfamily{#3}\fontseries{#4}\fontshape{#5}%
  \selectfont}%
\fi\endgroup%
{\renewcommand{\dashlinestretch}{30}
\begin{picture}(10691,3676)(0,-10)
\put(4695,3249){\ellipse{1800}{810}}
\put(6226,456){\ellipse{2160}{900}}
\put(6216,452){\ellipse{2070}{810}}
\put(3966,486){\ellipse{2160}{900}}
\put(3966,490){\ellipse{2070}{810}}
\path(3975,3024)(2985,2259)
\path(3090.744,2416.538)(2985.000,2259.000)(3164.118,2321.583)
\path(4515,2844)(4515,2349)
\path(4455.000,2529.000)(4515.000,2349.000)(4575.000,2529.000)
\path(4875,2844)(6000,2214)
\path(5813.633,2249.598)(6000.000,2214.000)(5872.265,2354.299)
\path(5145,2889)(6765,2259)
\path(6575.492,2268.320)(6765.000,2259.000)(6618.986,2380.161)
\path(5595,3159)(9600,2349)
\path(9411.678,2325.873)(9600.000,2349.000)(9435.466,2443.491)
\dottedline{105}(6585,2574)(7665,2574)
\path(1950,1584)(1410,864)
\path(1470.000,1044.000)(1410.000,864.000)(1566.000,972.000)
\path(1410,2259)(3390,2259)(3390,1584)
	(1410,1584)(1410,2259)
\path(3750,2259)(5730,2259)(5730,1584)
	(3750,1584)(3750,2259)
\path(8565,2304)(10545,2304)(10545,1629)
	(8565,1629)(8565,2304)
\path(2431,1575)(2836,945)
\path(2688.193,1063.967)(2836.000,945.000)(2789.134,1128.858)
\path(4425,1584)(3840,909)
\path(3912.546,1084.320)(3840.000,909.000)(4003.229,1005.728)
\path(4740,1584)(5640,819)
\path(5463.992,889.860)(5640.000,819.000)(5541.710,981.293)
\path(9549,1639)(8334,874)
\path(8454.353,1020.680)(8334.000,874.000)(8518.291,919.132)
\path(9689,1628)(10679,863)
\path(10499.882,925.583)(10679.000,863.000)(10573.256,1020.538)
\put(4560,3159){\makebox(0,0)[lb]{\smash{{\SetFigFontNFSS{9}{10.8}{\familydefault}{\mddefault}{\updefault}$U$}}}}
\put(1635,3249){\makebox(0,0)[lb]{\smash{{\SetFigFontNFSS{11}{13.2}{\familydefault}{\mddefault}{\updefault}Good}}}}
\put(420,1809){\makebox(0,0)[lb]{\smash{{\SetFigFontNFSS{11}{13.2}{\familydefault}{\mddefault}{\updefault}Bad}}}}
\put(15,279){\makebox(0,0)[lb]{\smash{{\SetFigFontNFSS{11}{13.2}{\familydefault}{\mddefault}{\updefault}Good}}}}
\put(1860,1809){\makebox(0,0)[lb]{\smash{{\SetFigFontNFSS{9}{10.8}{\familydefault}{\mddefault}{\updefault}$U\setminus \set{I_1}$}}}}
\put(4245,1809){\makebox(0,0)[lb]{\smash{{\SetFigFontNFSS{9}{10.8}{\familydefault}{\mddefault}{\updefault}$U\setminus \set{I_2}$}}}}
\put(9060,1809){\makebox(0,0)[lb]{\smash{{\SetFigFontNFSS{9}{10.8}{\familydefault}{\mddefault}{\updefault}$U\setminus \set{I_n}$}}}}
\put(3120,369){\makebox(0,0)[lb]{\smash{{\SetFigFontNFSS{9}{10.8}{\familydefault}{\mddefault}{\updefault}$U\setminus \set{I_2, \neg F_1}$}}}}
\put(5370,369){\makebox(0,0)[lb]{\smash{{\SetFigFontNFSS{9}{10.8}{\familydefault}{\mddefault}{\updefault}$U\setminus \set{I_2,\neg F_2}$}}}}
\end{picture}
}}
      \caption{Zielonka tree for 
        $\phi_1 = (\Diamond\Box F_1 )\, \vee\, (\Diamond\Box F_2)\, \vee\,
        \bigwedge_{i\leq n}(\Box\Diamond I_j)$.}
      \label{figure:Zielonka-gen}
    \end{figure}
    The number $m_{\f}^v$ for the leaf nodes is 1, and also for all the 
    Bad nodes.
    The number is hence $n$ for root.

  \item
    Consider the (partial) Zielonka tree in 
    Figure~\ref{figure:Zielonka-gen-old}.
    The leaves (not shown) are Bad nodes.
    \begin{figure}[t]
      \strut\centerline{\setlength{\unitlength}{0.00043745in}
\begingroup\makeatletter\ifx\SetFigFontNFSS\undefined%
\gdef\SetFigFontNFSS#1#2#3#4#5{%
  \reset@font\fontsize{#1}{#2pt}%
  \fontfamily{#3}\fontseries{#4}\fontshape{#5}%
  \selectfont}%
\fi\endgroup%
{\renewcommand{\dashlinestretch}{30}
\begin{picture}(12627,7014)(0,-10)
\put(4470,5007){\ellipse{1800}{810}}
\put(1052,5006){\ellipse{1800}{810}}
\put(7620,5007){\ellipse{1800}{810}}
\put(11719,5003){\ellipse{1800}{810}}
\put(11727,5006){\ellipse{1710}{720}}
\put(4471,1448){\ellipse{2520}{810}}
\path(3525,6987)(5505,6987)(5505,6312)
	(3525,6312)(3525,6987)
\path(5190,6312)(7620,5412)
\path(7430.366,5418.252)(7620.000,5412.000)(7472.044,5530.782)
\path(5505,6492)(11265,5367)
\path(11076.837,5342.617)(11265.000,5367.000)(11099.839,5460.392)
\path(4515,6312)(4515,5412)
\path(4455.000,5592.000)(4515.000,5412.000)(4575.000,5592.000)
\path(3615,6312)(1095,5412)
\path(1244.333,5529.045)(1095.000,5412.000)(1284.694,5416.036)
\path(3750,6312)(2175,5367)
\path(2298.479,5511.059)(2175.000,5367.000)(2360.218,5408.160)
\path(3975,6312)(2805,5367)
\path(2907.329,5526.777)(2805.000,5367.000)(2982.730,5433.424)
\path(4920,6312)(6495,5367)
\path(6309.782,5408.160)(6495.000,5367.000)(6371.521,5511.059)
\path(4785,6312)(5550,5412)
\path(5387.707,5510.290)(5550.000,5412.000)(5479.140,5588.008)
\dottedline{90}(3390,5682)(4245,5682)
\dottedline{90}(4695,5682)(5190,5682)
\path(3570,3657)(5370,3657)(5370,2982)
	(3570,2982)(3570,3657)
\path(6135,3657)(7935,3657)(7935,2982)
	(6135,2982)(6135,3657)
\path(1140,3657)(2940,3657)(2940,2982)
	(1140,2982)(1140,3657)
\path(9105,3657)(10905,3657)(10905,2982)
	(9105,2982)(9105,3657)
\path(3795,4737)(1950,3657)
\path(2075.032,3799.713)(1950.000,3657.000)(2135.653,3696.151)
\path(4380,4602)(4380,3657)
\path(4320.000,3837.000)(4380.000,3657.000)(4440.000,3837.000)
\path(5145,4737)(7215,3657)
\path(7027.661,3687.067)(7215.000,3657.000)(7083.169,3793.457)
\path(5280,4827)(9645,3702)
\path(9455.722,3688.822)(9645.000,3702.000)(9485.671,3805.025)
\path(3885,4692)(3075,3747)
\path(3146.587,3922.714)(3075.000,3747.000)(3237.698,3844.619)
\path(4920,4647)(6000,3702)
\path(5825.026,3775.376)(6000.000,3702.000)(5904.046,3865.685)
\dottedline{90}(3480,4062)(4245,4062)
\dottedline{90}(4560,4062)(5370,4062)
\dottedline{90}(2760,4062)(3210,4062)
\dottedline{90}(5730,4062)(6360,4062)
\dottedline{90}(5460,5682)(5820,5682)
\dottedline{90}(6135,5682)(6720,5682)
\dottedline{90}(2850,5682)(3120,5682)
\dottedline{90}(2085,5682)(2580,5682)
\path(4380,2982)(4380,1857)
\path(4320.000,2037.000)(4380.000,1857.000)(4440.000,2037.000)
\path(4245,2982)(2175,1722)
\path(2297.559,1866.842)(2175.000,1722.000)(2359.953,1764.339)
\path(4920,2982)(8475,1677)
\path(8285.349,1682.704)(8475.000,1677.000)(8326.702,1795.354)
\path(1320,2982)(690,2622)
\path(816.515,2763.400)(690.000,2622.000)(876.052,2659.210)
\path(1500,2982)(1500,2577)
\path(1440.000,2757.000)(1500.000,2577.000)(1560.000,2757.000)
\path(1725,2982)(2130,2622)
\path(1955.604,2696.741)(2130.000,2622.000)(2035.328,2786.430)
\path(6495,2982)(6225,2667)
\path(6296.587,2842.714)(6225.000,2667.000)(6387.698,2764.619)
\path(6990,2982)(6990,2667)
\path(6930.000,2847.000)(6990.000,2667.000)(7050.000,2847.000)
\path(9645,2982)(9015,2577)
\path(9133.967,2724.807)(9015.000,2577.000)(9198.858,2623.866)
\path(9870,2982)(9870,2577)
\path(9810.000,2757.000)(9870.000,2577.000)(9930.000,2757.000)
\path(10140,2982)(10725,2577)
\path(10542.853,2630.126)(10725.000,2577.000)(10611.158,2728.789)
\path(7530,2982)(7800,2667)
\path(7637.302,2764.619)(7800.000,2667.000)(7728.413,2842.714)
\dottedline{90}(3255,2307)(4290,2307)
\dottedline{90}(4470,2307)(6405,2307)
\dottedline{90}(1095,2757)(1365,2757)
\dottedline{90}(1635,2757)(1860,2757)
\dottedline{90}(6495,2802)(6810,2802)
\dottedline{90}(7170,2802)(7575,2802)
\dottedline{90}(9420,2757)(9780,2757)
\dottedline{90}(10005,2757)(10365,2757)
\path(4425,1047)(4425,57)
\path(4365.000,237.000)(4425.000,57.000)(4485.000,237.000)
\path(4110,1047)(2445,147)
\path(2574.816,285.375)(2445.000,147.000)(2631.878,179.811)
\path(4875,1047)(6630,12)
\path(6444.475,51.755)(6630.000,12.000)(6505.433,155.119)
\dottedline{90}(3165,462)(4290,462)
\dottedline{90}(4560,462)(5685,462)
\path(600,4647)(150,4287)
\path(253.075,4446.297)(150.000,4287.000)(328.038,4352.593)
\path(1005,4602)(1005,4287)
\path(945.000,4467.000)(1005.000,4287.000)(1065.000,4467.000)
\path(1545,4647)(1995,4287)
\path(1816.962,4352.593)(1995.000,4287.000)(1891.925,4446.297)
\path(7170,4647)(6990,4467)
\path(7074.853,4636.706)(6990.000,4467.000)(7159.706,4551.853)
\path(7665,4602)(7665,4332)
\path(7605.000,4512.000)(7665.000,4332.000)(7725.000,4512.000)
\dottedline{90}(465,4467)(825,4467)
\dottedline{90}(1185,4467)(1590,4467)
\dottedline{90}(7260,4512)(7530,4512)
\dottedline{90}(7845,4512)(8160,4512)
\path(8070,4647)(8340,4422)
\path(8163.309,4491.140)(8340.000,4422.000)(8240.131,4583.326)
\put(4290,6537){\makebox(0,0)[lb]{\smash{{\SetFigFontNFSS{9}{10.8}{\familydefault}{\mddefault}{\updefault}$U$}}}}
\put(645,4917){\makebox(0,0)[lb]{\smash{{\SetFigFontNFSS{9}{10.8}{\familydefault}{\mddefault}{\updefault}$U\setminus\set{F_{1}}$}}}}
\put(3840,4917){\makebox(0,0)[lb]{\smash{{\SetFigFontNFSS{9}{10.8}{\familydefault}{\mddefault}{\updefault}$U\setminus\set{F_{\alpha}}$}}}}
\put(7035,4917){\makebox(0,0)[lb]{\smash{{\SetFigFontNFSS{9}{10.8}{\familydefault}{\mddefault}{\updefault}$U\setminus\set{F_{m}}$}}}}
\put(11175,4917){\makebox(0,0)[lb]{\smash{{\SetFigFontNFSS{9}{10.8}{\familydefault}{\mddefault}{\updefault}$U\setminus\set{F}$}}}}
\put(1185,3207){\makebox(0,0)[lb]{\smash{{\SetFigFontNFSS{9}{10.8}{\familydefault}{\mddefault}{\updefault}$U\setminus\set{F_{\alpha},I_{\alpha,1}}$}}}}
\put(3615,3207){\makebox(0,0)[lb]{\smash{{\SetFigFontNFSS{9}{10.8}{\familydefault}{\mddefault}{\updefault}$U\setminus\set{F_{\alpha},I_{\alpha,i}}$}}}}
\put(6180,3207){\makebox(0,0)[lb]{\smash{{\SetFigFontNFSS{9}{10.8}{\familydefault}{\mddefault}{\updefault}$U\setminus\set{F_{\alpha},I_{\alpha,n}}$}}}}
\put(9150,3207){\makebox(0,0)[lb]{\smash{{\SetFigFontNFSS{9}{10.8}{\familydefault}{\mddefault}{\updefault}$U\setminus\set{F_{\alpha},I_{\alpha}}$}}}}
\put(3345,1362){\makebox(0,0)[lb]{\smash{{\SetFigFontNFSS{9}{10.8}{\familydefault}{\mddefault}{\updefault}$U\setminus\set{F_{\alpha},I_{\alpha,i},F_{\alpha'}}$}}}}
\put(5595,1002){\makebox(0,0)[lb]{\smash{{\SetFigFontNFSS{9}{10.8}{\familydefault}{\mddefault}{\updefault}$\alpha'\notin\set{\alpha}$}}}}
\put(15,5457){\makebox(0,0)[lb]{\smash{{\SetFigFontNFSS{11}{13.2}{\familydefault}{\mddefault}{\updefault}Good}}}}
\put(15,1272){\makebox(0,0)[lb]{\smash{{\SetFigFontNFSS{11}{13.2}{\familydefault}{\mddefault}{\updefault}Good}}}}
\put(15,6447){\makebox(0,0)[lb]{\smash{{\SetFigFontNFSS{11}{13.2}{\familydefault}{\mddefault}{\updefault}Bad}}}}
\put(15,3252){\makebox(0,0)[lb]{\smash{{\SetFigFontNFSS{11}{13.2}{\familydefault}{\mddefault}{\updefault}Bad}}}}
\end{picture}
}}
      \caption{Zielonka tree for 
        $\phi_2 =  (\Diamond\Box F )\, \vee\, 
        \bigvee_{\alpha\leq m} \left(
          \Diamond\Box F_{\alpha} \,\wedge\,
          \left( \bigwedge_{i\leq n} \Box\Diamond I_{\alpha,i}\right)\,\wedge\,
          \Box\Diamond I_{\alpha}\right)$}
      \label{figure:Zielonka-gen-old}
    \end{figure}
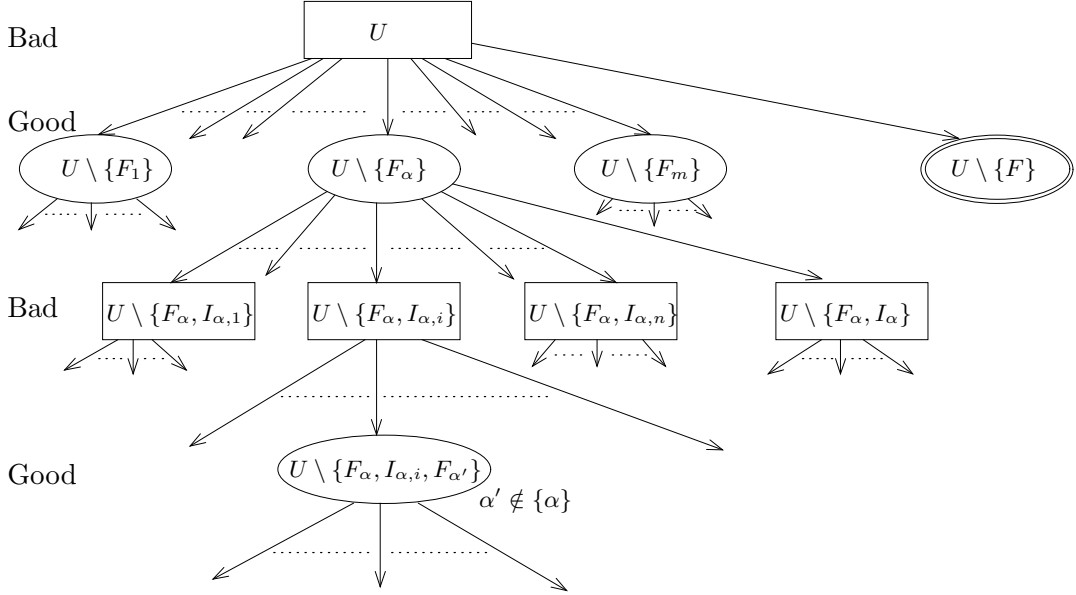
    To compute the $m_{\f}^v$ number for the root, pick an outgoing edge from
    each Bad node, and retain all edges from Good nodes.
    For such an edge choice $\mcale$, let $\leaf(\mcalz_{\f,\ap_N},\mcale)$
    denote the number of leaves reachable from the root in the resulting
    graph.
    The $m_{\f}^v$ number for the root is then 
    $\max_{\mcale}\left(\leaf(\mcalz_{\f,\ap_N},\mcale)\right)$.
    For the Zielonka tree in Figure~\ref{figure:Zielonka-gen-old}, let
    $\mcale$ be any such edge choice.
    It can be seen that each Good node in the resulting graph leads to 
    $n+1$ reachable Good nodes in the next Good level below it.
    Also, there are $m$ Good levels.
    Thus the number of  leaves  reachable from the root in the resulting
    graph for any $\mcale$ is $(n+1)^m$.
    \qed

  \end{enumerate}
\end{proof}

\begin{corollary}
\label{corollay:ZielonkaMemoryReceptive}
Let $\A$ be a timed automaton game with the clocks $C$ , 
and let $\ddot{\A}$ be the corresponding
enlarged game.
\begin{enumerate}
\item
  Let $\Phi^{\dagger}$ be as in Lemma~\ref{lemma:NewReceptiveGeneral}.
  Player~1 has a pure sure winning strategy in $\ddot{\A}$
   from $\sureu_1(\Phi^{\dagger})$
  that requires at most $(|C|+1)$ memory states.
\item
  Let $\Phi^*$ be as in Lemma~\ref{lemma:ReceptiveGeneral}.
  Player~1 has a pure sure winning strategy  in $\ddot{\A}$
  from $\sureu_1(\Phi^*)$
  that requires at most $(|C|+1)^{2^{|C|}}$ memory states.
\end{enumerate}
\end{corollary}
\begin{proof}
For both cases, we first Lemma~\ref{lemma:ZielonkaMemoryReceptiveFinite}
to the finite game structure $\ddot{\A}^{\f}$ to obtain a pure sure winning
strategy $\pi_1^{\ddot{\A}^{\f}}$ in the finite game structure $\ddot{\A}^{\f}$; 
and then we obtain a pure sure winning
strategy $\pi_1^{\ddot{\A}}$ in the game structure $\ddot{\A}$
by letting
$\pi_1^{\ddot{\A}}(\ddot{r}[0..k]) = \tuple{\Delta,a_1}$ such that
(a)~$\pi_1^{\ddot{\A}^{\f}}\left(\reg(\ddot{r}[0..k])\right) = 
\tuple{\ddot{R},a_1}$, and
(b)~$\reg(\ddot{r}[k]+\Delta)=\ddot{R}$.
Thus, the memory requirement for a player-1
winning strategy in $\ddot{\A}$  is at  as most as that for in the finite game
 $\ddot{\A}^{\f}$.
\qed
\end{proof}

\subsection{Finite Memory Receptive Strategies for
Safety Objectives}

Player~1 can ensure it stays in a set $Y$ in a receptive fashion if
it uses a receptive strategy that only plays moves to $Y$ states at each 
step.
The next theorem uses this fact to characterize safety strategies.

\begin{theorem}[Memory requirement for safety]
\label{theorem:Safety} Let $\A$ be a timed automaton game and
$\ddot{\A}$ be the corresponding enlarged game. 
Let $Y$ be a
union of regions of $\A$.
Then the following assertions hold.
\begin{enumerate}
\item
  $\sure_1^{\ddot{\A}}(\Box Y) = 
  \sureu_1^{\ddot{\A}}\left( (\Box Y) \wedge \Phi^{\sharp} \right)$, where 
  $\Phi^{\sharp} = \Phi^*$  (as defined in Lemma~\ref{lemma:ReceptiveGeneral}), 
  or $\Phi^{\sharp} = \Phi^{\dagger}$  
  (as defined in Lemma~\ref{lemma:NewReceptiveGeneral}).
\item
  Player~1 has  a 
  pure, finite-memory,  receptive, region   strategy  in $\ddot{\A}$ that is
  sure winning for
  the safety objective $\safe(Y)$ at every state in
  $ \sure_1^{\ddot{\A}}(\Box Y)$, that requires at most $(|C|+1)$ memory
  states (where $|C|$ is the number of clocks in $\A$).
\item 
  Player~1 has  a 
  pure, finite-memory,  receptive,  strategy  in $\A$ that is
  sure winning for
  the safety objective $\safe(Y)$ at every state in
  $ \sure_1^{\A}(\Box Y)$, that requires at most $(|C|+1)\cdot 2^{3\cdot|C|+1}$ 
  memory states, i.e. $\big(\,\lg(|C|+1)\, +\, 3\cdot|C|+1\,\big)$ bits of memory
  (where $|C|$ is the number of clocks in $\A$).

\end{enumerate}
\end{theorem} 

\begin{proof}
\begin{enumerate}
\item ($\Leftarrow$).
  If a state $\ddot{s}\in \sureu_1^{\ddot{\A}}(\Box Y \wedge
  \Phi^{\sharp})$, then there exists a player-1 winning strategy 
  $\pi_1$ such that
  given any player-2 strategy $\pi_2$, we have that every run $\ddot{r}$
  in $\outcomes(\ddot{s},\pi_1,\pi_2)$ satisfies both $\Box Y$ and
  $\Phi^{\sharp})$.
  Since $\Phi^{\sharp})$ is satisfies, the strategy $\pi_1$ is a receptive
  strategy by Lemmas~\ref{lemma:ReceptiveGeneral} 
  and~\ref{lemma:NewReceptiveGeneral}.
  Moreover this strategy ensures 
  that the game stays in $Y$.

  ($\Rightarrow$).
  If $\ddot{s}\notin \sureu_1^{\ddot{\A}}(\Box Y\wedge \Phi^{\sharp} )$, 
  then for every player-1 strategy $\pi_1$,
  there exists a player-2 strategy $\pi_2$ such that one of the
  resulting runs either violates $\Box Y$, or $\Phi^{\sharp}$.
  If $\Phi^{\sharp}$ is violated, then $\pi_1$ is not a receptive strategy.
  If $\Box Y$ is violated, then player~2 can switch over to a receptive
  strategy as soon as the game gets outside $Y$.
  Thus, in both cases $s\notin \sure_1^{\ddot{\A}}(\Box Y)$.

\item
  The result follows from (a)~the first part of the lemma, 
  (b)~observing that
  $\Box Y \wedge \Phi^{\sharp})$ is an $\omega$-regular objective, 
  (c)~Lemma~\ref{lemma:RegionStrategies}, and
  (d)~the first part of Corollary~\ref{corollay:ZielonkaMemoryReceptive}
  (the memory requirement to ensure 
  $\Box Y \wedge \Phi^{\sharp})$ is the same as that to ensure 
  $\Phi^{\sharp})$.
  We note that the characterization of Lemma~\ref{lemma:ReceptiveGeneral}
  for receptive strategies gives a memory bound of $(|C|+1)^{2^{|C|}}$ for
  safe receptive strategies.

\item 
  It suffices to show that in the structure $\A$, 
  player~1 needs only $(3\cdot |C| +1)$ bits to
  maintain the predicates used in
  the definition of $\ddot{\A}$ in memory.
  Then, with the help of these $(3\cdot |C|+1)$ bits, player~1 can play
  as if it is playing in $\ddot{\A}$.
  We assume that player~1 can observe the ``flow'' during a transition.
  That is, if the game moves from $s$ to $s'$ in a single game transition,
  player~1 can observe the ``intermediate'' states (arising from time passage)
  ``in between'' $s$ and $s'$.
  Then, player~1 needs only one bit for each of the predicates added to 
  $\A$ in the construcion of $\ddot{\A}$.
  These bits are updated during the flow of the transition.
  There are $(3\cdot|C|+1)$ predicates.

\qed

\end{enumerate}
\end{proof}

\subsection{Memory Requirement of  Receptive Region Strategies for
Safety Objectives}

We now show memoryless \emph{region} strategies for  safety objectives do 
not suffice (where the regions are as 
classically defined for timed automata).

\begin{example}[Memory necessity  of  winning region strategies for safety]
\label{example:MemoryRequiredRegionSafety}
Consider the timed automaton game $\A_3$ 
in Figure~\ref{figure:ExampleReceptiveRegion}.
\begin{figure}[t]
      \strut\centerline{\setlength{\unitlength}{0.00043745in}
\begingroup\makeatletter\ifx\SetFigFontNFSS\undefined%
\gdef\SetFigFontNFSS#1#2#3#4#5{%
  \reset@font\fontsize{#1}{#2pt}%
  \fontfamily{#3}\fontseries{#4}\fontshape{#5}%
  \selectfont}%
\fi\endgroup%
{\renewcommand{\dashlinestretch}{30}
\begin{picture}(10081,5829)(0,-10)
\put(998,1402){\ellipse{1980}{990}}
\put(5086,1372){\ellipse{1980}{990}}
\put(9083,1374){\ellipse{1980}{990}}
\put(5107,5142){\ellipse{1980}{990}}
\path(5123,1869)(6068,3309)(3953,3309)
	(3953,4074)(5708,4074)(4898,4659)
\path(5079.052,4602.252)(4898.000,4659.000)(5008.793,4504.971)
\path(4313,1689)(4311,1690)(4307,1693)
	(4299,1697)(4286,1704)(4270,1714)
	(4249,1725)(4224,1738)(4196,1753)
	(4166,1768)(4135,1784)(4103,1799)
	(4071,1814)(4039,1827)(4007,1840)
	(3976,1852)(3945,1863)(3914,1872)
	(3883,1881)(3851,1889)(3819,1896)
	(3785,1903)(3750,1909)(3713,1914)
	(3686,1918)(3658,1921)(3629,1924)
	(3598,1927)(3567,1930)(3535,1932)
	(3501,1934)(3467,1936)(3431,1938)
	(3395,1940)(3357,1941)(3319,1942)
	(3280,1943)(3241,1943)(3201,1943)
	(3161,1943)(3120,1942)(3080,1941)
	(3040,1940)(3000,1939)(2961,1937)
	(2922,1935)(2884,1932)(2846,1929)
	(2810,1926)(2774,1923)(2740,1919)
	(2706,1915)(2674,1910)(2643,1906)
	(2612,1901)(2583,1895)(2555,1890)
	(2528,1884)(2494,1876)(2461,1867)
	(2430,1857)(2400,1846)(2370,1835)
	(2341,1822)(2313,1807)(2284,1792)
	(2256,1775)(2227,1756)(2197,1736)
	(2168,1715)(2138,1692)(2109,1669)
	(2079,1645)(2052,1621)(2025,1598)
	(2001,1577)(1980,1558)(1963,1542)
	(1949,1529)(1928,1509)
\path(2016.966,1676.586)(1928.000,1509.000)(2099.724,1589.690)
\path(1883,1194)(1885,1192)(1889,1189)
	(1896,1183)(1907,1174)(1922,1161)
	(1942,1146)(1964,1128)(1990,1109)
	(2017,1088)(2046,1068)(2076,1047)
	(2106,1028)(2136,1009)(2166,992)
	(2196,977)(2225,962)(2255,949)
	(2285,938)(2316,927)(2348,918)
	(2381,909)(2416,901)(2453,894)
	(2480,889)(2508,885)(2538,881)
	(2568,877)(2600,873)(2632,869)
	(2666,866)(2701,863)(2737,861)
	(2774,858)(2812,856)(2851,854)
	(2891,852)(2931,851)(2971,850)
	(3012,849)(3053,849)(3094,848)
	(3135,848)(3175,849)(3215,849)
	(3255,850)(3294,851)(3332,853)
	(3369,854)(3405,856)(3440,858)
	(3474,861)(3506,863)(3538,866)
	(3568,869)(3598,872)(3626,875)
	(3653,879)(3690,884)(3725,890)
	(3758,897)(3790,904)(3821,912)
	(3851,921)(3881,930)(3910,941)
	(3940,953)(3970,966)(4000,979)
	(4030,994)(4060,1009)(4089,1025)
	(4116,1040)(4142,1055)(4164,1068)
	(4184,1079)(4199,1089)(4223,1104)
\path(4102.160,957.720)(4223.000,1104.000)(4038.560,1059.480)
\path(6023,1509)(6025,1511)(6029,1514)
	(6036,1520)(6047,1529)(6062,1542)
	(6082,1557)(6104,1575)(6130,1594)
	(6157,1615)(6186,1635)(6216,1656)
	(6246,1675)(6276,1694)(6306,1711)
	(6336,1726)(6365,1741)(6395,1754)
	(6425,1765)(6456,1776)(6488,1785)
	(6521,1794)(6556,1802)(6593,1809)
	(6620,1814)(6648,1818)(6678,1822)
	(6708,1826)(6740,1830)(6772,1834)
	(6806,1837)(6841,1840)(6877,1843)
	(6914,1845)(6952,1847)(6990,1849)
	(7030,1851)(7070,1853)(7110,1854)
	(7150,1855)(7191,1855)(7232,1856)
	(7272,1856)(7312,1856)(7351,1856)
	(7390,1855)(7429,1854)(7466,1853)
	(7502,1852)(7538,1850)(7572,1849)
	(7605,1847)(7637,1845)(7667,1842)
	(7697,1840)(7725,1837)(7752,1834)
	(7778,1831)(7816,1827)(7852,1821)
	(7885,1815)(7917,1808)(7947,1801)
	(7977,1793)(8005,1783)(8033,1773)
	(8061,1761)(8089,1749)(8117,1736)
	(8143,1722)(8169,1708)(8193,1694)
	(8215,1681)(8234,1669)(8249,1660)(8273,1644)
\path(8089.949,1693.923)(8273.000,1644.000)(8156.513,1793.769)
\path(8228,1104)(8226,1103)(8222,1100)
	(8214,1096)(8202,1089)(8186,1079)
	(8165,1068)(8142,1055)(8115,1040)
	(8086,1025)(8056,1009)(8026,994)
	(7995,979)(7965,966)(7936,953)
	(7907,941)(7878,930)(7850,921)
	(7822,912)(7794,904)(7765,897)
	(7736,890)(7705,884)(7673,879)
	(7648,875)(7622,872)(7596,868)
	(7568,865)(7539,862)(7510,860)
	(7479,857)(7447,855)(7415,853)
	(7381,852)(7347,851)(7312,850)
	(7276,849)(7240,848)(7204,848)
	(7167,849)(7131,849)(7094,850)
	(7058,851)(7022,853)(6986,855)
	(6951,857)(6917,859)(6883,862)
	(6850,865)(6818,868)(6787,872)
	(6757,876)(6728,880)(6700,884)
	(6672,889)(6646,894)(6613,901)
	(6582,908)(6552,916)(6522,924)
	(6492,934)(6463,944)(6432,955)
	(6402,968)(6370,981)(6338,996)
	(6304,1012)(6270,1029)(6235,1047)
	(6200,1066)(6165,1085)(6131,1104)
	(6099,1122)(6070,1139)(6044,1155)
	(6022,1167)(6005,1178)(5978,1194)
\path(6163.441,1153.853)(5978.000,1194.000)(6102.264,1050.618)
\path(1028,1869)(1028,1871)(1028,1874)
	(1028,1881)(1028,1891)(1028,1907)
	(1029,1927)(1029,1953)(1029,1985)
	(1030,2022)(1031,2066)(1031,2115)
	(1032,2169)(1033,2228)(1034,2291)
	(1036,2358)(1037,2427)(1039,2499)
	(1040,2572)(1042,2646)(1044,2720)
	(1046,2793)(1049,2866)(1051,2937)
	(1054,3007)(1056,3075)(1059,3140)
	(1062,3204)(1065,3265)(1068,3324)
	(1071,3380)(1075,3435)(1079,3486)
	(1083,3536)(1087,3584)(1091,3630)
	(1095,3673)(1100,3716)(1105,3756)
	(1110,3795)(1116,3833)(1122,3870)
	(1128,3905)(1134,3940)(1141,3973)
	(1148,4006)(1157,4046)(1166,4084)
	(1176,4122)(1187,4159)(1198,4195)
	(1210,4231)(1223,4267)(1236,4302)
	(1250,4336)(1265,4370)(1281,4404)
	(1298,4437)(1316,4469)(1335,4501)
	(1354,4532)(1375,4562)(1396,4592)
	(1419,4621)(1442,4649)(1466,4676)
	(1491,4702)(1517,4728)(1543,4752)
	(1570,4775)(1598,4798)(1627,4819)
	(1656,4840)(1686,4859)(1716,4877)
	(1747,4895)(1779,4911)(1811,4927)
	(1844,4942)(1877,4956)(1911,4969)
	(1946,4981)(1982,4993)(2018,5004)
	(2050,5013)(2083,5022)(2117,5030)
	(2153,5038)(2189,5046)(2227,5053)
	(2266,5060)(2307,5066)(2350,5072)
	(2395,5078)(2441,5083)(2490,5088)
	(2542,5093)(2596,5098)(2652,5103)
	(2711,5107)(2772,5111)(2836,5115)
	(2903,5118)(2971,5122)(3042,5125)
	(3115,5128)(3189,5131)(3264,5134)
	(3339,5136)(3414,5139)(3488,5141)
	(3561,5143)(3631,5145)(3698,5146)
	(3761,5148)(3819,5149)(3872,5150)
	(3919,5151)(3961,5152)(3995,5153)
	(4024,5153)(4046,5153)(4063,5154)(4088,5154)
\path(3908.000,5094.000)(4088.000,5154.000)(3908.000,5214.000)
\path(9173,1869)(9173,1871)(9173,1874)
	(9173,1881)(9173,1891)(9172,1907)
	(9172,1927)(9172,1953)(9171,1985)
	(9171,2022)(9170,2066)(9169,2115)
	(9168,2169)(9166,2228)(9165,2292)
	(9163,2358)(9161,2428)(9160,2500)
	(9157,2573)(9155,2647)(9153,2721)
	(9150,2794)(9147,2867)(9144,2939)
	(9141,3008)(9138,3076)(9134,3142)
	(9131,3206)(9127,3267)(9123,3326)
	(9119,3383)(9114,3437)(9110,3490)
	(9105,3540)(9100,3588)(9094,3633)
	(9089,3678)(9083,3720)(9077,3761)
	(9070,3800)(9063,3838)(9056,3875)
	(9048,3911)(9040,3946)(9032,3981)
	(9023,4014)(9012,4054)(9000,4093)
	(8988,4131)(8975,4169)(8961,4206)
	(8946,4243)(8931,4279)(8914,4315)
	(8897,4350)(8879,4385)(8860,4419)
	(8839,4453)(8818,4486)(8796,4519)
	(8774,4551)(8750,4582)(8725,4613)
	(8700,4642)(8673,4671)(8646,4700)
	(8618,4727)(8589,4753)(8560,4778)
	(8530,4803)(8500,4826)(8469,4848)
	(8437,4869)(8405,4890)(8373,4909)
	(8340,4927)(8307,4944)(8274,4961)
	(8240,4976)(8205,4991)(8171,5004)
	(8135,5017)(8100,5030)(8063,5042)
	(8031,5051)(7998,5060)(7964,5069)
	(7930,5077)(7894,5085)(7857,5093)
	(7819,5100)(7780,5107)(7739,5113)
	(7696,5119)(7651,5125)(7605,5130)
	(7556,5135)(7505,5140)(7452,5145)
	(7397,5149)(7339,5154)(7279,5158)
	(7216,5162)(7152,5165)(7086,5169)
	(7018,5172)(6949,5175)(6879,5178)
	(6809,5181)(6739,5183)(6670,5185)
	(6602,5187)(6537,5189)(6475,5191)
	(6416,5193)(6362,5194)(6313,5195)
	(6269,5196)(6231,5197)(6199,5198)
	(6172,5198)(6152,5198)(6136,5199)(6113,5199)
\path(6293.000,5259.000)(6113.000,5199.000)(6293.000,5139.000)
\put(4808,1509){\makebox(0,0)[lb]{\smash{{\SetFigFontNFSS{11}{13.2}{\familydefault}{\mddefault}{\updefault}$l_0$}}}}
\put(758,1509){\makebox(0,0)[lb]{\smash{{\SetFigFontNFSS{11}{13.2}{\familydefault}{\mddefault}{\updefault}$l_1$}}}}
\put(8813,1509){\makebox(0,0)[lb]{\smash{{\SetFigFontNFSS{11}{13.2}{\familydefault}{\mddefault}{\updefault}$l_2$}}}}
\put(5798,114){\makebox(0,0)[lb]{\smash{{\SetFigFontNFSS{11}{13.2}{\familydefault}{\mddefault}{\updefault}$1>y>0 \rightarrow x:=0$}}}}
\put(1478,159){\makebox(0,0)[lb]{\smash{{\SetFigFontNFSS{11}{13.2}{\familydefault}{\mddefault}{\updefault}$1>y>0 \rightarrow x:=0$}}}}
\put(6923,474){\makebox(0,0)[lb]{\smash{{\SetFigFontNFSS{11}{13.2}{\familydefault}{\mddefault}{\updefault}$a_2^1$}}}}
\put(2873,474){\makebox(0,0)[lb]{\smash{{\SetFigFontNFSS{11}{13.2}{\familydefault}{\mddefault}{\updefault}$a_2^0$}}}}
\put(4673,1149){\makebox(0,0)[lb]{\smash{{\SetFigFontNFSS{11}{13.2}{\familydefault}{\mddefault}{\updefault}$2>y$}}}}
\put(8633,1104){\makebox(0,0)[lb]{\smash{{\SetFigFontNFSS{11}{13.2}{\familydefault}{\mddefault}{\updefault}$2>y$}}}}
\put(533,1149){\makebox(0,0)[lb]{\smash{{\SetFigFontNFSS{11}{13.2}{\familydefault}{\mddefault}{\updefault}$2>y$}}}}
\put(6698,2004){\makebox(0,0)[lb]{\smash{{\SetFigFontNFSS{11}{13.2}{\familydefault}{\mddefault}{\updefault}$1>y>x$}}}}
\put(7013,2274){\makebox(0,0)[lb]{\smash{{\SetFigFontNFSS{11}{13.2}{\familydefault}{\mddefault}{\updefault}$a_1^1$}}}}
\put(1658,2094){\makebox(0,0)[lb]{\smash{{\SetFigFontNFSS{11}{13.2}{\familydefault}{\mddefault}{\updefault}$x=0 \rightarrow y:=0$}}}}
\put(2603,2364){\makebox(0,0)[lb]{\smash{{\SetFigFontNFSS{11}{13.2}{\familydefault}{\mddefault}{\updefault}$a_1^0$}}}}
\put(4898,5064){\makebox(0,0)[lb]{\smash{{\SetFigFontNFSS{11}{13.2}{\familydefault}{\mddefault}{\updefault}$\bad$}}}}
\put(4178,3399){\makebox(0,0)[lb]{\smash{{\SetFigFontNFSS{11}{13.2}{\familydefault}{\mddefault}{\updefault}$y>1$}}}}
\put(4178,3759){\makebox(0,0)[lb]{\smash{{\SetFigFontNFSS{11}{13.2}{\familydefault}{\mddefault}{\updefault}$a_1^3, a_2^3$}}}}
\put(848,5064){\makebox(0,0)[lb]{\smash{{\SetFigFontNFSS{11}{13.2}{\familydefault}{\mddefault}{\updefault}$y>1$}}}}
\put(8093,5154){\makebox(0,0)[lb]{\smash{{\SetFigFontNFSS{11}{13.2}{\familydefault}{\mddefault}{\updefault}$y>1$}}}}
\put(848,5469){\makebox(0,0)[lb]{\smash{{\SetFigFontNFSS{11}{13.2}{\familydefault}{\mddefault}{\updefault}$a_1^4, a_2^4$}}}}
\put(8093,5559){\makebox(0,0)[lb]{\smash{{\SetFigFontNFSS{11}{13.2}{\familydefault}{\mddefault}{\updefault}$a_1^5, a_2^5$}}}}
\end{picture}
}}
      \caption{A time automaton game $\A_3$ where player-1 does not have 
        receptive region strategies for the safety objective.}
      \label{figure:ExampleReceptiveRegion}
    \end{figure}
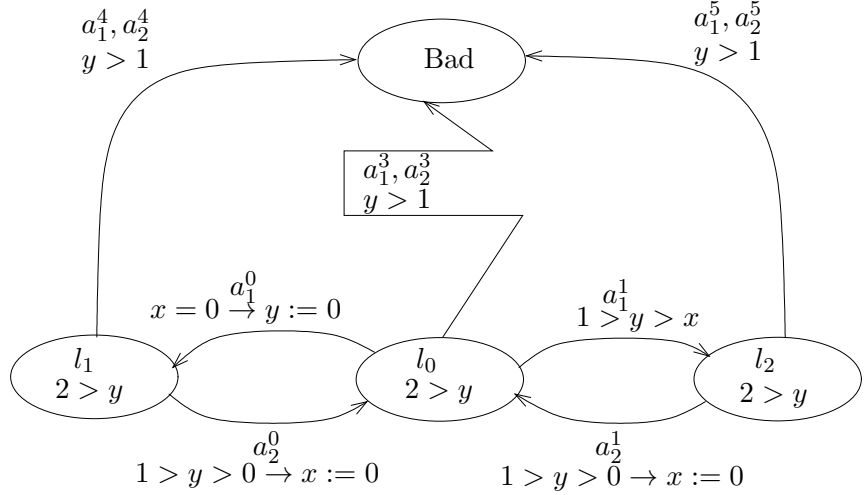
The edges $a_1^j$ are player-1 edges and $a_2^j$ player-2 edges.
The safety objective of player-1 is to avoid the location ``$\bad$''.
It is clear that to avoid the bad location, 
player-1 must ensure that the game keeps cycling around the locations 
$l_0, l_1, l_2$, and that the clock value of $y$ never exceeds 1.
Cycling around only in  $l_0, l_1$ cannot be ensured by a receptive player-1
strategy as player~2 can take smaller and smaller time steps to take the 
$a_2^0$ transition.
Cycling around only in  $l_0, l_2$ also cannot be ensured by a receptive 
player-1 strategy as the clock value of would always need to stay below 1 
without being reset, implying that more than 1 time unit does not pass.
Thus, any receptive player-1 strategy which avoids the bad location must cycle
infinitely often between $l_0, l_1$, and also between $l_0, l_2$.

Suppose a player-1 \emph{memoryless} region strategy $\pi_1^*$ exists for 
avoiding the 
bad location,
starting from a state in the region $R= \tuple{l_0, x=0 \wedge 0<y<1}$.
Suppose $\pi_1^*$ always proposes the transition $a_1^0$ from the region 
$R_1$.
Then, player~2 can take the $a_2^0$ transitions with smaller and smaller time
delays and ensure that the region is $R$ after each $a_2^0$ transition.
This will make time converge, and player~1 will not be blameless, thus $\pi_1^*$
is not a receptive strategy.
Suppose $\pi_1^*$ always proposes the transition $a_1^1$ from the region 
$R_1$ (or proposes a non-zero time delay move, which has the 
equivalent effect of disabling the $a_1^0$ transition).
In this case,  player~2 can take the  $a_2^1$ transition to again ensure that 
the region is $R$ after the  $a_2^1$ transition.
This will result in the situation where
 the $l_0, l_2$ cycle is always taken,  time
is not  divergent, and player~1 is not blameless; thus $\pi_1^*$
is again not a receptive strategy.

We now demonstrate that a finite-memory (actually memoryless in this case)
receptive player-1 strategy $\pi_1^{\dagger}$ exists from states in the region 
$R= \tuple{l_0, x=0 \wedge 0<y<1}$ for avoiding the bad location.
If the current  state 
is in the region $R$ with the clock value of $y$ being less than
$1/2$, then player~1 proposes  the $a_1^1$ transition with a delay which
will make make clock $y$ have a value greater than $1/2$.
If the current  state 
is in the region $R$ with the clock value of $y$ being  greater than or
equal to $1/2$, then player~1 proposes to take the $a_1^0$ transition 
(immediately).
This strategy ensures that against any player-2 receptive strategy:
(1)~the
game will  cycle
infinitely often between $l_0, l_1$, and also between $l_0, l_2$, and
(2)~the clock $y$ will be at least $1/2$ infinitely often, and also be reset
infinitely often, giving us time divergence.
Thus, $\pi_1^{\dagger}$ is a receptive memoryless player-1 winning strategy.

Finally, we demonstrate a player-1 finite-memory receptive \emph{region} 
strategy $\pi_1^{\ddagger}$ for avoiding the bad location, starting from
a state in
the region $R= \tuple{l_0, x=0 \wedge 0<y<1}$.
The strategy acts as follows when at region $R$.
If the previous cycle was to $l_1$, 
the strategy $\pi_1^{\ddagger}$ proposes to take the edge $a_1^1$ 
with a delay which
will make make clock $y$ have a value greater than $1/2$.
If the  previous cycle was to $l_2$, the strategy $\pi_1^{\ddagger}$ proposes to 
take the edge $a_1^0$ (immediately).
It can be verified that the strategy $\pi_1^{\ddagger}$ requires only one memory
state, and is a player-1 winning  receptive region  strategy.
\qed
\end{example}

\begin{theorem}[Memory necessity  of  winning region strategies for safety]
\label{theorem:RegionMemorySafety} There is a timed automaton game $\A$,
a union of regions $Y$  of $\A$, and a state $s$ such that player~1 does not 
have a  winning memoryless receptive region strategy  from $s$,
but has a winning receptive region strategy  from $s$ 
that requires at most $(|C|+1)$ memory
states (where $|C|$ is the number of clocks in $\A$), for the objective
of staying in the  set $Y$.
\end{theorem}
\begin{proof}
Example~\ref{example:MemoryRequiredRegionSafety} presents such a timed
automaton game.
The memory bound follows from Theorem~\ref{theorem:Safety}.
\qed
\end{proof}

\bibliographystyle{alpha}
\bibliography{timed}

\end{document}